\documentclass[aps, amsmath, amssymb, notitlepage, twocolumn, superscriptaddress, longbibliography, nofootinbib, floatfix]{revtex4-2}

\usepackage[dvipsnames, svgnames, x11names]{xcolor}
\definecolor{darkred}{rgb}{0.5,0,0}
\definecolor{darkblue}{rgb}{0,0,0.5}
\definecolor{darkgreen}{rgb}{0,0.5,0}
\usepackage[normalem]{ulem}

\usepackage[pdftex, colorlinks=true, linkcolor=darkblue, citecolor=blue, urlcolor=darkred]{hyperref}

\usepackage{amsmath, amsfonts, amssymb, amsthm, dsfont, bm, bbm, breqn, mathrsfs}
\usepackage{bm}  
\usepackage{physics}

\usepackage{wrapfig}
\usepackage{graphicx}
\graphicspath{{figs/}{}}
\usepackage[caption=false]{subfig}  
\usepackage{tikz}

\usepackage{verbatim, soul, wasysym, enumerate, multirow,footmisc}
\usepackage[capitalise]{cleveref} 

\usepackage{aliascnt}
\newtheorem{theorem}{Theorem}
\newenvironment{thmprime}[1]
{\renewcommand{\thetheorem}{\ref{#1}'}
	\addtocounter{theorem}{-1}
	\begin{theorem}}
	{\end{theorem}}
\newaliascnt{corollary}{theorem}
\newtheorem{corollary}[corollary]{Corollary}
\aliascntresetthe{corollary}
\crefname{corollary}{Corollary}{Corollaries}
\Crefname{corollary}{Corollary}{Corollaries}
\newtheorem{prop}[theorem]{Proposition} 
\newtheorem{definition}{Definition}

\newaliascnt{lemma}{theorem}
\newtheorem{lemma}[lemma]{Lemma}
\aliascntresetthe{lemma}

\newtheorem{fact}[theorem]{Fact} 

\newcommand{\Eq}[1]{Eq.~(\ref{#1})}
\newcommand{\Ineq}[1]{Ineq.~(\ref{#1})}
\newcommand{\Def}[1]{Definition~\ref{#1}}
\newcommand{\Lem}[1]{Lemma~\ref{#1}}
\newcommand{\Thm}[1]{Theorem~\ref{#1}}
\newcommand{\Sec}[1]{Sec.~\ref{#1}}

\newcommand{\cRef}[1]{Ref.~\cite{#1}}
\newcommand{\cRefs}[1]{Refs.~\cite{#1}}
\newcommand{\Fig}[1]{Fig.~\ref{#1}}

\newcommand{\ignore}[1]{}

\usepackage{mdframed}
\usepackage[textwidth=1.6cm]{todonotes}
\newcommand{\bigO}[1]{\mathcal{O}(#1)}

\newcommand{\EqDef}{\stackrel{\mathrm{def}}{=}}

\DeclareMathOperator*{\dist}{dist}

\DeclareMathOperator{\poly}{poly}
\DeclareMathOperator{\supp}{supp}
\DeclareMathOperator{\gap}{gap}
\DeclareMathOperator{\EX}{\mathbb{E}}

\newcommand{\prob}[1]{\mathbb{P}\left(#1\right)}

 \newcommand{\LC}[1]{\mathrm{LC}( #1)}

\newcommand{\PTr}[2]{\mathrm{Tr}_{#1}\left[ #2\right]}

\newcommand {\br} [1] {\ensuremath{ \left( #1 \right) }}
\newcommand {\minusspace} {\: \! \!}
\newcommand {\fn} [2] {\ensuremath{ #1 \minusspace \br{ #2 }}}
\newcommand{\mutinf}[2]{\fn{\mathrm{I}}{#1  :  #2}}
 \newcommand{\ent}[1]{\fn{\mathrm{S}}{#1}}
 \newcommand{\av}[1]{{ \langle {#1} \rangle }}
\newcommand{\fid}[2]{{\mathrm{F}\left({#1} , {#2} \right)}}
\newcommand{\bur}[2]{{\mathrm{D_{B}}\left({#1} , {#2} \right)}}
 \newcommand{\condinf}[3]{{\mathrm{I}}{\big(#1  :  #2 | #3\big)}}
 \newcommand{\drel}[2]{\fn{\mathrm{D}}{#1 \middle \| #2}}

\newcommand{\BBC}{\mathbb{C}}
\newcommand{\BBH}{\mathbb{H}}

\newcommand{\LH}{\mathrm L(\BBH)}

    \newcommand{\Id}{\mathbbm{1}}
\newcommand{\mcC}{\mathcal{C}}
\newcommand{\mcD}{\mathcal{D}}

 \newcommand{\mcF}{\mathcal{F}}
\newcommand{\mcH}{\mathcal{H}}
\newcommand{\mcL}{\mathcal{L}}
 \newcommand{\mcM}{\mathcal{M}}
 \newcommand{\mcP}{\mathcal{P}}
\newcommand{\mcR}{\mathcal{R}}
\newcommand{\mcS}{\mathcal{S}}
\newcommand{\mcT}{\mathcal{T}}

\newcommand\tab[1][1cm]{\hspace{#1}}

\begin{document}

\title{A Unified Framework for Locally Stable Phases}

\author{Zhi Li}\thanks{Co-first author}\email{zli@ibm.com}\affiliation{\IBM}
\author{Raz Firanko}\thanks{Co-first author}\email{rfiranko@pitp.ca}\affiliation{\PI}\affiliation{\IQC}
\author{Timothy H. Hsieh}\email{thsieh@pitp.ca}\affiliation{\PI}
\newcommand*{\IBM}{IBM Quantum, IBM Research}
\newcommand*{\PI}{Perimeter Institute for Theoretical Physics, Waterloo, Ontario N2L 2Y5, Canada}
\newcommand*{\IQC}{Institute for Quantum Computing (IQC), University of Waterloo, Ontario, Canada}

\begin{abstract}
We propose a unifying framework for characterizing pure and mixed-state phases of matter across equilibrium, non-equilibrium, and metastable regimes.  
We introduce the concept of locally stable states, defined by the operational property that any local operation (including post-selection) can be reversed by a local channel. We prove that local stability is equivalent to a state being short-range correlated—defined by the decay of both correlations and conditional mutual information.  We demonstrate that these properties are invariant under locally reversible channels, thus defining locally stable phases.  Furthermore, we prove that local stability implies both the decay of a family of nonlinear correlators, including the fidelity correlator, and the decay of correlations in the canonical purification, thus bridging the gap between mixed and pure states.  Along the way, we establish two results which may be of independent interest: we show that post-selection on locally stable / short-range correlated states can be implemented via local channels and that quantum Markov chains can be characterized by the local computability of nonlinear observables.

\end{abstract}
\maketitle

\tableofcontents 

\section{Introduction}

\begin{figure*}[!ht]
    \centering
    \includegraphics[width=\textwidth]{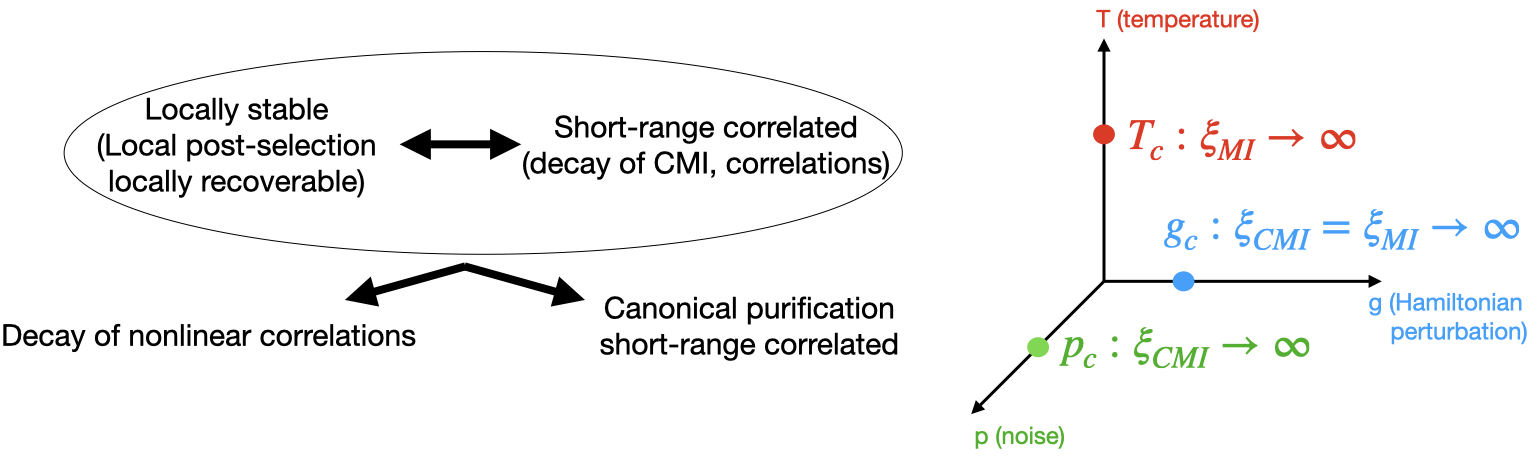}
    \caption{(left) Summary of main results, which provides a unified framework for the (right) generalized landscape of phases including equilibrium thermal (red) and quantum (blue), and non-equilibrium (green) phase transitions driven by noise, for example. The former can be diagnosed by conventional correlation lengths, for example given by mutual information (MI), and the latter can be diagnosed by conditional mutual information (CMI) and its associated length scale, the Markov length. }
    \label{fig:intro-figure}
\end{figure*}

The concept of a phase of matter is central to condensed matter physics and intimately tied to the structure of spatial correlations.  In equilibrium, Gibbs states of local Hamiltonians at finite and zero temperature (ground states) constitute phases of matter whose stability is typically dictated by decay of correlations of local operators.  In recent years, advances in the controllability of open quantum systems and error correction have motivated the study of non-equilibrium states arising from local decoherence ~\cite{bao2023mixedstatetopologicalordererrorfield,Fan_2024,WangLi2024AnomalyOpenQuantumSystems,coser2019classification, SangZouHsieh2024MixedStatePhases,PhysRevX.14.041031,
  ChenGrover2024SeparabilityTransitions,
  ChenGrover2024UnconventionalMixedStateTransition,ref:Sang2025stability,ref:sang2025mixedstatephases,
  SohalPrem2025NoisyIntrinsicMixedOrder,
  WangWuWang2025IntrinsicMixedStateTopo,
  EllisonCheng2025TowardClassificationMixedState,ZhangXuZhangXuBiLuo2025, ma2023average,Lee2025symmetryprotected, ref:swssb1,ref:lessa2025fidelity, ref:Sala2024Renyi,wang2025decoherenceinducedselfdualcriticalitytopological,liu2025coherenterrorinducedphase, zhang2025stabilitymixedstatephasesweak,hauser2026strongtoweaksymmetrybreakingopen}.  For such generally mixed states, 
the behavior of correlation functions alone is insufficient to fully characterize phase transitions, and correlations nonlinear in the state are essential.  In particular, the decay of conditional mutual information (CMI) has been identified as an important criterion for the stability of a phase \cite{ref:Sang2025stability}, in addition to decay of correlations.  Moreover, phenomena such as strong-to-weak symmetry breaking (SWSSB) are characterized exclusively by nonlinear correlations 
such as CMI and Renyi/fidelity correlators~\cite{ref:swssb1,ref:lessa2025fidelity,ref:Sala2024Renyi,ref:Weinstein2025CPcor,ref:Bloch2026swssb}. 
However, a general understanding of how all these correlation
measures are related (and furthermore, their relation to pure
state linear correlations) has been lacking, despite important progress when there is a global symmetry \cite{ref:lessa2025fidelity,ref:Weinstein2025CPcor}.

To add to the mix, metastability is another important phenomenon beyond equilibrium, with examples ranging from false vacua \cite{Devoto_2022} to transient order in driven quantum materials \cite{doi:10.1126/science.1197294,Ilyas2024,Li2025} and self-correcting classical and quantum memories \cite{bergamaschi2025rapidmixinggibbsstates}.  Recent works have defined metastable states in terms of both dynamical properties (approximate stationarity~\cite{ref:chen2025metastability}) and a notion of local energy gap~\cite{ref:Lucas2025metstablity}.  These works have highlighted the importance of considering the effect of local operations on the state. 
How is metastability, defined dynamically or operationally, related to the above stable phases of matter, thus far defined from their correlation structure?

In this work, we provide a framework that unifies pure and mixed-state phases, both in and out of equilibrium or metastable.  We refer to this unified class of states as locally stable, and we define a state to have this property if {\it any} local operation (channel or Kraus operator/post-selection) can be reversed by a local channel.  We prove that this operational property is  equivalent to the state being short-range correlated, which we define as decay of {\it both} CMI and correlations.  We show that local stability is invariant under locally reversible channels, hence defining a locally stable {\it phase}.
From a dynamical perspective, we show that in certain settings, the Lindbladian giving rise to a steady state can be used to recover against local perturbations to the state, and we identify the Lindbladian gap as a key ingredient.
We then establish that local stability implies the decay of a family of nonlinear correlators including the fidelity correlator, without assuming any symmetry.  Finally, we show that local stability of a state implies decay of correlations in its canonical purification, thus further bridging mixed and pure state phases.

As a byproduct, we obtain two results that may be of independent interest.  We show that post-selection of a local measurement on any short-range correlated state can be physically implemented by a local channel.  Moreover, we show that a state being a quantum Markov chain is equivalent to the local computability of nonlinear observables.

\section{Short-range Correlated (SRC) States}\label{sec:gapped}
Throughout this work we study families of states $\rho$ defined 
on $D$-dimensional lattices $\Lambda$ of sites, each with a $d$-dimensional Hilbert space.
For further settings and notation, we refer the reader to Appendix~\ref{app:settings}.
\begin{definition}[Short-range correlation]
\label{def:gapped}
We say that a state $\rho$ is short-range correlated (SRC) if for any contiguous region
$A\subseteq \Lambda$ and a buffering annulus $B$ separating $A$ from the complement $C$,
and for any observables $O_A\in L(\BBH_A)$ and $O_C\in L(\BBH_C)$ such that $\norm{O_A}_\infty=\norm{O_C}_\infty=1$, we have
\begin{align}
    \label{eq:DOC}
    \left|\langle O_A O_C \rangle_\rho 
      - \langle O_A \rangle_\rho \langle O_C \rangle_\rho\right| 
    \le f(A)  e^{-\operatorname{dist}(A,C)/\eta}, \\
    \label{def:CMI-decay}
    I(A{:}C|B)_\rho 
      \le g(A)\, e^{-\operatorname{dist}(A,C)/\zeta}.
\end{align}
for some constants $\eta,\zeta$ and  functions $f$ and $g$ no larger than exponentials.
\end{definition}

 This definition is a generalization of the notion of short-range correlations commonly used for thermal or gapped ground states, for which correlations between distant observables decay exponentially with their separation~\cite{ref:hastings2006spectral}.
Property \eqref{eq:DOC} is the usual notion of decay of operator correlations, where $\eta$ is known as the \emph{correlation length}.
The additional property \eqref{def:CMI-decay} is the approximate Markov chain property, with $\zeta$  dubbed the \emph{Markov length}. While quantum and thermal phase transitions typically involve a diverging correlation length, noise-induced mixed-state transitions involve a diverging Markov length ~\cite{ref:Sang2025stability,ref:sang2025mixedstatephases} (see Fig. \ref{fig:intro-figure}).  As finite Markov length has been established in Gibbs states of local Hamiltonians ~\cite{ref:Poulin2012markov,ref:chen2025markovian,ref:Kuwahara2025CMI}, the diverging Markov length is a signature of intrinsically non-equilibrium phenomena in quantum and classical~\cite{chen2026localreversibilitydivergentmarkov} systems.

While~\Eq{eq:DOC} is formulated in terms of connected operator 
correlations, one could equivalently quantify correlations using other conventional measures, such as 
the mutual information $\mutinf{A}{C}_\rho$ or the trace distance 
$\norm{\rho_{AC}-\rho_A\otimes\rho_C}_1$. 
In our work we focus on the case where the region $A$ is local $(|A|=O(1))$, hence all these correlation measures give equivalent criteria for exponential clustering.

More quantitatively, defining
\begin{equation}
\mcC(A,C)
\EqDef\sup_{\|O_A\|_\infty=\|O_C\|_\infty=1}
\bigl|
\av{ O_A O_C}_\rho
-
\av{O_A}_\rho\av{ O_C}_\rho
\bigr|,
\end{equation}
one has the bounds
\begin{equation} \label{eq:first-equivalence}
\mcC(A,C) \le  \norm{\rho_{AC}-\rho_{A}\otimes \rho_{C}}_1
        \lesssim \sqrt{\mutinf{A}{C}_\rho},
\end{equation}
as well as
\begin{equation} \label{eq:second-equivlence} 
\mutinf A C_\rho 
\lesssim
\log(d_A)\cdot \varepsilon + \varepsilon \log(1/\varepsilon )
\end{equation}
for $\varepsilon=\norm{\rho_{AC}-\rho_A\otimes\rho_C}_1$,
and
\begin{equation}\label{eq:third-equivlence}
    \norm{\rho_{AC}-\rho_A\otimes\rho_C}_1 
\leq
\sqrt{2} d_A\, \mcC(A,C). 
\end{equation}
Both inequalities in \Eq{eq:first-equivalence} are straightforward applications of
H\"older inequality with $(\rho_{AC}-\rho_A\otimes \rho_C)$ and $O_AO_C$, and Pinsker’s bound, respectively.
Inequality~\eqref{eq:second-equivlence} is a standard consequence of Alicki-Fannes-Winter continuity argument \cite{Alicki2004cont,winter2016tight}
and  inequality~\eqref{eq:third-equivlence} is proved in Ref.~\cite{li2026global} (see also \cite{correa2022maximal}).
To summarize, for a local region $A$ with $|A|=O(1)$, all correlation measures discussed above give equivalent criteria for exponential clustering and can be used interchangeably in the definition of SRC. 
For more discussions, see the end of \cref{sec:equiv}.


\section{Local Stability}
\label{sec:recovery}
We identify an operational characterization for SRC states.
\begin{definition}[Local Stability] \label{def:local-rec}
    A state $\rho$ is locally stable if for any local region $A$ and any buffering annulus region $B$ of radius $r>r_0$,
    there exists a recovery channel $\mcR_{AB}$ on $AB$ such that, for any quantum channel $\Phi_A$ on $A$ with any Kraus decomposition $\Phi _A(\cdot)=\sum F_k(\cdot)F_k^\dagger$, 
    \begin{align} \label{eq:recovery}
        \sum_k p_k \norm{\mcR_{AB}\Big(\frac{1}{p_k}F_k\rho F_k^\dagger\Big)-\rho}_1\le h(A) \cdot  e^{-r/\xi} .
    \end{align}
    Here $h(\cdot)$ is some function of $|A|$, $\xi$ and $r_0$ are constants, and $p_k=\Tr(F_k\rho F_k^\dagger)$.
    The function $h$ and the constants $\xi,r_0$ are independent of the system size.
\end{definition}


To understand the condition \cref{eq:recovery}, we may view $\Phi$ as a measurement device:
with probability $p_k$ it returns the post-measurement state $\frac{1}{p_k}F_k\rho F_k^\dagger$.
The condition (\ref{eq:recovery}) states that, on average, $\mcR$ acts as a universal (independent of the measurement device and outcome) recovery channel for all post-measurement states.

A straightforward consequence of the triangle inequality is that local stability implies that $\mcR$ recovers the action of the overall channel $\Phi _A$:
\begin{align} \label{eq:channel-recovery}
    \norm{\mcR\circ\Phi_A(\rho)-\rho}_1 \le h(A) \cdot  e^{-r/\xi} .
\end{align} 
However, individual trajectories  $\frac{1}{p_k}F_k\rho F_k^\dagger$ need not be obtainable from $\rho$ via a quantum channel.
Therefore, recovering trajectories is a stronger requirement than recovering only the output of the overall channel. 
In this sense, $\mcR$ can be viewed as a strengthened analog of a Petz recovery map.  

We note that a state that admits an $\epsilon$ error recovery of trajectories is called an ``$\epsilon$-strongly Markov state'' in \cRef{ref:chen2025metastability}, though the relation between local support of recovery and SRC properties of the state was not considered.  
We also note that, for systems subject to a global symmetry, the {\it lack} of local stability against postselection was dubbed ``weak spontaneity'' in Ref.~\cite{ref:lessa2025fidelity}; as an example, the authors noted that the coherent and incoherent cat states $\frac{1}{\sqrt{2}}(|0...0\rangle+|1...1\rangle)$, $\frac{1}{2}(|0...0\rangle\langle0...0|+|1...1\rangle\langle1...1|)$ are unstable (weakly spontaneous) against $Z$ measurement.
   
\subsection{Equivalence between SRC and local stability}\label{sec:equiv}
The main result of this section establishes an equivalence between local stability and short-range correlations, modulo some technical details specified in \cref{thm:LR->gap,thm:gap->LR}.


\begin{theorem}[Local stability $\Rightarrow$ SRC] \label{thm:LR->gap}
    Let $\rho$ be a locally stable state (\cref{def:local-rec}) for some $h$ and $\xi$.
    Then for a region $A$, a buffer $B$ and the exterior $C$, any operators $O_A$ and $O_C$ such that $\norm{O_A}_\infty=\norm{O_C}_\infty=1$, we have
    \begin{align} 
        \label{eq:LR->DOC} 
        |\av{O_AO_C}_\rho-\av{O_A}_\rho\av{O_C}_\rho|  \lesssim 
        h(A) e^{-\dist(A,C)/\xi} ,\\
        \label{eq:LR->CMI} \condinf{A}{C}{B} \lesssim
        (|A|+|B|)h(A) e^{-\dist(A,C)/2\xi} .
    \end{align}
\end{theorem}
Note that although there is a $|B|$ dependence in the second bound, the overall bound is still exponentially small in $r=\dist(A,C)$, since $|B|=\poly(r)$.
\begin{proof}
    We first show that \cref{eq:recovery} implies \cref{eq:LR->DOC}. 
    We may assume without loss of generality that $O_A$ is positive semi-definite (PSD), since a general $O_A$ can be decomposed into Hermitian and anti-Hermitian parts, and both parts can be reduced to the PSD case by shifting a scalar matrix.
    Consider the recovery channel $\mcR$ given by \cref{eq:recovery}.
    Note that $\mcR$ only applies on $AB$, therefore $\mcR^*(O_C)=O_C$, implying
    \begin{align}\label{eq:OAOC}\begin{split}
        \Tr[O_A O_C \rho]
        &=\Tr[ \mcR^*(O_C) \sqrt{O_A} \rho\sqrt{O_A}]\\
        & =  \Tr[ O_C \mcR\big( \sqrt{O_A} \rho\sqrt{O_A}\big)].
    \end{split}
    \end{align}
    Note that as $O_A$ is normalized, the map $\rho\mapsto \sqrt{O_A}\rho\sqrt{O_A}$ is completely positive and trace non-increasing, and thus it can be completed to a quantum channel on $A$ (that is, $\sqrt{O_A}$ is one of the Kraus operators). 
    \cref{eq:recovery} then implies that
    \begin{align}
         \norm{\mcR\big( \sqrt{O_A} \rho\sqrt{O_A}\big)- \Tr(\rho O_A)\rho}_1\le h(A) e^{-r/\xi} .
    \end{align}
    Plugging it into \cref{eq:OAOC} and using $\norm{O_C}= 1$, we find
    \begin{align}
        |\av{O_AO_C}_\rho-\av{O_A}_\rho \av{O_C}_\rho|
        \le h(A) e^{-r/\xi}.
    \end{align}

    Next, let us show that \cref{eq:recovery} implies \cref{eq:LR->CMI}. 
    Consider the special case where $\Phi$ is the depolarizing channel, then \cref{eq:channel-recovery} implies that:
    \begin{equation}
            \norm{\mcR\circ\mcT(\rho_{BC})-\rho_{ABC}}_1 \le h(A) \cdot  e^{-r/\xi},
    \end{equation}
    where $\mcT:B\rightarrow AB$ is a ``stitching" map $\mcT(\rho_B)=\frac{\Id_A}{d_A}\otimes \rho_B$.
    Namely, the channel $\mcP\EqDef \mcR\circ \mcT:B\rightarrow AB$ can approximately recover $\rho_{ABC}$ from $\rho_{BC}$, up to an error $\delta\le h(A)e^{-r/\xi}$.
    \cref{eq:LR->CMI} then follows from the fact that good recovery implies small CMI (see Fact \ref{fact:FoR->CMI} in Appendix~\ref{app:settings}).
\end{proof}

\begin{figure}
    \centering
\begin{tikzpicture}[scale=.8]

\def\rA{0.5}
\def\rBOne{1.35}
\def\rBTwo{2.70}

\definecolor{Acol}{RGB}{120,140,170}    
\definecolor{BOnecol}{RGB}{132,168,160} 
\definecolor{BTwocol}{RGB}{214,188,182} 
\definecolor{Outline}{RGB}{50,65,95}

\foreach \x in {-4,-3,...,4} {
  \foreach \y in {-4,-3,...,4} {
    \fill[black] (\x,\y) circle (0.06);
  }
}

\fill[BTwocol, opacity=0.6] (0,0) circle[radius=\rBTwo];
\fill[BOnecol, opacity=0.6] (0,0) circle[radius=\rBOne];
\fill[Acol,    opacity=0.6] (0,0) circle[radius=\rA];

\draw[Outline, thick] (0,0) circle[radius=\rBTwo];
\draw[Outline, thick] (0,0) circle[radius=\rBOne];
\draw[Outline, thick] (0,0) circle[radius=\rA];

\node at (0,0) {\large $A$};
\node at (-0.72,0.42) {\large $B_1$};
\node at (-1.7,1.45) {\large $B_2$};
\node at (-3.55,2.20) {\large $C$};

\end{tikzpicture}
    \caption{The partition of the lattice used in \Thm{thm:gap->LR}.}
    \label{fig:Gap-recovery}
\end{figure}

We now prove the converse direction, that the short-range correlations property implies local stability. The
decay profile in the local stability definition (\Def{def:local-rec}) depends on the
functional form of $g$ in~\Def{def:gapped}.
This dependence is made precise in the theorem below, followed by a discussion on the implications for different choices of $g$.

\begin{theorem}[SRC $\Rightarrow$ local stability] \label{thm:gap->LR} 
    Let $\rho$ be an SRC state, and consider an arbitrary local
    quantum channel given in a Kraus form $\Phi_A(\cdot)=\sum F_k(\cdot)F_k^\dagger$.
    Then for any annulus $B$ of radius $r$ around $A$ and $r_1,r_2\ge0$ such that $r=r_1+r_2$, there is a universal recovery channel $\mcR$ on $AB$ that satisfies
    \begin{equation}\label{eq:gap->LR}
        \sum p_k \norm{\mcR\Big(\frac{1}{p_k}F_k\rho F_k^\dagger\Big)-\rho}_1
        \le C_1 e^{-r_1/\eta} + C_2e^{-r_2/2\zeta},
    \end{equation}
    where $p_k=\Tr(F_k\rho F_k^\dagger)$,
    $C_1=e^{\bigO{|A|}} f(A)$ and
    $C_2= {g(c (|A|+r_1^D))}^{\frac12}$ for a constant $c$.
\end{theorem}

Before we sketch the proof, let us discuss some of its implications.
Note that r.h.s. of \cref{eq:gap->LR} contains both $r_1$ and $r_2$ explicitly, and we can optimize the bound by choosing suitable $r_1, r_2$ subject to the constraint $r=r_1+r_2$.

First, if $g(x)=\poly(x)$, as expected from high-temperature 
Gibbs states~\cite{ref:bakshi2025dobrushin}, 
we may choose both $r_1$ and $r_2$ to be linear in $r$.
For example, we may choose $r_1$ and $r_2$ so that the exponential decay rates match:
\begin{equation}
\label{eq:balanced-split}
r_1 = \frac{\eta}{\eta+2\zeta}\,r,
\quad
r_2 = \frac{2\zeta}{\eta+2\zeta}\,r .
\end{equation}
With this choice, we have
\begin{align}
\sum p_k \norm{\mcR\Big(\frac{1}{p_k}F_k\rho F_k^\dagger\Big)-\rho}_1 
\le  \poly\big(
d^{|A|},
 |A|\,r^D
\big)\cdot e^{-r/(\eta+2\zeta)}.
\end{align}
In particular, for a fixed $A$, the bound decays exponentially in $r$.

However, if we only know $g(x)=e^{O(x)}$, which is the state-of-the-art result for Gibbs state~\cite{ref:chen2025markovian}\footnote{The result in Ref.~\cite{ref:chen2025markovian} also contains a prefactor $|C|$, which can be eliminated using fact \ref{fact:FoR->CMI}, making the local Markovian property hold even for unbounded $C$.},
we have that $C_2=\exp\big(\bigO{A+r_1^D}\big)$.
In this case, choosing $r_1=\bigO{r}$ as in \cref{eq:balanced-split} would lead to superexponential growth in $r$.
Instead, we choose
\begin{equation}
\label{eq:sublinear-split}
r_1 =  (\alpha r)^{1/D},
\qquad
r_2 = r - r_1,
\end{equation}
with $\alpha >0$.
Choosing $\alpha$ small enough, we would find that the second term in~\eqref{eq:gap->LR}
drops exponentially in $r$
while the first term satisfies
\begin{align*}
d_A f(A)\,e^{-r_1/\eta}
=
d_A f(A)\,e^{-(\alpha r)^{1/D}/\eta}.
\end{align*}
We conclude that for $D>1$, \Thm{thm:gap->LR} 
produces a stretched exponential decay bound on local stability.


\begin{proof}[Sketch of Proof]
    Fixing $A$ and $C$, divide $B$ into two consecutive annular regions, denoted by $B_1$ and $B_2$, as illustrated in \Fig{fig:Gap-recovery}. 
    Assuming 
    \begin{equation}
        \norm{\rho_{AB_2C}-\rho_{A}\otimes\rho_{B_2C}}_1 \leq \epsilon_1,
    \end{equation}
    one can show that
    \begin{equation}\label{eq:gap2LR2}
        \sum_k  p_k \norm{\rho_{k,B_2C}-\rho_{B_2C}}_1\leq \epsilon_1,
    \end{equation}
    where $p_k=\Tr(F_k\rho F_k^\dagger)$ and $\rho _{k,B_2C}=\Tr_{AB_1}(\frac{1}{p_k}F_k\rho F_k^\dagger)$.
    Assuming
    \begin{equation}
        \condinf{AB_1}{C}{B_2}_\rho \leq \epsilon_2, 
    \end{equation}
    there exists a recovery map $\tilde{R}: B_2\to AB$ such that~\cite{ref:fawzi2015approximate}:
    \begin{equation}\label{eq:gap2LR4}
        \norm{\tilde{\mcR}(\rho_{B_2C})-\rho}_1\lesssim \sqrt{\epsilon_2}.
    \end{equation}
    Applying $\tilde\mcR$ to \cref{eq:gap2LR2}, the triangle inequality and \cref{eq:gap2LR4} imply that:
    \begin{equation}
        \sum_k p_k\norm{\tilde{\mcR}(\rho_{k,B_2C})-\rho}_1\lesssim \epsilon_1+\sqrt{\epsilon_2}.
    \end{equation}
    The optimal upper bound is given by plugging in $\epsilon_1$ and $\epsilon_2$ given by \cref{def:gapped} and tuning the division between $B_1$ and $B_2$.
    The full proof is presented in Appendix~\ref{app:recovery1}.
    \end{proof}


{\centering \textbf{\small Comments on correlation measures} \par}
\vspace{1em}

In the local stability $\Rightarrow$ SRC direction, we get \cref{eq:LR->DOC}; 
on the other hand, for the reverse direction, the $C_1$ in \cref{eq:gap->LR} contains an $e^{O(|A|)}$ factor.
This mismatch ultimately arises from converting the operator-correlation bound in \Def{def:gapped} into a trace-norm estimate. 
If one replaces \cref{eq:DOC} by a bound of the form
\begin{align}
    \norm{\rho_{AC} - \rho_A \otimes \rho_C}_1 \le f(A)e^{-\mathrm{dist}(A,C)/\eta},
\end{align}
then the prefactor $e^{O(|A|)}$ in $C_1$ in \cref{eq:gap->LR} can be removed.
However, we would instead need an $A$-dependent factor in \cref{eq:LR->DOC}.

The mismatch can be eliminated by choosing the one-way LOCC norm~\cite{ref:matthews2009locc}, 
characterizing the maximum bias achievable by one-way LOCC measurements, 
as the correlation measure. Namely, by replacing \cref{eq:DOC} with
\begin{align}
\norm{\rho_{AC} - \rho_A \otimes \rho_C}_{\text{LOCC}\to} \le f(A)\,e^{-\mathrm{dist}(A,C)/\eta},
\end{align}
no extra factors would appear in either \cref{eq:LR->DOC} or \cref{eq:gap->LR}.
Indeed, the monotonicity of the trace norm applied to \cref{eq:recovery} yields an upper bound on $  \sum_k  p_k \norm{\rho_{k,C}-\rho_C}_1$, the maximization  of which is essentially the one-way LOCC norm.
Conversely, in the proof of \cref{thm:gap->LR}, it suffices to assume \cref{eq:gap2LR2}, which is again given by the one-way LOCC norm.

Similarly, there is an extra factor $|A|+|B|$ in \cref{eq:LR->CMI}, arising from converting between CMI and recovery error.
One may eliminate this factor by replacing CMI by the (logarithmic) fidelity of recovery throughout this work.

\vspace{1em}
{\centering \textbf{\small Comment on symmetries} \par}
\vspace{1em}

The properties of short-range correlations and local stability, as well as their equivalence, can be refined to accommodate symmetries by restricting the local operators in both properties to be symmetric.  As an example, even the low-temperature, ordered phase of the 2d Ising model is SRC, locally stable with respect to Ising symmetric local operators. 



\subsection{Local implementability of postselection}

In this subsection, we explore another consequence of the short-range correlation property.
Namely, we demonstrate that local postselections on an SRC state can be physically realized through local quantum channels.
This provides an alternative perspective on the state's local stability.
In principle, since local stability ensures the local recovery of any measurement trajectory, one could adopt a brute-force strategy of successive measurement and recovery until the target outcome is obtained. 

Below, we develop a more systematic approach based on the following lemma, which establishes that quantum operations applied to Markov chains preserve the Markov property on average.

\begin{lemma} \label{lem:measurement-CMI}
    Consider a tripartite state $\rho_{ABC}$ and a family of Kraus operators $\{F_k\}$ on subsystem $A$ ($\sum F_k^\dagger F_k=1$).
    Denote $\sigma_k=\frac{1}{p_k}F_k \rho F_k^\dagger$ where $p_k=\Tr(F_k \rho F_k^\dagger)$.
    Then
    \begin{equation}
    \sum_k p_k \,\condinf A C B_{\sigma_k} \leq \condinf{A}{C}{B}_{\rho}.
    \end{equation}
\end{lemma}
\begin{proof}
    We introduce a classical ancilla $X$ that stores the outcome, so that there exists a quantum channel $A\to AX$ that converts
\begin{equation}
\rho_{ABC}\mapsto \rho'_{XABC} = \sum_k p_k \ketbra{k}_X \otimes \sigma_{k, ABC}.
\end{equation}
Explicit calculation shows that:
\begin{equation}
\condinf A C {BX}_{\rho'} = \sum_k p_k \,\condinf A C B_{\sigma_k}.
\end{equation}
So we are left to show that $\condinf A C {BX}_{\rho'} \le \condinf A C B_{\rho}$.
This is achieved by combining the data processing inequality:
\begin{equation}
\condinf {XA} C B_{\rho'} \leq \condinf A C B_\rho
\end{equation}
with the chain rule $\condinf{XA}{C}{B}_{\rho'} = \condinf{X}{C}{B}_{\rho'} + I(A:C|BX)_{\rho'}$,
hence,
\begin{equation}
\condinf A C{BX}_{\rho'} \leq \condinf{XA}CB_{\rho'}.
\end{equation}
\end{proof}

\begin{theorem}[local implementability of postselection] \label{thm:gap->LIM}
    Consider a state $\rho$ and a local quantum channel in the Kraus form $\Phi_A(\cdot)=\sum F_k(\cdot)F_k^\dagger$.
    If $\rho$ is SRC, then for any buffering region $B$ of radius $r=r_1+r_2$,
    there exists a (non-universal) family of local quantum channels $\{\Phi_k\}$ on $AB$ such that  
    \begin{equation} \label{eq:local-implementability}
        \sum p_k \norm{\Phi_k(\rho)-\sigma_k}_1
        \leq  C_1 e^{-r_1/\eta} + C_2e^{-r_2/2\zeta} ,
    \end{equation}
    where $p_k, C_1$ and $C_2$ are as defined in \cref{thm:gap->LR}, and $\sigma_k=\frac1{p_k}F_k\rho F_k^\dagger$.
\end{theorem}
The above theorem shows that SRC is sufficient for local implementability of postselection.
Together with \cref{thm:gap->LR}, this shows that an SRC state $\rho$ is two-way connected to trajectory states $\sigma_k$ via local quantum channels.

\begin{proof}
    Partition the lattice to $AB_1B_2C$ as in \cref{fig:Gap-recovery}. 
    Define $\Phi_k: AB\to AB$ as the map obtained by tracing out $AB_1$ followed by the recovery channel $\mcR_k:B_2\to AB$ associated with $\sigma_k$.
    Thus, we get
    \begin{align}
         \norm{\Phi_k(\rho)-\sigma_k}_1
        & \le \norm{\rho_{B_2C}- \sigma_{k,B_2C} }_1
        +\norm{\mcR_k(\sigma_{k,B_2C})-\sigma_k}_1,
    \end{align}
    where we used the triangle inequality and DPI with $\mcR_k$.
    To get a bound on \cref{eq:local-implementability} we take the sum over $k$. 
    According to \cref{eq:gap2LR2} the first term is bounded by
        $\norm{\rho_{AB_2C}-\rho_A\otimes\rho_{B_2C}}_1$.
    The second term is bounded (up to a constant prefactor) by $\sum_k p_k \big(\condinf{AB_1}{C}{B_2}_{\sigma_k}\big)^{\frac12}$.
    Using concavity of the square root together with \Lem{lem:measurement-CMI},
    this is further bounded by $\big(\condinf {AB_1}{C}{B_2}_\rho \big)^{\frac12}$.
    Finally, the r.h.s of \cref{eq:local-implementability} follows from the definition of SRC as seen in the proof of \cref{thm:gap->LR}.
\end{proof}

While SRC is sufficient for local implementability, in Appendix~\ref{app:countereg}, we provide a counterexample showing that it is not necessary.
On the other hand, while decay of correlations is necessary for local implementability, another counterexample in appendix \ref{app:countereg} shows that it is not sufficient. 
Therefore, local implementability of postselection is a property that sits strictly in between decay of ordinary correlations and our definition of SRC.

\subsection{Phase invariance}
We now show that local stability is a property of an entire phase of matter.  \cRef{ref:sang2025mixedstatephases} proposed a refined definition of mixed-state phase based on the notion of locally reversible circuits: 
a state $\rho$ is said to be in the same phase as $\rho '$ if they are connected by locally reversible circuits\footnote{We note that generalizing the notion of local reversibility from circuits to continuous evolution is an open problem, so here we focus on the circuit definition.}-- if there exists a shallow circuit channel  $\mcF:\rho \mapsto \rho' $ and a circuit channel $\; \mcD:\rho' \mapsto \rho$,
such that $\mcD$ reverses the action of $\mcF$ on $\rho$ gatewise. 
Formally,
    \begin{align} \label{eq:LR-circuit}
        \mcF = \mcF_T \dots \mcF_1, \qquad & \mcD = \mcD_1 \dots \mcD_T, \nonumber \\
        \mcF_t = \prod_{\ell} \mcF_{t,\ell}, \qquad & \mcD_t = \prod_{\ell} \mcD_{t,\ell}, \\
        \rho' = \mcF(\rho), \qquad & \rho = \mcD(\rho'), \nonumber \\
        \quad \bigl(\mcD_{k,\ell} \circ \mcF_{k,\ell}\bigr) \mcF_{k-1} \cdots \mcF_1(\rho) &= \mcF_{k-1} \cdots \mcF_1(\rho),~~(\forall k\le T). \nonumber
\end{align}
    Here, each $\mcF_t$ (and $\mcD_t$) is a layer of non-overlapping local channels, each of which act on a few contiguous sites.
    For simplicity of illustration, we assume that $\mcF_{k,\ell}$ and $\mcD_{k,\ell}$ share precisely the same support. In general, the gates in $\mcD$ may have slightly larger supports; our arguments remain essentially unchanged under this more general scenario.

It is shown \cite{yi2026universal} that decay of correlations and CMI are invariant under locally reversible circuits.
Therefore, our notion of SRC (\Def{def:gapped}), and hence 
local stability, are also universal properties of a phase. 
Here we provide an alternative and direct proof for the latter.
\begin{theorem}[Phase invariance of local stability] 
\label{Thm:invariance}
    Suppose that $\rho$ and $\rho'$ are connected by a locally reversible shallow circuit.
    Then $\rho$ is locally stable if and only if $\rho '$
    is.
\end{theorem}

\begin{figure*}[ht!]
    \centering
    \subfloat[\label{fig:LR-a}Locally reversible circuit between $\rho$ and $\rho '$]{
    \includegraphics[width= 0.45\linewidth]{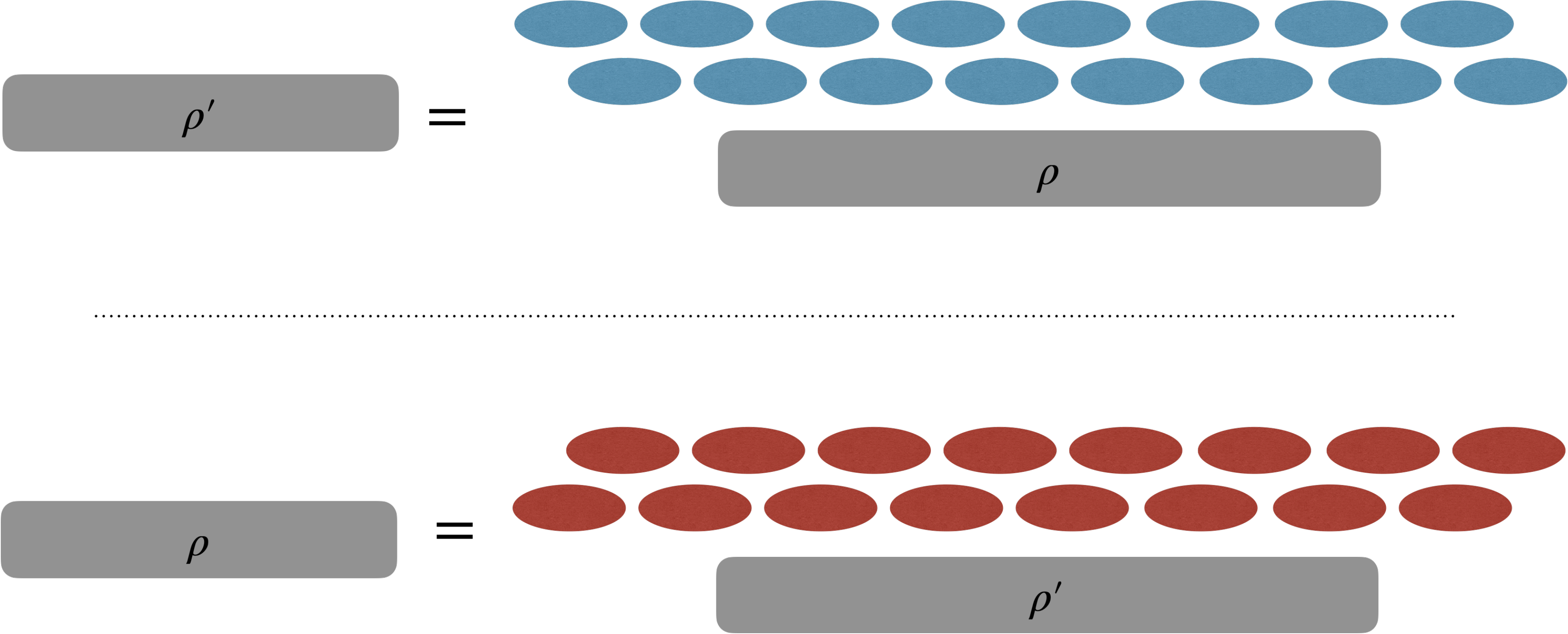}} \hfill
   \subfloat[\label{fig:LR-b}Local recovery channel $\mcR '$]{\includegraphics[width= 0.4\linewidth]{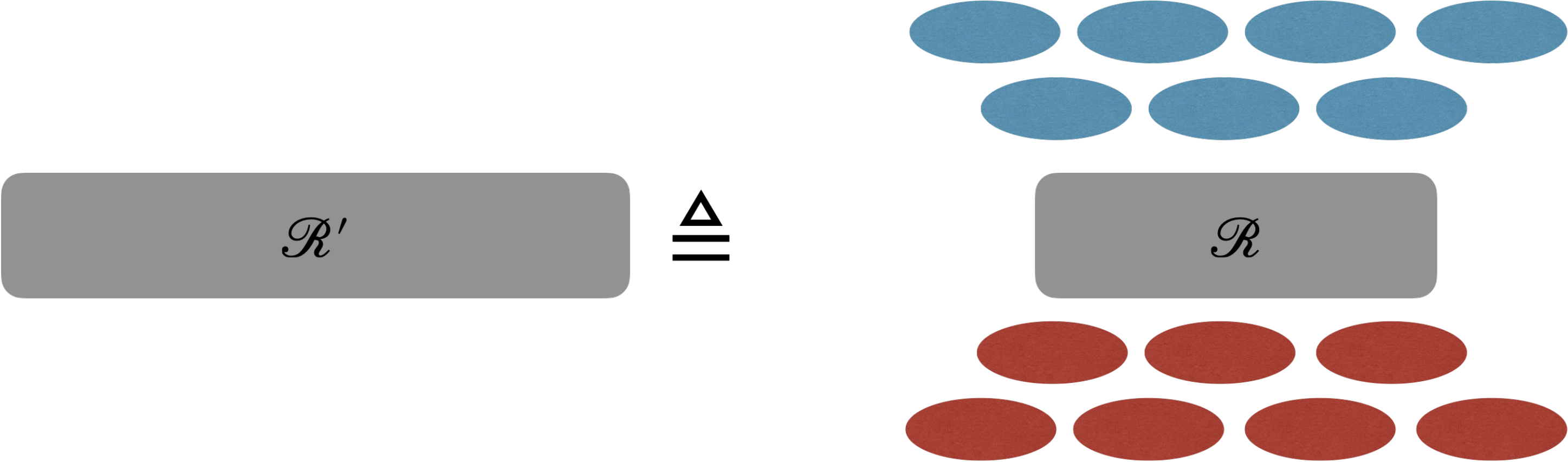}}\\
    \subfloat[\label{fig:LR-c}Proof of recovery by $\mcR '$]{\includegraphics[width= 0.8\linewidth]{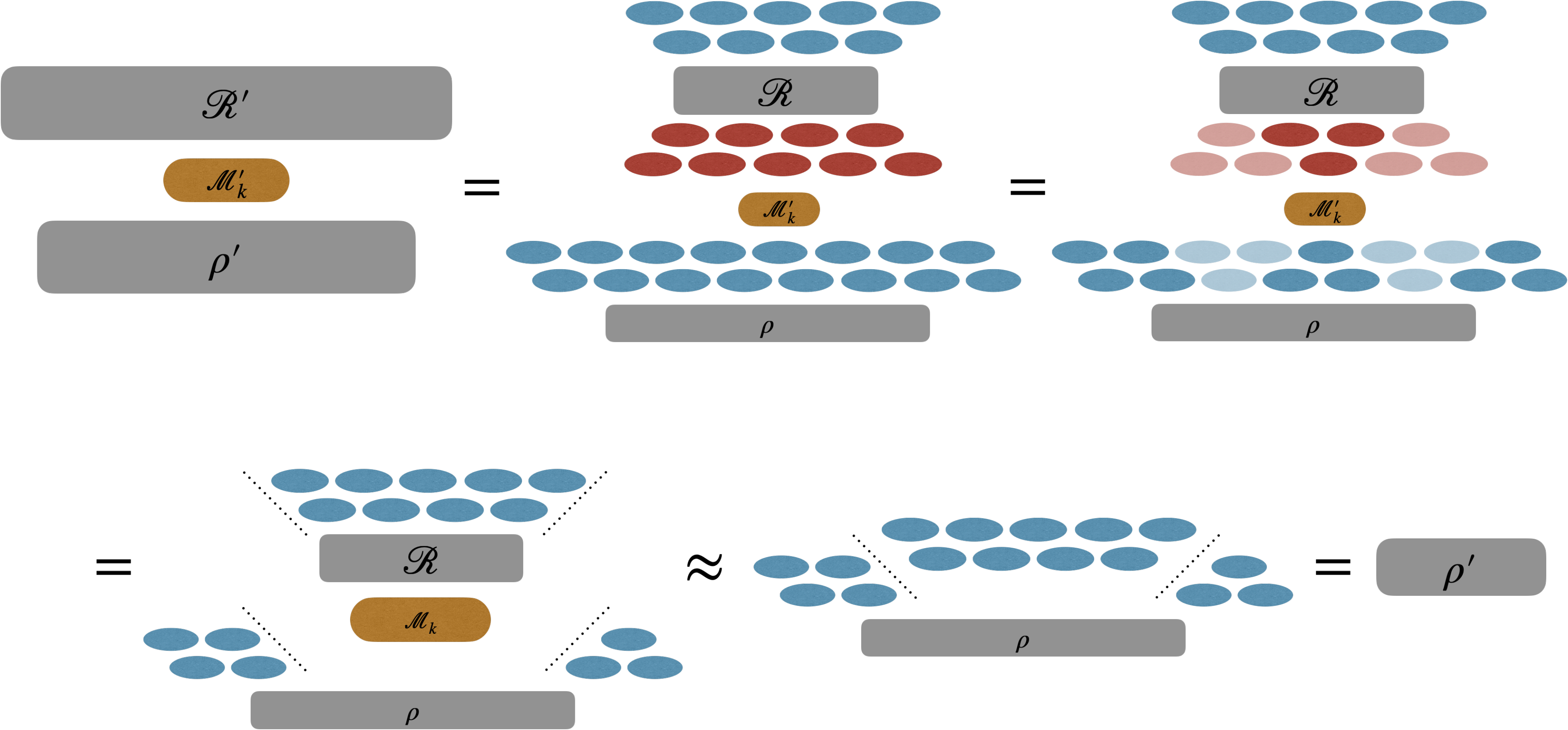}}
    \caption{Schematic proof of \Thm{Thm:invariance}: 
    \ref{fig:LR-a} - Locally reversible circuit $\mcF:\rho\mapsto\rho'$ 
     and the reverse circuit $\mcD$ .
    \ref{fig:LR-b} - Recovery channel $\mcR'$ for $\rho'$ based on      ``dressing'' the recovery channel $\mcR$ for $\rho$.    \ref{fig:LR-c} - $\mcR'$  recovers $\rho'$ from the effect of local operations.}
    \label{fig:LRC}
\end{figure*}
\begin{proof}
    (The proof idea is illustrated in \Fig{fig:LRC}.) 
    Let $\mcF:\rho \mapsto \rho' $ and $\; \mcD:\rho' \mapsto \rho$
    be shallow quantum circuit channels as described in \cref{eq:LR-circuit}.
    Denote their depth as $T$.
    Consider a local channel acting on $\rho'$ and supported on $A$, with a decomposition $\Phi^\prime_A(\cdot)=\sum_k \mcM^\prime_k(\cdot)$ where each $\mcM^\prime_k(\cdot)=F_k (\cdot) F_k^\dagger$ is a single Kraus operation.
    
    To recover such local operations, we first pick a recovery channel $\mcR$ for $\rho$ supported around $A$ with a radius (counting from the boundary of $A$) $r\gg T$, as guaranteed by \Def{def:local-rec}, 
    then define a local quantum channel $\mcR'$ as:
    \begin{align} \begin{split}
        \mcR' &\EqDef \Big(\prod_{i\in  \mcR_\uparrow} \mcF_i\Big) 
        \circ \mcR \circ\Big(\prod_{i\in  \mcR_\downarrow} \mcD_i\Big) \\
        & \EqDef  \mcF_{  \mcR_\uparrow}\circ \mcR \circ  \mcD_{  \mcR_\downarrow}.
    \end{split}\end{align}
    Here, $ \mcR_\uparrow $ and $ \mcR_\downarrow $ denotes the forward and backward light cones of $\mcR$ under $\mcF$ and $\mcD$, respectively (they are space-time regions rather than spatial regions, see \Fig{fig:LR-b}).
   We then have:
    \begin{align} \label{eq:why-dressed?}
        \mcR' \mcM^\prime _k(\rho')  
        &=\mcF_{  \mcR_\uparrow}  \mcR    (\mcD_{  \mcR_\downarrow}\mcM'_k   \mcF_{\mcR_\uparrow}) \mcF_{  (\mcR_\uparrow)^c} (\rho)\\
        & = \mcF_{  \mcR_\uparrow} \mcR  \mcM_k \mcF_{  (\mcR_\uparrow)^c} (\rho)
        \nonumber 
        =\mcF \mcR  
        \mcM_k (\rho).
    \end{align}
    Here, $\mcM_k
    =\mcD_{A_\uparrow} \circ 
    \mcM^\prime_k\circ  \mcF_{A_\downarrow}$ is a local operation supported in the light cone of $A$ under $\mcF$, and  $\mcR_\uparrow^c$ is the spacetime complement of $\mcR_\uparrow$.
    The second equality is due to local reversibility  
    and the third equality is due to disjointness of supports (for an illustration, see \cref{fig:LR-c}).
    
    The set $\{\mcM_k \}$ defines a quantum instrument
    acting on $\rho$, which is recoverable by $\mcR$ by assumption:
    \begin{align} \label{eq:dressed-R}
        \sum_k p_k\norm{\frac{1}{p_k}\mcR\circ\mcM_k(\rho)-\rho}_1 
        \le \tilde f(A)\,e^{-(r-T)/\xi},
    \end{align}
    where $\tilde f (A)=f(A^{+T})$, which admits the same scaling as $f$ (here $A^{+T}$ is the $T$-radius neighbor of $A$), and
    \begin{align}
    p_k&=\Tr{\mcM_k(\rho)}
    =\Tr{\mcD_{A_\uparrow} \mcM^\prime_k  \mcF_{A_\downarrow}(\rho)}
     = \Tr{ \mcM^\prime_k  \mcF_{A_\downarrow}(\rho)   }
     \nonumber\\
     &= \Tr{ \mcM^\prime_k  \mcF(\rho) }
     =\Tr{ \mcM^\prime_k(\rho')  }.
    \end{align}
    Here, the third equality is due to the trace preservation of the $\mcD_{A_\downarrow}$ and the fourth equality is due to the light cone property.
    Finally, acting $\mcF$ on \cref{eq:dressed-R}, using monotonicity of the trace distance, we get:
     \begin{align}
         \sum_k p_k \norm{\frac{1}{p_k}\mcR'(F_k \rho' F_k^\dagger) -\rho'}_1 \le \tilde f (A)   e^{-(r-T)/\xi}.
     \end{align}
     Noticing that the radius of $\supp(\mcR')$ is at most $r+T$, this proves that $\mcR'$ serves as a local recovery channel for $\rho'$. 
\end{proof}

\subsection{Local stability via dynamics}
\label{sec:dynamics}

When an SRC state arises dynamically as the steady state of a Lindbladian, a natural question is whether that Lindbladian can serve as the recovery map against local operations.  
One motivation for this question is the use of a truncated Gibbs sampler to recover against partially tracing out a Gibbs state~\cite{ref:chen2025markovian,ref:bakshi2025dobrushin}. 
Such recovery, as well as the one required in \cref{def:local-rec}, requires knowledge of the location on which the quantum operations act.
Moreover, the aforementioned papers address specific cases of Gibbs samplers.
This raises the broader question of how general
such recovery properties are, and whether they 
extend beyond specific constructions to generic Lindblad dynamics.

Here, we argue that, provided that the Liouvillian is reversible and gapped,
the natural global short-time evolution can serve as a recovery map for local stability, 
complementing the Petz map approach in \cref{thm:gap->LR}.

Formally, let $\mcL=\sum_i\mcL_i$ be a local Lindbladian 
satisfying a detailed-balance condition with respect to a unique, full-rank steady state $\rho$. 
Although the quantum detailed-balance exhibits many definitions \cite{ref:Temme2010chi},
our main focus here will be on the KMS condition and the GNS condition as follows:
\begin{align} \label{def:detailed-balance}
    \mcL\circ \Gamma = \Gamma \circ \mcL^\dagger, \qquad \Gamma (A) = \begin{cases}
        A  \rho  & \mathrm{GNS}, \\
        \rho^{\frac 12}A \rho^{\frac 12} & \mathrm{KMS} .
    \end{cases}
\end{align}
Then, for a local channel and its decomposition $\Phi_A(\cdot)=\sum_k F_k(\cdot)F_k^\dagger $, define $\sigma_k=\frac{1}{p_k}F_k\rho F_k^\dagger$ where $p_k=\Tr{F_k\rho F_k^\dagger}$.
We suggest considering the natural time evolution $\mcR_t \EqDef e^{\mcL t}$ as a potential universal recovery map $\sigma_k\mapsto \rho$.


In the following, we establish a standard convergence guarantee
and prove it in Appendix~\ref{app:recovery}.
\begin{prop}[Fast convergence from gap] \label{clm:Var-conv}
    Let $\rho$ be a unique, full-rank steady state of a detailed-balance Lindbladian $\mcL$ with a gap $\gamma$. Then for any state $\tilde \rho$ we get
    \begin{align}
        \norm{e^{\mcL t}(\tilde\rho)-\rho}_1 \lesssim e^{-\gamma t} \times 
        \begin{cases}
            \exp(\frac{\ln 2}{2}\mathrm D_2(\tilde \rho||\rho)) & \mathrm{GNS},\\
            \exp(\frac{\ln 2}{2}\tilde {\mathrm D}_2(\tilde \rho||\rho)) & \mathrm{KMS} ,
        \end{cases}
    \end{align}
    where $\mathrm{D}_2$ is the Petz-Rényi relative entropy
    and $\tilde{\mathrm{D}}_2$ is the sandwiched-Rényi 
relative entropy with $\alpha=2$ (see 
    Appendix~\ref{app:settings} for definitions). 
\end{prop}

This statement applies to \emph{any} initial state $\tilde\rho$.
It reflects how the relaxation time of a Lindbladian 
is controlled by the deviation of the initial state from the steady state, captured by the Rényi-2 divergences, as well as the spectral gap of the Lindbladian $\gamma$, in direct analogy to the Hamiltonian scenario.
In general, an initial state with exponentially small overlap with $\rho$ 
might lead to relaxation times that scale polynomially
with system size.

In our setting, however, the relevant initial states arise from applying local Kraus outcomes to $\rho$.
If such perturbed states are not significantly different from the steady state, the Lindblad evolution may thermalize quickly and recover the steady state.
Furthermore, if $\mcL$ is additionally frustration-free (as in Gibbs samplers),
the evolution is highly localized around the corresponding patch due to the Lieb-Robinson bound,
and hence can be truncated to that region and made into a recovery map as in \Def{def:local-rec}.
We omit any analysis of such truncation here and refer to \cRef{ref:brandao2015area} for related results.

To formally deduce fast recovery, one would ideally want a mixing bound of the form:
\begin{align} \label{eq:fast-converge}
    \sum_k p_k\norm{\mcR_t(\sigma_k)-\rho}_1 \le C(A)\, e^{-\gamma t}.
\end{align}
For such a bound to hold, 
it is important that $\mathrm D_2(\sigma_k\|\rho)$ or
$\tilde{\mathrm D}_2(\sigma_k\|\rho)$ scale with the size of the region $A$ rather than the full system.
It is reasonable to expect this property to hold for physically relevant SRC states.
Indeed, for such states, \cref{thm:gap->LR,thm:gap->LIM} show that local operations perturb the state only locally around $A$, suggesting that the resulting distance naturally scales with $|A|$. 
Adding another perspective, for SRC states one can also show the existence of a single recovery map $\mcR_{B_2\to B_2C}$ that approximately reconstructs the full states $\rho,\sigma_k$ from their $AB$ marginals $\rho_{AB},\sigma_{k,AB}$, respectively (see \cref{fig:Gap-recovery} for the geometry). 
The data processing inequality would then imply that the distance is captured by the difference localized around $A$ (however, bounding the error term seems to be challenging).
While we do not have a general proof for this scaling, we verify it explicitly
in Appendix~\ref{app:recovery} for several classes of thermal states.
These include commuting Hamiltonians,
one-dimensional systems, and high-temperature regimes.

In appendix \cref{app:detectRecover}, we provide a more explicit global recovery map tailored to commuting Hamiltonians,
based on their local GNS Gibbs sampler and the detectability lemma (see similar constructions in \cRefs{ref:Kastoryano2016commutinggibbssampler,ref:fang2026GibbssamplingDL}).
It follows from stacking the infinite-time projectors onto the local kernels in a tower.


\section{Non-linear Correlators}
\label{sec:non-lin}

\subsection{Definition and basic properties}
Given a state $\rho$ and operator $O$, we consider a family of observables as follows:
\begin{equation}\label{eq:Cpq}
\mcC_{p,q}(O,\rho)
    \EqDef \norm{\rho^{\frac{p}{2}}O \rho^{\frac{q}{2}}}_{2s}
    =
    \left[\Tr\big(\rho^{\frac{p}{2}} O \rho^{q} O^\dagger \rho^{\frac{p}{2}}\big)^{s}\right]^{\frac{1}{2s}},
\end{equation}
where $s=\frac{1}{p+q}$ and $0< p,q\leq 1$\footnote{
Some of the arguments below using concavity/convexity also apply to $(p,q)$ or $(q,p)\in(-1,0]\times[1,2]$. However, we restrict $p,q\geq 0$ to ensure that $\mcC_{p,q}$ behaves well (e.g., does not diverge).
The line $pq=0$ was omitted since $\mcC_{p,q}$ is not continuous and our results below become irrelevant there.}.
Unlike the ordinary expectation value $\expval{O}_\rho=\Tr(\rho O)$, these observables 
are in general nonlinear in $O$ and $\rho$.

\begin{figure}[]
    \centering
        \begin{tikzpicture}[>=stealth, line cap=round]

  \draw[->] (0,0) -- (4.5,0) node[right] {$p$};
  \draw[->] (0,0) -- (0,4.5) node[above] {$q$};

  \draw (0,0) -- (4,0) -- (4,4) -- (0,4) -- cycle;

  \draw[line width=3pt, green!70!black] (0,0) -- (4,4);
  \draw[line width=3pt, blue!70!black] (0,4) -- (4,0);
  \draw[line width=3pt, magenta] (0,4) -- (4,4);

  \fill (4,4) circle (4pt);

  \path[fill] (2,2.22) -- (1.80,1.82) -- (2.20,1.82) -- cycle;

  \begin{scope}[xshift=4.8cm, yshift=3.9cm]

    \def\dy{0.55}

    \fill (0,0) circle (3pt);
    \node[right] at (0,0) {\;fidelity};

    \path[fill] (0,-\dy) -- (-0.18,-\dy-0.28) -- (0.18,-\dy-0.28) -- cycle;
    \node[right] at (0,-\dy-0.14) {\;R\'enyi-1};

    \draw[line width=3pt, green!70!black]
          (-0.4,-2.3*\dy) -- (0.7,-2.3*\dy);
    \node[right] at (0.7,-2.5*\dy) {$\norm{\rho^{\frac{p}{2}}O \rho^{\frac{p}{2}}}_{\frac 1p}$};

    \draw[line width=3pt, magenta]
          (-0.4,-3.5*\dy) -- (0.7,-3.5*\dy);
    \node[right] at (0.7,-3.5*\dy) {sandwiched R\'enyi};

    \draw[line width=3pt, blue!70!black]
          (-0.4,-4.5*\dy) -- (0.7,-4.5*\dy);
    \node[right] at (0.7,-4.5*\dy) {Petz-R\'enyi};

  \end{scope}

\end{tikzpicture}

    \caption{$\mcC_{p,q}$ for $p,q\in (0,1]^2$.}
    \label{fig:placeholder}
\end{figure}

The family \cref{eq:Cpq} includes most non-linear observables appearing in the previous literature, usually for diagnosing  mixed-state phases in the presence of a global symmetry~\cite{ref:Zhang2026fidelity,ref:lessa2025fidelity,ref:Weinstein2025CPcor,
liu2025diagnosing}.  
Let us review some familiar cases.

First, we set $p=q$ and denote $\mcC_{p,p}$ as $\mathrm F_p$.
For $p=1$, this is the fidelity correlator proposed in Ref.~\cite{ref:lessa2025fidelity}:
\begin{equation}\label{def:fidcor}
    \mathrm F_1(O,\rho)=\expval{O}_F = \Tr(O\rho O^\dagger)^{\frac{1}{2}}\cdot \mathrm F(\rho,\frac{O\rho O^\dagger}{\Tr(O\rho O^\dagger)}),
\end{equation}
where $\mathrm F(\cdot,\cdot)$ is the (unsquared) Uhlmann fidelity.
For $p=\frac{1}{2}$, this is the Renyi-1 correlator~\cite{ref:Weinstein2025CPcor}
(also dubbed Wightman correlator in~\cite{liu2025diagnosing}):
\begin{equation}
    \mathrm F_{\frac{1}{2}}(O,\rho)=\left[\Tr(\sqrt{\rho}O\sqrt{\rho}O^\dagger)\right]^{\frac{1}{2}}.
\end{equation}
For unitary $O$, it equals the square root of the Holevo fidelity between $\rho$ and $O\rho O^\dagger$, analogous to \cref{def:fidcor}.
In general, for $0\leq p\leq 1$~\cite{liu2025diagnosing}, 
\begin{equation}
    \mathrm F_p(O,\rho)=\norm{\rho^{\frac{p}{2}}O \rho^{\frac{p}{2}}}_{\frac{1}{p}}.
\end{equation} 

Second, if we set $p+q=1$, then
\begin{equation}
    \mcC_{1-q,q}(O,\rho)
=    \left[\Tr\big(O \rho^{q} O^\dagger \rho^{1-q}\big)\right]^{\frac{1}{2}}.
\end{equation}
This observable has been considered in \cite{liu2025diagnosing,ref:Huang2025hydro-SWSSB,ref:Kliesch2014locality}.
In particular, if $O$ is a unitary $U$, then
\begin{equation}
    \mcC_{1-q,q}(O,\rho) = 2^{\frac{q-1}{2}\mathrm D_q(O\rho O^\dagger || \rho)},
\end{equation}
where $\mathrm D$ is the (Petz-type) Rényi relative entropy~\cite{petz1986quasi}.

Third, besides the (Petz-type) Rényi relative  entropy, we can equally well consider the sandwiched Rényi entropy $\tilde{\mathrm D}$ by setting $q=1$ and defining $p=\frac{1-\alpha}{\alpha}$, then
\begin{equation} \label{eq:sandwich-Reny}
\begin{aligned}
    \mcC_{p,1}(O,\rho)
    &=    \left[\Tr\big(\rho^{\frac{1-\alpha}{2\alpha}} O \rho O^\dagger \rho^{\frac{1-\alpha}{2\alpha}}\big)^{\alpha}\right]^{\frac{1}{2\alpha}}\\
    &=  2^{\frac{\alpha-1}{2\alpha}\tilde{\mathrm D}_\alpha(O\rho O^\dagger || \rho)}.
\end{aligned}
\end{equation}
This type of observable is considered in \cRef{ref:lessa2025fidelity}.

Finally, we comment on the general case. Setting $p=\frac{1-\alpha}{z}$, $q=\frac{\alpha}{z}$ and taking $O$ to be a unitary,
$\mcC_{p,q}$ takes the form of a divergence,
\begin{align} \label{eq:pq->az}
     \mcC_{p,q}(O,\rho) & = 2^{\frac{\alpha-1}{2z}\cdot\mathrm{D}_{\alpha,z}(O\rho O^\dagger\| \rho)},
    \end{align}
 where $\mathrm{D}_{\alpha,z}$ is the
$\alpha$-$z$-Rényi divergence~\cite{jaksic2011entropic,ref:audenaert2015alpha} (see Appendix~\ref{app:settings}).
    
{~}

We record some useful properties of $\mcC_{p,q}$ as follows. 
Items \ref{prop5:bul-pos}-\ref{prop5:bul-ti} are basic, and item \ref{prop5:bul-log} can be proved succinctly using complex interpolation theory. 
For full proofs, see Appendix~\ref{app:non-lin-basic}.
\begin{prop}\label{prop:Cpqbasic}
    The two-parameter family $\mcC_{p,q}(O,\rho)$ satisfies the following:
    \begin{enumerate}[(1)]
        
        \item \label{prop5:bul-pos} 
        Positivity: $\mcC_{p,q}(O,\rho)\geq 0$. For any $(p,q)$, 
        $\mcC_{p,q}(O,\rho)=0$ if and only if $O(\supp(\rho))\subseteq \supp(\rho)^\perp$;
        
        \item \label{prop5:bul-mon}
        Monotonicity: $\mcC_{p,q}(O,\rho)$ is monotonically non-increasing in both $p$ and $q$;
        
        \item \label{prop5:bul-range} 
        Range: $|\expval{O}_\rho|\leq\mcC_{p,q}(O,\rho)\leq\norm{O}_\infty$;

        \item \label{prop5:bul-sym} 
        Symmetry: $\mcC_{p,q}(O,\rho) = \mcC_{q,p}(O^\dagger,\rho)$;

        \item \label{prop5:bul-mult} 
        Multiplicativity:
        $\mcC_{p,q}(O_A\otimes O_B,\rho_A\otimes \rho_B)
        =\mcC_{p,q}(O_A,\rho_A)\cdot\mcC_{p,q}( O_B,\rho_B)$
    
        \item \label{prop5:bul-ti}
        Continuity: $\abs{\mcC_{p,q}(O,\rho)-\mcC_{p,q}(O,\sigma)} \lesssim
        \norm{O}_\infty \bur{\rho}{\sigma}^{\min\{p,q\}}$, where 
        $\bur{\rho}{\sigma}$ is the Bures metric;

        \item \label{prop5:bul-log}
        Log-convexity in $(p,q)$: the logarithmic of \cref{eq:Cpq} is 
        jointly convex in $(p,q)$, namely,
        \begin{equation}
            \mcC_{(1-\theta)p_0+\theta p_1,(1-\theta)q_0+\theta q_1}(O,\rho)\le 
        \mcC_{p_0,q_0}^{1-\theta}\mcC_{p_1,q_1}^{\theta}.
        \end{equation}
    \end{enumerate}
\end{prop}

 Another important property of $\mcC$ is given in Fact~\ref{thm:DPI} below.
It is a generalized version of Lieb's concavity theorem~\cite{ref:Lieb1973convexity}, and we refer the reader to \cRef{ref:zhang2020wigner} for a proof.
Such concavity implies a data-processing inequality (DPI).
We include a proof here for completeness.
\begin{fact}[Concavity and DPI]\label{thm:DPI}
\leavevmode
\begin{enumerate}[(1)]
    \item \label{bul1:DPI} For any state $\rho$, operator $O$, and $(p,q)\in(0,1]^2$,
    the function 
    \begin{equation}
        f(\rho)=\Tr\big(\rho^{\frac{p}{2}} O \rho^{q} O^\dagger \rho^{\frac{p}{2}}\big)^{\frac{1}{p+q}}
    \end{equation} 
    is concave in $\rho$.
    \item \label{Bul2-DPI} For a bipartite state $\rho_{AB}$ and an operator $O_A$ on $A$,
    the observables $\mcC_{p,q}(O_A,\cdot)$ satisfy ``data-processing inequality": for any quantum channel $\Phi: B\to B'$, we have:
    \begin{equation}\label{eq:DPICpq}
        \mcC_{p,q}(O_A,\rho_{AB}) \le  \mcC_{p,q}(O_A,\Phi(\rho_{AB})).
    \end{equation}
\end{enumerate}
\end{fact}
\begin{proof}[Proof of $(1)\Rightarrow(2)$]
    We introduce an environment $E$ to dilate the channel $\Phi$ to a unitary:
    \begin{equation}
        \Phi(\cdot) 
        = \Tr_{E'}\left[V_{BE}(\cdot\otimes\ketbra{0_E})V_{BE}^\dagger\right],
    \end{equation}
    where $B'E'=BE$.
    Equivalently, 
    \begin{equation}\label{eq:temp-1}
        \Phi(\cdot)\otimes \frac{\Id_{E'}}{\dim(E')} 
        = \EX \left[U_{E'}V_{BE}(\cdot\otimes\ketbra{0_E})V_{BE}^\dagger U_{E'}^\dagger\right],
    \end{equation}
    where $\EX = \int dU_{E'}$ is the Haar average over $U_{E'}$.
    
    Now consider the function $f(\rho)$ for $O=O_A$.
    It suffices to show $f(\rho_{AB})\leq f(\Phi(\rho_{AB}))$.
    
    By item (\ref{bul1:DPI}), $f(\rho)$ is concave in $\rho$.
    In particular, conjugating $\rho_{AB}\otimes\ketbra{0_E}$ with $U_{E'}V_{BE}$ and taking the average over $U_{E'}$, we find:
    \begin{equation} \label{eq:EX}
    \begin{aligned}
        &\EX \left[f(U_{E'} V_{BE}(\rho_{AB}\otimes\ketbra{0_E})V_{BE}^\dagger U_{E'}^\dagger)\right]\\
        \leq& 
        f\left(\EX \left[U_{E'}V_{BE}(\rho_{AB}\otimes\ketbra{0_E})V_{BE}^\dagger U_{E'}^\dagger\right]\right).
    \end{aligned}
    \end{equation}
    Since $O_A$ acts only on $A$, $f$ is invariant under conjugation by any unitary on $BE$. Hence, the left-hand side of \cref{eq:EX} is constant, even before averaging:
    \begin{equation}
        (\text{\cref{eq:EX}, l.h.s.}) = f(\rho_{AB}\otimes\ketbra{0_E}) = f(\rho_{AB}).
    \end{equation}
    For the right-hand side, \cref{eq:temp-1} implies:
    \begin{equation}
        (\text{\cref{eq:EX}, r.h.s.}) = f\Big(\Phi_B(\rho_{AB})\otimes \frac{\Id_{E'}}{\dim(E')}\Big)=f\big(\Phi_B(\rho_{AB})\big),
    \end{equation}
    where the last equality uses the fact that $f$ is  homogeneous of degree 1 in $\rho$, enforced by the choice $s=\frac{1}{p+q}$.
\end{proof}

\subsection{Quantum Markov chains and local computability}

Here we introduce a notion central to our work which we term
\emph{local computability}.
Recall that traditional expectation values satisfy $\Tr{\rho O_A}=\Tr{\rho_A O_A}$,
and therefore depend only on the reduced state $\rho_A$ rather than the global state.
In contrast, the family of observables defined in \cref{eq:Cpq}, being nonlinear in $\rho$ and $O_A$, a priori require information about the full state rather than $\rho_A$ alone.

However, we show that, as a consequence of the concavity and DPI, for a QMC $\rho_{ABC}$, these nonlinear observables can be evaluated using only the marginal on $AB$:
\begin{align}\label{eq:localcompute1}
        \mcC_{p,q}(O_A,\rho_{AB}) = \mcC_{p,q}(O_A,\rho_{ABC}).
\end{align} 
Interestingly, the converse also holds, providing an equivalence characterization of quantum Markov chains.

\begin{theorem}[QMC $\Leftrightarrow$ local computability]\label{thm:local-comp}
\leavevmode
\begin{enumerate}[(1)]
\item   \label{bul1:QMC-operational} If $\rho_{ABC}$ is a Markov chain with respect to $A$-$B$-$C$, then for any operator $O_A$, \cref{eq:localcompute1} holds.
    \item \label{bul2:QMC-operational} Conversely, fix a $(p,q)$ in $(0,1]\times (0,1]$.
    If \cref{eq:localcompute1} holds for all operators $O_A$, then $\rho_{ABC}$ is a quantum Markov chain.
\end{enumerate}
\end{theorem}

\begin{proof}
    (1)     
    We apply the data-processing inequality (\ref{eq:DPICpq}) twice.
    First, applying it to the state $\rho_{ABC}$ and $\Tr_C$, we find:
    \begin{align}
        \mcC_{p,q}(O_A,\rho_{ABC}) \leq \mcC_{p,q}(O_A,\rho_{AB}).
    \end{align}
    Second, applying it to the state $\rho_{AB}$ and the Petz recovery channel $\mcR: B\to BC$, we find:
    \begin{align}
        \mcC_{p,q}(O_A,\rho_{AB}) \leq \mcC_{p,q}(O_A,\rho_{ABC}).
    \end{align}
    Combining two inequalities above, we have proved \cref{eq:localcompute1}.

(2) (sketch)
We give the main ideas here and defer the full proof to Appendix~\ref{app:non-lin-local-comp}.
Let $\ket{\Psi}_{ABCE}$ purify $\rho_{ABC}$, and set $\omega\EqDef\Psi_{CE}$, $\eta\EqDef\Psi_E$, and $\Phi\EqDef\Tr_C$.
To each $O_A$ we associate the operator 
\begin{equation}
    X_O\EqDef\Tr_A[(O_A\otimes\Id)\Psi_{ACE}]\in \mathrm L (\mathbb H_{CE}).
\end{equation}
We show that the nonlinear observables become weighted norms on the complementary side:
\begin{equation}\label{eq:sketch-purification}
    \begin{aligned}
        \mcC_{p,q}(O_A,\rho_{AB})&=\norm{\omega^{\frac{q-1}{2}}X_O\,\omega^{\frac{p-1}{2}}}_{2s},\\
        \mcC_{p,q}(O_A,\rho_{ABC})&=\norm{\eta^{\frac{q-1}{2}}\Phi(X_O)\,\eta^{\frac{p-1}{2}}}_{2s}.
    \end{aligned}
\end{equation}
Moreover, by the duality $\condinf{A}{C}{B}_\Psi=\condinf{A}{C}{E}_\Psi$ for the pure state $\Psi$, $\rho_{ABC}$ is a quantum Markov chain as long as a single channel recovers every $X_O$ from $\Phi(X_O)$.
We show that the Petz map $\mcP_{\omega,\Phi}$ does so.

\emph{Case 1: $p=q=1$.}
Here \cref{eq:sketch-purification} reduces to trace norms.
Applying local computability to $O_A/\Tr X_O-t\Id$ for $O_A\ge0$ gives, for $\tau_O\EqDef X_O/\Tr X_O$,
\begin{equation}
    \norm{\tau_O-t\omega}_1=\norm{\Phi(\tau_O)-t\eta}_1\qquad\text{for all }t\ge0 .
\end{equation}
By the characterization of sufficiency via hypothesis testing~\cite{jencova2024hypothesis}, $\Phi$ is sufficient for $(\tau_O,\omega)$, and Petz's theorem~\cite{petz1986sufficient,petz1988sufficiency} gives $\mcP_{\omega,\Phi}(\Phi(X_O))=X_O$.

\emph{Case 2: $(p,q)\neq(1,1)$.}
We apply local computability to $\Id+zO_A$ and take the mixed derivative $\partial_z\partial_{\bar z}$ at $z=0$.
The second variation of \cref{eq:sketch-purification} is a positive mixture of quantum Fisher information-type metrics $Q_t(\sigma;X)\EqDef\Tr[X^\dagger(L_\sigma+tR_\sigma)^{-1}(X)]$:
\begin{flalign}
    &\int_0^\infty w_{p,q}(t)\,\Delta_t(X_O)\,dt=0\\
    &\Delta_t(X)\EqDef Q_t(\omega;X)-Q_t\big(\eta;\Phi(X)\big)
\end{flalign}
where $w_{p,q}(t)>0$ for every $t>0$.
Each $\Delta_t$ is nonnegative by monotonicity~\cite{lesniewski1999monotone}, so $\Delta_t(X_O)=0$ for all $t$.
Integrating against $t^{-1/2}$ yields equality in the $\chi^2$-type quantity $\norm{\omega^{-1/4}X_O\,\omega^{-1/4}}_2^2$, and the recovery bound of \cRef{gao2024sufficient} then gives $\mcP_{\omega,\Phi}(\Phi(X_O))=X_O$.
\end{proof}

As the nonlinear correlators are closely related to the $\alpha$-$z$-Rényi divergence, it is natural to expect that \cref{thm:local-comp} (\ref{bul2:QMC-operational}) can be understood via the sufficiency \cite{hiai2024alpha} of the $\alpha$-$z$-Rényi divergence. Indeed, this turns out to be the case: we give such an alternative proof in Appendix~\ref{app:non-linear-alternative}, provided that $(p,q)\neq(1,1)$ and $\rho_{ABC}$ has full rank.

For what follows, we need an approximate version of \cref{thm:local-comp} (\ref{bul1:QMC-operational}), that approximate Markov chains implies approximate local computability.
We defer the proof to Appendix~\ref{app:non-lin} for conciseness.
\begin{thmprime}{thm:local-comp}[1, approximate version] \label{thm:local-comp-1}
For any tripartite state $\rho_{ABC}$, assuming $p\leq q$ without loss of generality (WLOG), we have:
    \begin{align} \label{eq:localcompute3}
        0\leq \mcC_{p,q}(O_A,\rho_{AB})-\mcC_{p,q}(O_A,\rho_{ABC})
        \lesssim \norm{O_A}_\infty\, \condinf{A}{C}{B}^{\frac p2}.
    \end{align}
\end{thmprime}


It is also natural to consider quantitative versions of \cref{thm:local-comp} (\ref{bul2:QMC-operational}), relating approximate local computability to a small conditional mutual information.
We expect that the proofs above provide a route to such estimates: at $p=q=1$, one can use quantitative recovery bounds based on hypothesis testing~\cite{jencova2024hypothesis}; for the remaining parameters, one may instead use the recovery bound of~\cite{gao2024sufficient}.
A direct implementation, however, may lead to unfavorable dependence on subsystem dimensions and small eigenvalues, as well as losses in the error exponent.
We leave the development of sharper quantitative bounds for future work.

\subsection{Decay of nonlinear correlation}\label{sec:ncordecay}

Our main result in this section is that for SRC states, the connected correlators defined by the above observables always exhibit clustering.

\begin{theorem}[SRC $\Rightarrow $ $\mcC_{p,q}$ clustering]\label{thm:clusterCpq}
    Let $0<p,q\leq 1$ and assume WLOG $p\le q$.
    Let $\rho$ be an SRC state on a $D$-dimensional lattice, and let $O_1,O_2$ be any operators supported in small regions $A_1,A_2$ of $O(1)$ sites.
    Then
    \begin{equation} \label{eq:thm-corr}
        |\mcC_{p,q}(O_1O_2,\rho)-\mcC_{p,q}(O_1,\rho)\mcC_{p,q}(O_2,\rho)|
        \approx 0,
    \end{equation}
    where the error term is upper bounded by $ C\norm{O_1 }_\infty\norm{O_2}_\infty e^{-p\Big(\frac{\dist(A_1,A_2)}{\xi}\Big)^{\frac 1D}}$.
    Here, $\xi$ depends on $\eta,\zeta,D$, and the constant $C$ depends on $|A_1|,|A_2|$.
\end{theorem}


The $1/D$ exponent in the exponential arises from the gap between the correlation bound in \Def{def:gapped} and the trace-norm estimate. 
If the SRC state $\rho$ enjoys a favorable correlation scaling of the form
\begin{align}\label{eq:tracedecaypoly}
    \norm{\rho_{AC}-\rho_A\otimes \rho_C}_1 \le \poly(A)\,e^{-\frac{\dist(A,C)}{\eta}},
\end{align}
then the error term in \cref{eq:thm-corr} is improved to
\begin{align}
    C\norm{O_1 }_\infty\norm{O_2}_\infty e^{-\frac{p\dist(A_1,A_2)}{\xi}} ,
\end{align}
where $C$ depends on $A_1,A_2$ and 
$\xi$ depends on $\eta$ and $\zeta$.

    To convey the main idea of the proof, we first consider the simplified setting in which both the mutual information and the conditional mutual information vanish exactly beyond certain length scale.
    For this simplified version, the bound in \cref{eq:thm-corr} will be identically zero. 
    The proof for the general case is given in Appendix~\ref{app:non-lin}.

\begin{proof}[Proof (for vanishing MI and CMI)]

We divide the lattice into $A=A_1A_2$, $B=B_1B_2$ where $B_i$ is a buffer around $A_i$, and $C$ is the complement (see \Fig{fig:corr-proof}).
We choose $B_i$ to be large enough such that we have exact Markov chain and exact vanishing correlations.
In particular, the same state $\rho_{ABC}$ is a QMC for all three partitions $A_1-B_1-CB_2A_2$, $A_2-B_2-CB_1A_1$, and $A$-$B$-$C$ (the Markov property for the last partition follows from that for the first two).

Now we use \cref{thm:local-comp}(1) three times:
\begin{equation}\label{eq:3equations}
  \begin{aligned}
    \mcC_{p,q}(O_1,\rho_{A_1B_1}) = \mcC_{p,q}(O_1,\rho_{ABC}),\\
    \mcC_{p,q}(O_2,\rho_{A_2B_2}) = \mcC_{p,q}(O_2,\rho_{ABC}),\\
    \mcC_{p,q}(O_1O_2,\rho_{AB}) = \mcC_{p,q}(O_1O_2,\rho_{ABC}).\\
\end{aligned}   
\end{equation}
Here, $O_1$ and $O_2$ are supported on $A_1$ and $A_2$ respectively.

Next, due to the vanishing of correlation, we have $\rho_{AB}=
\rho_{A_1 B_1}\otimes \rho_{A_2B_2}$ (see \cref{eq:first-equivalence}), implying tensorization of the 
correlations by multiplicativity (\cref{prop5:bul-mult} in Proposition~\ref{prop:Cpqbasic})
\begin{align} \label{eq:tensorization}
    \mcC_{p,q}(O_1O_2,\rho_{AB})=\mcC_{p,q}(O_1,\rho_{A_1B_1})\cdot\mcC_{p,q}(O_2,\rho_{A_2B_2}) .
\end{align}

Combining \cref{eq:3equations,eq:tensorization}, we then conclude \cref{eq:thm-corr} with zero on the r.h.s.. 
\end{proof}

\begin{figure}
    \centering
    \begin{tikzpicture}[scale=0.6]

\def\rA{0.5}
\def\rB{2}

\definecolor{Acol}{RGB}{110,135,185}
\definecolor{Bcol}{RGB}{175,195,230}
\definecolor{Outline}{RGB}{55,75,120}

\foreach \x in {0,1,...,9} {
  \foreach \y in {0,1,...,9} {
    \fill[black] (\x,\y) circle (0.055);
  }
}

\fill[Bcol, opacity=0.78] (2,2) circle[radius=\rB];
\fill[Acol, opacity=0.95] (2,2) circle[radius=\rA];
\draw[Outline, thick] (2,2) circle[radius=\rB];
\draw[Outline, thick] (2,2) circle[radius=\rA];

\fill[Bcol, opacity=0.78] (6.9,7.0) circle[radius=\rB];
\fill[Acol, opacity=0.95] (6.9,7.0) circle[radius=\rA];
\draw[Outline, thick] (6.9,7.0) circle[radius=\rB];
\draw[Outline, thick] (6.9,7.0) circle[radius=\rA];

\node at (2,2) {\large $A_1$};
\node at (6.9,7.0) {\large $A_2$};

\node at (0.85,2.75) {\large $B_1$};
\node at (5.55,7.45) {\large $B_2$};

\node at (1.45,7.25) {\large $C$};

\end{tikzpicture}
    \caption{Lattice partition in \cref{thm:clusterCpq}.}
    \label{fig:corr-proof}
\end{figure}


{~}

In the context of strong-to-weak symmetry breaking (SWSSB) \cite{ref:swssb1,ref:lessa2025fidelity,ref:Sala2024Renyi}, the operators $O_i$ are chosen to be charged under a symmetry, forcing the 1-point correlator $\mathrm F _p(O_i,\rho)$ to vanish automatically when $\rho$ is strongly symmetric (that is, $S\rho\propto \rho$ where $S$ is the symmetry).
Therefore, the vanishing of \emph{connected} correlator is equivalent to $\mathrm F_p(O_1O_2,\rho)=0$.
In standard ground-state physics, however, the concept of short-range correlation is defined independently of any underlying symmetry, where connected correlators are the relevant quantities to consider. 
Our approach here extends such a generalization to nonlinear correlators in the context of mixed-state physics.

While $\mcC_{p,q}$ already encloses most previously discussed non-linear observables,
one notable exception is the Rényi-2 correlator~\cite{ref:swssb1,ref:Sala2024Renyi}, defined as:
\begin{equation}
    \frac{\Tr(O\rho O^\dagger \rho)}{\Tr(\rho^2)}=\frac{\norm{\rho^{\frac{1}{2}}O\rho^{\frac{1}{2}}}_2^2}{\norm{\rho}_2^2}.
\end{equation}
This is outside our range of $(p,q,s)$ and in fact the function $\Tr(O\rho O^\dagger \rho)$ is not concave or convex in $\rho$.
Nevertheless, note that $\norm{\rho^{\frac{1}{2}}O\rho^{\frac{1}{2}}}_2\leq \norm{\rho^{\frac{1}{2}}O\rho^{\frac{1}{2}}}_1 = F_1$, hence the numerator must exhibit clustering if $\rho$ is a strongly symmetric SRC state and $O$ is charged, as in the SWSSB scenario.

\section{Canonical Purification} 
\label{sec:Purification}

In this section we study the correlation properties of the canonical purification, also dubbed the thermofield double (TFD) state, of a state $\rho$, defined by
\begin{equation}
    \ket{\sqrt{\rho}} = \sum_{i=1}^N\sqrt \rho \otimes \Id \ket{i,i} ,
\end{equation}
Here 
$\{\ket{i}\}$ is the computational basis. 
We discuss the relation between $\rho$ being SRC and the decay of correlations in $\ket{\sqrt\rho}$. 
Throughout this section, we denote regions in the auxiliary (purifying) system with a bar, e.g., $\bar A$ is the region in the auxiliary system corresponding to $A$ in the original system.

\subsection{Decay of correlation}\label{sec:TFD-cluster}
Our main results in this section are that $\rho$ being QMC implies local marginals of $\ket{\sqrt{\rho}}$ are locally computable, and that $\rho$ being SRC implies decay of correlations in $\ket{\sqrt\rho}$.
These are formulated in the following \cref{thm:TFD-LC} and \cref{thm:canonical-clustering}, which are proven in Appendix~\ref{app:purifications}.

\begin{theorem}[Local computability of  $\ket{\sqrt\rho}$]\label{thm:TFD-LC}
For any tripartite state $\rho_{ABC}$: 
\begin{align}
&\fid{\Tr_{B\bar B C\bar C}\ketbra{\sqrt{\rho_{ABC}}}}{\Tr_{B\bar B}\ketbra{\sqrt{\rho_{AB}}}} \notag\\
&\ge 2^{-\frac12\condinf ACB_\rho}.
\end{align}
\end{theorem}

This theorem is closely related to \cref{thm:local-comp} with $p=q=1/2$.
In fact, for operators of the form $O_A\otimes O^*_{\bar A}$, we have: 
\begin{equation} \label{eq:CP-Cpq}
        \mcC_{\frac{1}{2},\frac{1}{2}}(O,\rho)
        =\left[\Tr(\sqrt{\rho}O\sqrt{\rho}O^\dagger)\right]^{\frac{1}{2}}
        = \expval{O\otimes O^*}{\sqrt{\rho}}^{\frac{1}{2}}.
    \end{equation}
For states with small CMI, the approximated local computability \cref{thm:local-comp-1} asserts that
\begin{equation}
        \mcC_{\frac12,\frac12}(O_A,\rho_{AB}) \approx \mcC_{\frac12,\frac12}(O_A,\rho_{ABC}).
\end{equation}
Therefore,
\begin{equation}
\expval{O_A\otimes O^*_{\bar A}}_{\sqrt{\rho_{ABC}}} \approx \expval{O_A\otimes O^*_{\bar A}}_{\sqrt{\rho_{AB}}};
\end{equation}
namely, the $A\bar A$ marginals of $\ket{\sqrt{\rho_{ABC}}}$ and $\ket{\sqrt{\rho_{AB}}}$ are close to each other measured by $O_A\otimes O^*_{\bar A}$, which can be further extended to all operators acting on both the physical and purification systems.
However, \cref{thm:TFD-LC} is still not a simple corollary of \cref{thm:local-comp}, since the above argument can only yield the 2-norm closeness of two marginals, rather than the closeness in fidelity\footnote{See version 1 of this manuscript.}.
Instead, we use an approximate recovery result for fractional powers of the density operators \cite{JungeLaRacuente2022}.

\Cref{thm:TFD-LC} has two interesting corollaries. The first is the clustering property of $\ket{\sqrt\rho}$ for a SRC state $\rho$.

\begin{corollary}[Clustering in $\ket{\sqrt\rho}$]
\label{thm:canonical-clustering}
    Let $\rho$ be an SRC state on a $D$-dimensional lattice.
    Then its canonical purification $\ket{\sqrt{\rho}}$ exhibits decay of correlations of the form 
    \begin{equation} \label{eq:thm-canonical}
        \expval{O_{A_1\bar A_1}O_{A_2\bar A_2}}_{\sqrt{\rho}} -        \expval{O_{A_1\bar {A_1}}}_{\sqrt{\rho}}                 \expval{O_{A_2\bar A_2}}_{\sqrt{\rho}} 
        \approx 0,
    \end{equation}
    where the error term is upper bounded by 
        $C\norm{O_{A_1\bar A_1}}_\infty \norm{O_{A_2\bar A_2}}_\infty e^{-\Big(\frac{\dist(A_1,A_2)}{\xi}\Big)^{\frac 1D}}$.
\end{corollary}



\cref{thm:canonical-clustering} is similar to \cref{thm:clusterCpq} with the sole difference that it uses the new local computability result for the TFD instead of \cref{thm:local-comp}.
Therefore, the error bound of \cref{eq:thm-canonical} behaves as that in \cref{thm:clusterCpq}, where $\xi$ depends on $\eta,\zeta,D$, while the constant $C$ depends on $|A_1|,|A_2|$.
If, in addition, the SRC state $\rho$ satisfies \cref{eq:tracedecaypoly}, then the error term can be improved to an exponential decay in $\dist(A_1,A_2)$.

In particular, since thermal Gibbs states of local Hamiltonians are always approximately Markovian \cite{ref:Poulin2012markov,ref:Kuwahara2025CMI,ref:chen2025markovian}, our result implies, in essence, that Gibbs states are short-range correlated if and only if their canonical purifications (the thermofield double) are short-range correlated. 

\subsection{Mixed-state phases via canonical purification}
\Cref{thm:canonical-clustering} also shows that, if we have a continuous family of SRC mixed states, then the corresponding family of canonically purified states always exhibit decay of correlations.
One physical interpretation is that a phase transition of the canonical purification always implies a phase transition in the mixed state: the ordinary correlation length or the Markov length must diverge.
In other words, the classification of (mixed or pure state) phases based on canonical purification is at least as coarse as the classification based on MI and CMI.
This partially justifies previous studies using canonical purification to understand mixed state phases \cite{lin2021entanglement, PhysRevB.94.155125, lu2024bilayerconstructionmixedstate}.

Such physical interpretation can be formalized using the following corollary of \cref{thm:TFD-LC}.

\begin{corollary}[channel-unitary correspondence]\label{prop-circuitU}
    Let $\rho_{ABC}$ be a tripartite state, and let $\Phi$ be a $C\to C$ quantum channel.
    There exists a unitary $U$ on $B\bar B C\bar C$ such that:
    \begin{equation}\label{conj-circuit}
        \begin{aligned}
            &\norm{\ket{\sqrt{\Phi(\rho_{ABC})}}-U\ket{\sqrt{\rho_{ABC}}}}\\
            \qquad\le &\sqrt{(\ln 2)\bigl(\condinf ACB_\rho-\condinf ACB_{\Phi(\rho)}\bigr)}.
        \end{aligned}
    \end{equation}
\end{corollary}

\begin{proof}
Let us compare the reduced states of $\ket{\sqrt{\rho_{ABC}}}$ and $\ket{\sqrt{\Phi(\rho_{ABC})}}$ on $A\bar A$, denoted as $\gamma_{A\bar A}$ and $\gamma^\Phi_{A\bar A}$.
Consider a Stinespring isometry $V:C\to CE$ for $\Phi$, and denote $\omega_{ABCE}=V\rho_{ABC}V^\dagger$, so that $\Tr_E\omega=\Phi(\rho)$.
The chain rule gives:
\begin{equation}
    \begin{aligned}
        \condinf A E {BC}_\omega
        &=\condinf A {CE} B_\omega-\condinf A C B_\omega\\
        &=\condinf ACB_\rho-\condinf ACB_{\Phi(\rho)}.
    \end{aligned}
\end{equation}
The $A\bar A$ marginal of $\ket{\sqrt{\omega}}$ is $\gamma_{A\bar A}$, since the isometry acts as $V\otimes V^*$ on the canonical purification.
Applying \cref{thm:TFD-LC} to $\omega$ with the tripartition $A$-$BC$-$E$, and using $D_B^2=2-2F$ together with $1-e^{-x}\le x$, we get:
\begin{equation}
    \begin{aligned}
        \bur{\gamma^\Phi_{A\bar A}}{\gamma_{A\bar A}}
        &\le \sqrt{2-2\cdot 2^{-\frac12 \condinf A E {BC}_\omega }}\\
        &\le \sqrt{(\ln 2)\bigl(\condinf ACB_\rho-\condinf ACB_{\Phi(\rho)}\bigr)}.
    \end{aligned}
\end{equation}
Uhlmann's theorem then implies there exists a unitary $U$ on $B\bar B C\bar C$ such that:
\begin{equation}
    \norm{\ket{\sqrt{\Phi(\rho_{ABC})}}-U\ket{\sqrt{\rho_{ABC}}}}
    =\bur{\gamma^\Phi_{A\bar A}}{\gamma_{A\bar A}},
\end{equation}
which proves \cref{conj-circuit}.
\end{proof}

Consequently, given a finite-depth circuit of locally reversible channels, we expect that the above result can be applied repeatedly, gate by gate (see, e.g., Ref.~\cite{ref:Sang2025stability}), to construct a finite-depth unitary circuit. This would strengthen the correspondence between mixed-state phases and pure-state phases via canonical purification\footnote{We thank Chong Wang for suggesting this correspondence between locally reversible circuits and local unitary circuits on the canonical purification.  See upcoming work \cite{colleagues} in this direction.}.


On the other hand, Ref.~\cite{chiralpaper} shows that there exist two mixed states that are in different mixed-state phases, yet their canonical purifications are in the same pure state phase.
Therefore, the classification based on canonical purification is strictly more coarse than the mixed state classification.

\subsection{On the converse implication}
It is natural to ask whether the converse of \cref{thm:canonical-clustering} holds: do short-range correlations in the canonical purification imply an SRC mixed state (\Def{def:gapped})? 
Note that the decay of correlation is obvious, so the non-trivial property to be established is the decay of CMI.

In its most general formulation, the answer is negative. Ref.~\cite{negari2026instability} shows an example where the mixed state has algebraic decaying CMI while the canonical purification is short range correlated.
(We also note that for a tripartite state $\rho_{ABC}$, the vanishing of $A\bar A$-$C\bar C$ correlations in the canonical purification does not imply the vanishing of $\condinf{A}{C}{B}$ in the original state.
See Appendix~\ref{sec:counteregTFD} for an explicit example.)

However, the answer is positive under under certain conditions. 
\begin{theorem}[Clustering in $\ket{\sqrt{\rho}}$ $\Rightarrow$ SRC]
    (1) For a stabilizer (mixed) state\footnote{Here the stabilizer mixed state for a stabilizer group $S$ is the maximally mixed state in the stabilizer subspace. We do not assume $S$ is locally generated, otherwise $\rho$ is already a Markov network \cite{ref:Poulin2012markov} and there is nothing to prove.} $\rho$, if its canonical purification obeys decay of MI, then $\rho$ has decay of CMI.
    
    (2) For a 1D classical (mixed) state $\rho$, if its canonical purification has exact zero correlation length, i.e., $\mutinf{A\bar A}{C\bar C}_{\ket{\sqrt{\rho}}}=0$ for any $A$ and $C$ that are separated by at least one site, then $\rho$ has superpolynomial decay of CMI.
\end{theorem}

\begin{proof}
    (1) We first claim that for a bipartite stabilizer state $\rho_{AB}$, we have
    \begin{equation}\label{eq:stabnogap}
        \mutinf AB_\rho = S(A\Bar{A})_{\ket{\sqrt\rho}}.
    \end{equation}
    To show it, we note that \cRef{bravyi2006ghz} implies that $\rho_{AB}$ is local-Clifford equivalent to the tensor product of maximally entangled states, 
    maximally correlated classical states ($\frac{1}{2}(\ketbra{00}+\ketbra{11})$), and decoupled states on $A$ and $B$.
    Each state above contributes equally to both sides of \cref{eq:stabnogap}, hence \cref{eq:stabnogap} holds.    

    Applying \cref{eq:stabnogap} three times on $A|BC$, $C|AB$ and $B|AC$, a straightforward calculation shows that for any tripartite stabilizer state $\rho_{ABC}$, we have:
    \begin{equation}
        \mutinf{A}{C}_\rho+\condinf{A}{C}{B}_\rho=\mutinf{A\Bar{A}}{C\Bar{C}}_{\ket{\sqrt\rho}}.
    \end{equation}
    Therefore, for stabilizer states, decay of MI in canonical purification implies decay of CMI in the original mixed state.

    (2) Consider a classical state $\rho=\sum_s p_s\ketbra{s}$ where $s$ denotes bit strings.
    Its canonical purification is $\ket{\sqrt{\rho}} = \sum_s \sqrt{p_s}\ket{s\bar{s}}$.
    We note that, for this state in the computational basis, all amplitudes are nonnegative.
    
    Let us measure $B\bar{B}$ in the computational basis.
    Each measurement outcome has the form of $s_B s_{\bar B}$.
    We have:
    \begin{equation}
        \prob{s_B{s_{\bar B}},\ket{\sqrt{\rho}}} = \sum_{s_A,s_C}p_{s_As_Bs_C} = \prob{s_B,\rho}.
    \end{equation}
    Here the l.h.s is the probability of getting $s_B s_{\bar B}$ upon measuring $B\bar{B}$ on $\ket{\sqrt{\rho}}$ while the r.h.s. is the probability of getting $s_B$ upon measuring $B$ on $\rho$.
    For each measurement outcome, we denote the after-measurement pure state on $A\bar{A}C\bar{C}$ as $\ket{\psi(s_B{s_{\bar B }})}$.
    It is exactly the canonical purification of the conditional classical state $P_{s_B}\rho P_{s_B}/\prob{s_B,\rho}$.
    It follows from monotonicity of MI that \cite{dutta2021canonical}:
    \begin{equation}
        \condinf{A}{C}{s_B} \leq \ent{A\bar{A}}_{\ket{\psi(s_B s_{\bar B})}},
    \end{equation}
    where the l.h.s. is the classical mutual information of $\rho$ conditioned on $B$ being $s_B$.
    Therefore,
    \begin{equation}
    \begin{aligned}
        \condinf{A}{C}{B}_\rho 
        &= \sum_{s_B} \prob{s_B,\rho} \condinf{A}{C}{s_B}\\
        &\leq \sum_{s_B} \prob{s_B s_{\bar B},\ket{\sqrt{\rho}}} \ent{A\bar{A}}_{\ket{\psi(s_B{s_{\bar B}})}}.
    \end{aligned}
    \end{equation}
    The r.h.s. is the average entanglement entropy between $A\bar{A}$ and $C\bar{C}$ in the post-measurement states after measuring $B\bar{B}$ in the computational basis.
    It is dubbed measurement induced entanglement (MIE) \cite{hastings2016quantum,lin2023probing}.

    Ref.~\cite{hastings2016quantum} (proposition 3.1) shows that if a 1D pure state with nonnegative amplitudes has exact zero correlation length, then its MIE has super-polynomial decay.
    Therefore, in our scenario, $\condinf{A}{C}{B}_\rho$ must decay super-polynomially.
\end{proof}



\section{Discussion and Outlook}
In this work, we have established a unified framework to characterize stable phases of matter, encompassing both pure and mixed states, in or out of equilibrium. 
We proposed that a phase is ``locally stable" if any local operation can be reversed by a local quantum channel, and we have shown that this operational property is equivalent to the static property of short-range correlations as defined by exponential decay of MI and CMI.  These properties serve as robust, invariant characteristics of a phase. 
Furthermore, we demonstrated that local stability directly implies the clustering of a broad family of non-linear correlators, independent of symmetry considerations. 
In particular, an important result of our work relates mixed and pure states: we proved that a mixed state is locally stable only if its canonical purification exhibits short-range correlations.  As a byproduct, we found that local stability provides a sufficient condition for local post-selection to be implemented locally by a quantum channel, and we also found an equivalent characterization of quantum Markov chains in terms of local computability.

Aside from the technical open questions already mentioned in the text, our work motivates the following questions.

 

\textbf{Lindbladians and locally stable steady states.}
There remain many further questions regarding relating dynamics to correlation properties.  For example, given a locally stable state, can one construct a parent local Lindbladian with the target state as a steady state?  Does the parent Lindbladian always serve as an efficient recovery against local operations?  If so, then a single trajectory of many measurements, with each disturbance dissipated by Lindbladian dynamics, could be used to probe the state \cite{ref:chen2025metastability,jiang2026predictingpropertiesquantumthermal}.  Furthermore, what are necessary and sufficient conditions on a local Lindbladian for its steady state to be locally stable?  And does such a Lindbladian always rapidly dissipate local perturbations to the steady state?  

It is also interesting to understand the role of the spectral gap in these questions. 
In \cref{sec:dynamics}, we noted that although Proposition \ref{clm:Var-conv} holds in general, initial states arising from local perturbations to a steady state may still remain close to the steady state. 
Along similar lines, for such initial states, a suitable notion of local gap may be more relevant than the global spectral gap. 
We leave a further exploration of this question and its connection to Refs.~\cite{ref:Lucas2025metstablity,ref:chen2025markovian} for future work.

\textbf{Correlation and entanglement in canonical purification.} 
In \cref{thm:canonical-clustering}, we established that the canonical purification of SRC states exhibits clustering as measured by the two-point correlation function.
It is natural to ask whether this clustering property can be strengthened.
In fact, the same argument used for \cref{thm:canonical-clustering} would also yield decay of local-global correlations.
However, for global-global correlations, one would need a suitable strengthening of \cref{def:gapped} for our proofs to apply.
It still remains to see whether such strengthening is necessary, and whether it holds in physically relevant examples.

An even stronger formulation of the question is whether $\ket{\sqrt\rho}$ inherits decay of mutual information whenever $\rho$ itself has decaying MI and CMI. 
More generally, it is interesting to see if one can bound the entropic quantities of $\ket{\sqrt\rho}$ in terms of those of $\rho$.


\begin{acknowledgments} 
We thank Itai Arad, Liang Fu, Xiaoyang Huang, Leonardo Lessa, Subhayan Sahu, Shengqi Sang, David Sutter, Chong Wang, and Yijian Zou for many helpful discussions.  
After posting version 1, we used AI tools to revise our results: GPT-5.6 suggested Ref.~\cite{JungeLaRacuente2022} for \cref{thm:TFD-LC}, and GPT-6 found the proof of \cref{thm:local-comp} (\ref{bul2:QMC-operational})(case 2). 
Research at Perimeter Institute is supported in part by the Government of Canada through the Department of Innovation, Science and Industry Canada and by the Province of Ontario through the Ministry of Colleges and Universities.  

{\it Note added}: While this work was completed, we became aware of a concurrent independent work \cite{chennote}, which has overlap with the connection between the local stability and short-range correlations properties (\cref{sec:equiv}).
We would also like to draw attention to an upcoming independent work \cite{colleagues}, which has overlap with our \cref{thm:TFD-LC} and \cref{prop-circuitU},
and pursues the locally reversible channel / local unitary connection more extensively.

\end{acknowledgments}

\bibliography{Refs}

\onecolumngrid
\appendix

\section{Notations and Useful Facts} \label{app:settings}

\newcommand{\parasec}[1]{\medskip\noindent\textbf{#1.}\enspace\smallskip}


In this appendix, we set out the conventions, notations, and mathematical facts used in the main text. Further mathematical details can be found in standard textbooks, e.g., Refs.~\cite{ref:wilde2013book,ref:Watrous2018book}. 

\parasec{Setup and notations}

We consider $d$-dimensional qudits on a lattice $\Lambda\subset \mathbb{Z}^D$.
Each site $i\in\Lambda$ carries a local Hilbert space $\BBH_i\cong \BBC^d$; for a region $X\subseteq \Lambda$, 
$\BBH_X = \bigotimes_{i\in X}\BBH_i$ with $d_X=\dim(\BBH_X) = d^{|X|}$.

The space of operators of $\BBH\to \BBH$ is denoted by $\LH$.
We only consider finite dimensional spaces, so all operators are bounded.
The space of (mixed) states in a Hilbert space $\BBH$ is denoted by $\mcS(\BBH)$.
Reduced states (partial traces) will be denoted by subscripts, e.g., $\rho_A$.

We use the standard big-O notation $\bigO{\;}$.
We also denote $x\lesssim y$ if $x=\bigO{y}$.

\parasec{Quantum Operations}

Throughout this work, we refer to completely positive trace-non-increasing maps as quantum operations, and refer to trace-preserving quantum operations as quantum channels.
Any quantum operation can be (non-uniquely) decomposed into the Kraus form:
\begin{align}
\Phi(\cdot)=\sum_k F_k(\cdot)F_k^\dagger,~~~~\sum_k F_k^\dagger F_k \leq \Id.
\end{align}
It is a quantum channel if and only if $\sum_k F_k^\dagger F_k = \Id$.

\parasec{Operators and norms}
 
We consider the family of norms on $\LH$ known as Schatten norms, labeled by an index $p\ge 1$ (well-defined for any $p>0$ but we need $p\ge 1$ to fullfill the triangle inequality), defined as
$
    \norm{A}_p = \big( \Tr{|A|^p}\big) ^{1/p},
$
where $|A|= \sqrt{A^\dagger A}$.
The $p=\infty$ case (maximal singular value) is usually referred to as operator norm.
The $p=1$ norm of $A$ (sum of all singular values) is usually referred to as the trace norm, and
admits the following duality form $\norm{A}_1=\sup_{\norm {E}_\infty=1}\Tr{EA}$.

We often make use of the H\"older inequality: for $\frac{1}{p}+\frac 1 q = \frac 1 r$ and $p,q,r>0$:
\begin{align} \label{eq:Holder}
    \norm{A\cdot B}_r \le \norm{A}_p\norm{B}_q.
\end{align}

The Bures distance between states $\rho,\sigma$ is defined as:
\begin{equation}
    \bur{\rho}{\sigma} =\min_{u\in\mathrm{unitary}}\|\sqrt\rho - u\sqrt\sigma\|_2.
\end{equation}
It is related to the (unsquared) Uhlmann fidelity $F(\rho,\sigma)\EqDef\|\sqrt\rho\sqrt\sigma\|_1$ by $D_B^2 = 2-2F$.
It is related to the trace norm by the Fuchs–van de Graaf inequality 
\begin{equation}\label{eq:FvG}
    1-F(\rho,\sigma) \le \frac{1}{2}\|\rho-\sigma\|_1 \le \sqrt{1-F(\rho,\sigma)^2},
\end{equation}
which implies
\begin{equation}
    \frac{1}{2}\|\rho-\sigma\|_1\leq  D_B(\rho,\sigma)\le \|\rho-\sigma\|_1^{\frac12}.
\end{equation}
It is related to the 2-norm by:
\begin{equation}\label{eq:DBvs2}
    \bur{\rho}{\sigma}\leq \norm{\sqrt\rho-\sqrt\sigma}_2 \leq \sqrt{2}\bur{\rho}{\sigma}.
\end{equation}
Here, the second inequality follows from Eq.(VII.39) in \cRef{ref:Bhatia-book} (that $ \norm{|A|-|B|}_2 \leq \sqrt{2}\norm{A-B}_2$    for arbitrary $A$ and $B$).

For two pure states $\ket{\psi}$ and $\ket{\phi}$, we use $\norm{\ket{\psi}-\ket{\phi}}$ to denote their distance in the Hilbert space.

\parasec{Entropic quantities}

We make use of standard information measures.
The von Neumann entropy is $\mathrm{S}(\rho) = -\mathrm{Tr}(\rho\log\rho)$ (throughout the work, the base of $\log$ is 2).
The mutual information (MI) and conditional mutual information (CMI) for a state $\rho$ on $ABC$ are defined as:
\begin{align}
    \mutinf AB_\rho &= \mathrm{S}(A) + \mathrm{S}(B) - \mathrm{S}(AB), \\
    \condinf ACB_\rho &= \mathrm{S}(AB) + \mathrm{S}(BC) - \mathrm{S}(ABC) - \mathrm{S}(B).
\end{align}
Both are non-negative by the subadditivity and the strong subadditivity~\cite{ref:Watrous2018book}.
A useful identity is $\condinf ACB=\mutinf {AB}{C}-\mutinf BC$.

We will also deal with a two-parameter family of relative entropy called $\alpha$-$z$ R\'enyi divergence~\cite{jaksic2011entropic,ref:audenaert2015alpha}.
For $\alpha\in(0,1)\cup(1,\infty)$ and $z>0$, it is defined as:
\begin{equation}\label{eq:alpha-z-def}
    \mathrm{D}_{\alpha,z}(\rho\|\sigma) \EqDef \frac{1}{\alpha-1}\log\mathrm{Tr}\!\left[\left(\sigma^{\frac{1-\alpha}{2z}}\rho^{\frac{\alpha}{z}}\sigma^{\frac{1-\alpha}{2z}}\right)^{\!z}\right].
\end{equation}
Setting $z=1$ gives the Petz-type R\'enyi divergence $\mathrm D_\alpha (\rho\|\sigma) \EqDef\frac{1}{\alpha-1} \log \Tr{\rho^\alpha \sigma^{1-\alpha}}$; 
setting $z=\alpha$ gives the sandwiched R\'enyi divergence $\tilde{\mathrm D}_\alpha (\rho\|\sigma) \EqDef\frac{1}{\alpha-1} \log \Tr{\big(\sigma^{\frac{1-\alpha}{2\alpha}} \rho \sigma^{\frac{1-\alpha}{2\alpha}}
\big)^{\alpha}}$; 
setting $\alpha\to 1$ recovers the Umegaki relative entropy for all $z$: $\mathrm{D}(\rho\|\sigma)=\Tr(\rho\log\rho-\rho\log\sigma)$.

\parasec{Quantum Markov chains and sufficiency}

A state $\rho_{ABC}$ is an exact $A$-$B$-$C$ quantum Markov chain (QMC) if one of the following mutually equivalent conditions holds \cite{petz1986sufficient}:
\begin{itemize}
    \item $\mathrm{I}(A:C|B)_\rho = 0$;
    \item there exists a quantum channel $\mcR:B\to BC$ such that $\mcR(\rho_{AB})=\rho_{ABC}$;
    \item $\mcP_{\rho_{BC},\mathrm{Tr}_C}(\rho_{AB})=\rho_{ABC}$, where $\mcP_{\rho_{BC},\mathrm{Tr}_C}$ is the Petz recovery map associated with $\rho_{BC}$ and $\mathrm{Tr}_C$ (here the inverse is taken on $\supp(\rho_{B})$):
\begin{equation}\label{eq:Petz-def}
    \mcP_{\rho_{BC},\mathrm{Tr}_C}(\cdot) = \rho_{BC}^{\frac12}\bigl(\rho_{B}^{-\frac12}(\cdot)\rho_{B}^{-\frac12}\otimes \Id_C\bigr)\rho_{BC}^{\frac12}.
\end{equation}
\end{itemize}

 
More generally, a channel $\Phi$ is called sufficient for a pair of states $(\rho,\sigma)$ if there exists a quantum channel $\mcR$ such that $\mcR\circ\Phi(\rho)=\rho$ and $\mcR\circ\Phi(\sigma)=\sigma$.
For example, $\rho_{ABC}$ being a QMC is equivalent to $\Tr_C$ being sufficient for $(\rho_{ABC},\rho_{BC})$.
Sufficiency can be characterized by the equality in the data processing inequality (DPI): 
\begin{equation}
    \mathrm{D}(\Phi(\rho)\|\Phi(\sigma))\le \mathrm{D}(\rho\|\sigma) ,\qquad\text{equality iff $\Phi$ is sufficient}.  
\end{equation}

The $\alpha$-$z$ R\'enyi divergence $\mathrm{D}_{\alpha,z}$ also satisfies DPI for appropriate choice of $(\alpha,z)$.
The equality also implies sufficiency, with the Petz recovery map providing a universal recovery channel. For a proof, see \cRef{hiai2024alpha}.

\parasec{Approximate quantum Markov chains and approximate recovery}

The above equivalence between QMC and recovery admits an approximate version.
We call $\rho_{ABC}$ an approximate QMC if the CMI $\condinf{A}{C}{B}$ is small.
We define the fidelity of recovery (FoR) as:
\begin{equation}
    \mathrm{F}(A;C|B)=\max_{\mcR_{B\to AB}}\fid{\rho_{ABC}}{\mcR(\rho_{BC})}.
\end{equation}
The following facts establish an equivalence between small CMI and large FoR:

\begin{fact}[\cRef{ref:fawzi2015approximate}]
\label{fact:CMI}
$\condinf{A}{C}{B}_\rho\leq \epsilon$ implies there is a quantum channel $\mcR:B\rightarrow AB$ such that $\sigma_{ABC}\EqDef \mcR(\rho_{BC})$ satisfies
$\fid{\rho}{\sigma}\geq 2^{-\frac{1}{2}\epsilon}$, and hence, $\frac{1}{2}\norm{\rho-\sigma}_1 \leq \bur{\rho}{\sigma}\leq  \sqrt{(\ln2)\epsilon}$.
\end{fact}
Moreover, the recovery channel $\mcR$ can always be taken as a rotated Petz map, see \cRef{ref:junge2018petz}.

\begin{fact}\label{fact:FoR->CMI}
    If there exists a quantum channel $\mcR:B\rightarrow AB$ such that $\sigma_{ABC}\EqDef \mcR(\rho_{BC})$ satisfies $\norm{\rho-\sigma}_1\leq \delta$, then
    \begin{align}
        &\condinf{A}{C}{B}_\rho \le \delta\log \min\{d_{AB},d_C\} + 2\kappa(\delta/2),\label{eq:smallCMI1}\\
        &\condinf{A}{C}{B}_\sigma \le \delta\log \min\{d_{B},d_C\}+ 2\kappa(\delta/2) \label{eq:smallCMI2}.
    \end{align}
    where $\kappa(x)=(1+x)h_2(\frac{x}{1+x})$ (here $h_2$ is the binary entropy function).
\end{fact}
We sometimes simplify the r.h.s. using $\delta\log d +2\kappa(\delta/2)\leq 5\sqrt{\delta}\log d$ for any $d\geq 2$, which follows from maximizing $\sqrt{\delta}+\frac{2\kappa(\delta/2)}{\sqrt{\delta}}$ over $\delta\in[0,2]$.

We note that, while bounds of similar flavor have appeared in previous literature (see Proposition 35 in \cite{berta2015renyi} and Eq.(10) in \cite{ref:fawzi2015approximate}), their r.h.s. only contain $d_C$.
In contrast, for our purpose, we require the $d_{AB}$ bound, as $C$ may be unbounded.

\begin{proof}[Proof of Fact~\ref{fact:FoR->CMI}]
    \begin{equation}
        \condinf{A}{C}{B}_\rho = \mutinf{AB}{C}_\rho-\mutinf{B}{C}_\rho
        =(\mutinf{AB}{C}_\rho-\mutinf{AB}{C}_\sigma)+(\mutinf{AB}{C}_\sigma-\mutinf{B}{C}_\rho).
    \end{equation}
    The second term $\leq 0$ due to the data-processing inequality.
    For the first term, the continuity bound in Ref.~\cite{shirokov2017tight} states that
    \begin{equation}
        \mutinf{AB}{C}_\rho-\mutinf{AB}{C}_\sigma \leq \delta\log \min\{d_{AB},d_C\}+2\kappa(\delta/2)
    \end{equation}
    where $\kappa(x)=(1+x)h_2(\frac{x}{1+x})$.
    This proves \cref{eq:smallCMI1}.

    Similarly,
    \begin{equation}
        \condinf{A}{C}{B}_\sigma = \mutinf{AB}{C}_\sigma-\mutinf{B}{C}_\sigma
        =(\mutinf{AB}{C}_\sigma-\mutinf{B}{C}_\rho)+(\mutinf{B}{C}_\rho-\mutinf{B}{C}_\sigma).
    \end{equation}
    The first term $\leq 0$ due to the data-processing inequality.
    The second term is bounded as \cite{shirokov2017tight}:
    \begin{equation}
        \mutinf{B}{C}_\rho-\mutinf{B}{C}_\sigma \leq \delta\log \min\{d_B,d_C\} +2\kappa(\delta/2).
    \end{equation}
    This proves \cref{eq:smallCMI2}.
\end{proof}

\section{Proofs in \Sec{sec:recovery}}
\label{app:recovery}

\subsection{Local stability from SRC}\label{app:recovery1}

We start by proving \Thm{thm:gap->LR}, which states that 
SRC states are locally stable.

\begin{proof}[Proof of \Thm{thm:gap->LR}]
Consider $B_1$ to be an annulus buffer region of radius $r_1$ around $A$, 
$B_2$ to be an annulus buffer region of radius $r_2$ around $B_1$, and denote $C\EqDef (AB_1B_2)^c$ (see \Fig{fig:Gap-recovery}).
We use decay of correlations in $\rho$ to show that for every unnormalized trajectory $F_k\rho F_k^\dagger$ 
and any observable $O\in \mathrm L (B_2 C)$, we have:
\begin{equation}
\begin{aligned}
    \left|\Tr{F_k \rho F_k^\dagger O} - \Tr{F_k \rho F_k^\dagger }\Tr{\rho O}\right|  
    &=
     \left|\Tr{F_k^\dagger F_k \rho  O} - \Tr{ F_k^\dagger F_k \rho }\Tr{\rho O}\right| \\
      &\le 
    \norm{F_k^\dagger F_k }_\infty \norm{O}_\infty \, f(A) e^{-r_1/\eta},
\end{aligned}
\end{equation}
where we used the fact that $F_k$ and $O$ are non-overlapping.
Taking the supremum over all observables $O$ with $\norm{O}_\infty=1$ we get
\begin{align}
    \norm{ \PTr{AB_1}{F_k \rho F_k^\dagger - \rho\Tr{F_k \rho F_k^\dagger }}}_1 & \le 
    \norm{F_k^\dagger F_k }_\infty f(A) \,e^{-r_1/\eta}.
\end{align}
Pulling out the probabilities $p_k= \Tr{F_k \rho F_k^\dagger }$ and summing over all $k$'s we get:
\begin{align}
    \sum_k p_k\norm{ \PTr{AB_1}{\frac{1}{p_k}F_k \rho F_k^\dagger - \rho}}_1 & \le 
    \sum_k \norm{F_k^\dagger F_k }_\infty f(A) e^{-r_1/\eta}
    \le d_A f(A) e^{-r_1/\eta},
\end{align}
where in the last move we used $\sum\norm{F_k^\dagger F_k}\le \sum\Tr{F_k^\dagger F_k}= \Tr{\Id_A}=d_A$.
Here the trace is over $\mathbb{H}_A$.

We finish by applying the recovery map from Fact~\ref{fact:CMI}:
$\mcP:B_2\to AB_1 B_2$ such that
$\norm{\mcP(\rho_{B_2 C})-\rho}_1\le \sqrt{(4\ln 2)\condinf{AB_1}{C}{B_2}} \le \sqrt {(4\ln 2)g(A B_1)}e^{-r_2/2\zeta}$.
This implies 
\begin{equation}
    \begin{aligned}
    \sum_k p_k\norm{ \mcP\circ \PTr{AB_1}{\frac{1}{p_k}F_k \rho F_k^\dagger} - \rho}_1 
    \le& 
    \sum_k p_k\left(\norm{ \mcP\circ \PTr{AB_1}{\frac{1}{p_k}F_k \rho F_k^\dagger- \rho} }_1
    + \norm{\mcP(\rho_{B_2C})- \rho}_1\right) \\
    \le & d_A f(A) e^{-r_1/\eta} + \sqrt {(4\ln 2)g(A B_1)}e^{-r_2/2\zeta},
\end{aligned}
\end{equation}
where the first inequality is triangle inequality, and the second is data processing inequality with $\mcP$. 
\end{proof}

\subsection{Recovery from dynamics}\label{app:detectRecover}

Here, we present some of the information omitted from \cref{sec:dynamics} on Lindbladian-based recovery.
For this entire section, we assume that $\mcL$ is a Lindbladian satisfying a detailed balance condition with respect to a unique, full-rank steady state $\rho$, having a spectral gap $\gamma>0$. 
Recall the versions of detailed balance condition of focus are:
\begin{align}  \label{def:detailed-balance2}
    \mcL\circ \Gamma = \Gamma \circ \mcL^\dagger, \qquad \Gamma (A) = \begin{cases}
        A  \rho  & \mathrm{GNS}, \\
        \rho^{\frac 12}A \rho^{\frac 12} & \mathrm{KMS} .
    \end{cases}
\end{align}
Here $\mcL^\dagger$ is the (superoperator) adjoint of $\mcL$, defined as $\Tr(A^\dagger \mcL^\dagger (B)) = \Tr((\mcL(A))^\dagger B)$.

First, let us prove proposition \ref{clm:Var-conv}, which states the mixing bound
$\norm{e^{\mcL t}(\tilde \rho)-\rho}_1\le C(\tilde\rho,\rho)e^{-\gamma t }$ 
where $C$ is related to the Rényi-2 divergence or its sandwiched 
analogue, depending on the choice of detailed balance.
This exploits the operational significance of the Lindbladian gap $\gamma$ as a (non-commutative) variance decay exponent~\cite{ref:Temme2010chi}:
\begin{align}\begin{split}\label{eq:Variance}
    \mathrm{Var}_\rho (O_t) & \le \mathrm{Var}_\rho(O_{t=0}) e^{-2\gamma t}, \\
    \text{where}\qquad \mathrm{Var}_\rho (O) & \EqDef \Tr{O^\dagger \Gamma (O)}-|\Tr{O \rho}|^2 .
\end{split}\end{align}
\begin{proof}[Proof of Proposition~\ref{clm:Var-conv}]
    Fix $t\ge0$. Let $E_t$ be the (Hermitian) observable that satisfies $\norm{E_t}_\infty=1$ and
    \begin{align} \label{eq:almost-operator-var}
        \norm{\tilde\rho_t-\rho}_1 = \Tr{E_t\cdot(\tilde \rho_t-\rho)}.
    \end{align}
    Writing $\tilde \rho_t-\rho$ as $\Gamma \big( 
    \Gamma^{-1}(\tilde \rho_t)-\Id\big)$
    and noting that $\Tr{\rho \cdot \big( \Gamma^{-1}(\tilde \rho_t)-\Id\big)}=0$,
    \cref{eq:almost-operator-var} becomes
    \begin{align}
        \norm{\tilde\rho_t-\rho}_1 = \mathrm{Cov}_\rho\Big(E_t,\Gamma^{-1}(\tilde \rho_t)-\Id\Big),
    \end{align}
    where we defined the operator covariance to be $\mathrm{Cov}_\rho(O_A,O_B) \EqDef \Tr{O_A^\dagger \Gamma(O_B)}-\Tr{O_A^\dagger\rho}\Tr{O_B\rho}$.
    Using the Cauchy-Schwartz inequality we get $\mathrm{Cov}_\rho\big(E_t,\Gamma^{-1}(\tilde \rho_t)-\Id\big)
    \le \sqrt{\mathrm{Var}_\rho\big(E_t\big)}\cdot
    \sqrt{\mathrm{Var}_\rho\big(\Gamma^{-1}(\tilde \rho_t)-\Id\big)}$.
     We use H\"older inequality to bound $\mathrm{Var}_\rho\big(E_t\big)\le 1$ and the
    detailed balance equation (\cref{def:detailed-balance2}) for $\Gamma^{-1}(\tilde \rho_t)-\Id=e^{\mcL ^\dagger t}\big(\Gamma^{-1}(\tilde \rho)-\Id\big)$, so that finally,
    \begin{align}
        \norm{\tilde\rho_t-\rho}_1 \lesssim 
        \sqrt{\mathrm{Var}_\rho
        \Big(e^{\mcL ^\dagger t}\big(\Gamma^{-1}(\tilde \rho)-\Id\big)\Big)} 
        \underbrace{\le}_{\text{\cref{eq:Variance}}} 
        e^{-\gamma t}\sqrt{\mathrm{Var}_\rho
        \big(\Gamma^{-1}(\tilde \rho)-\Id\big)}. 
    \end{align}
    
    We finish by reinterpreting the right-hand-side of our last inequality with respect to the choice of $\Gamma$ in \cref{def:detailed-balance2}.
    Short calculation shows that
    \begin{align} \label{eq:Var-is-D2}
        \mathrm{Var}_\rho\big(\Gamma^{-1}(\tilde \rho)-\Id\big)    
        =\Tr{\tilde \rho\cdot\Gamma^{-1}(\tilde \rho)}-1    
        = \begin{cases}
            \Tr{\tilde \rho^2 \rho^{-1}} -1 & \text{for } \Gamma:O\mapsto O\rho \;\;\,\qquad\text{(GNS)},\\
            \Tr{\tilde \rho \rho^{-\frac12}\tilde \rho\rho^{-\frac12}} -1 & \text{for } \Gamma:O\mapsto
            \rho^{\frac 12}O \rho^{\frac 12} \quad\text{(KMS)}.
        \end{cases}
        \end{align}
        The traces are exactly $2^{\mathrm D_2(\tilde \rho\|\rho)}$ in the GNS case, or
        $2^{\tilde {\mathrm D}_2(\tilde \rho\|\rho)}$ in the KMS case (see \cref{eq:sandwich-Reny}).
\end{proof}

Using the above result, the convergence time from a given initial state is guaranteed 
if we could control the initial Rényi-2 divergence with the steady state.
To put in context, we seek an upper bound of this form for 
the steady state obtained after a post-selected Kraus operation on $A$.
Ideally, such a bound should depend only on the local region $A$, resulting in \cref{eq:fast-converge}.
While SRC states are natural candidates to satisfy such a local bound, we leave establishing this for future research.
Instead, we show this for certain instances of SRC Gibbs states.

\begin{corollary}
    Let $\Phi_A(\cdot)=
    \sum_k F_k(\cdot)F_k^\dagger$ be a Kraus decomposition of a local channel. Define
    $p_k=\Tr{F_k\rho F_k^\dagger}$,
    $\sigma_k = \frac{1}{p_k } F_k\rho F_k^\dagger$, and $\mcR_t=e^{\mcL t}$. 
    Then 
    \begin{align}
            \sum_k p_k\norm{\mcR_t(\sigma_k)-\rho}_1 \le C(A)\, e^{-\gamma t},
    \end{align}
    is achieved for Gibbs samplers of: 1. Commuting Hamiltonians.
    2. 1D Hamiltonians. 3. High temperatures.
\end{corollary}

\begin{proof}
    For now, we work in the KMS regime, the GNS is proven similarly. 
    Plugging to \cref{eq:Var-is-D2} the input state $\tilde \rho =\sigma_k=\frac 1{p_k} F_k\rho F_k^\dagger$  for $\rho \propto e^{-\beta H}$ where $H$ is a local Hamiltonian, we get
    \begin{align} \label{eq:bound1}
        \Tr{\tilde \rho \rho^{-\frac12}\tilde \rho\rho^{-\frac12}}
        \le  \norm{\rho^{-\frac12}\tilde \rho\rho^{-\frac12}}_\infty
        \le \frac{1}{p_k}\norm{e^{\beta H/2}F_k e^{-\beta H/2} }_\infty^2 ,
    \end{align}
    where in the first move we used H\"older's inequality and in the second we used
    the $C^*$-identity of $\norm{\cdot}_\infty$.
    We seek for a bound of the form 
    \begin{align} \label{eq:bound2}
        \norm{e^{\beta H/2}F_k e^{-\beta H/2} } _\infty\le h(|A|)  \norm{F_k}_\infty,
    \end{align}
    where $h(x)=e^{\bigO{x}}$ is some function that encompasses a $\beta$ dependence.
    Assuming \cref{eq:bound2}, we get that $\mcR_t$ is a desired recovery channel:
    \begin{equation}
    \begin{aligned}
        \sum_k p_k\norm{\mcR_t(\sigma_k)-\rho}_1
    & \underbrace{\le}_{\text{Prop.}~\ref{clm:Var-conv}}
    \sum_k p_k\,e^{-\gamma  t} 
    \Big(\Tr{\sigma_k \rho^{-\frac12}\sigma_k\rho^{-\frac12}}\Big)^{\frac12}
    \\
   (\text{Eqs.~\eqref{eq:bound1} \& \eqref{eq:bound2}}) \qquad & \le e^{-\gamma t} h(A) \sum_k\sqrt{p_k}\norm{F_k}_\infty \\
     (\text{Cauchy-Schwarz}\& \,\norm{ F_k}_\infty^2\le \Tr{F_k^\dagger F_k}) \qquad & \le e^{-\gamma t} h(A)\Big(\sum_k p_k\Big)^{\frac12} \Big(\sum_{k'}\Tr{F_{k'}^\dagger F_{k'}}\Big)^{\frac12} 
     \\
     (\sum_k F_k^\dagger F_k=\Id) \qquad & \le e^{-\gamma t} h(A) d^{\frac{|A|}{2}}.
    \end{aligned}            
    \end{equation}
    Here in the 3rd line, the trace is over $\mathbb{H}_A$.
    
    It is left to justify \cref{eq:bound2} for specific Gibbs states. Note that in general,
    the expression in the l.h.s of \cref{eq:bound2}, which involves the imaginary time 
    evolution of a local operator, might scale exponentially with system size. 
    The cases known to be locally bounded are commuting Hamiltonians,
    1D Hamiltonians, and high temperature Hamiltonians where $\beta<\beta _c$ and 
    $T_c$ depends on the local geometry of the lattice~\cite{ref:arad2016connecting}. Assume that $H=\sum_i H_i$ where $0\leq H_i\leq \Id$,
    \begin{align} \label{eq:Gibbs-evolution-cases}
            \norm{e^{\beta H}F_k e^{-\beta H} } _\infty\le 
            \norm{F_k}_\infty\cdot
            \begin{cases}
                e^{\bigO{\beta |A|}} & \text{ Commuting Hamiltonians} ,\\
                e^{8\beta e^{8\beta}} e^{4\beta|A|} & \text{ 1D Hamiltonians} ,\\
                \Big(1-\beta/\beta_c\Big)^{-\bigO{|A|}} & \text{ High-$T$}. 
            \end{cases}
    \end{align}
    In the above, the commuting case is simply achieved by commuting the terms in $H$ that do not overlap with $F_k$.
    The 1D case is due to \cite{araki1969gibbs,perez2023locality},    
    and the high temperature case  is Lemma 3.1 of \cRef{ref:arad2016connecting}.
\end{proof}

{~}
So far in this section, we illustrated how $e^{\mcL t}$ can be seen as a  recovery map, but  only briefly addressed its locality properties.
Rigorous estimates can be obtained using the Lieb-Robinson bounds when the Lindbladian is frustration-free, yielding a strictly local recovery map as defined in \ref{def:local-rec}. We refer the reader to \cRefs{ref:brandao2015area,ref:chen2025markovian} for related results.

We complement this direction by proposing an alternative recovery map, 
which is strictly local, and tailored for commuting Hamiltonian Gibbs states.
To do this, we consider a detectability lemma superoperator derived from the Lindbladian~\cite{ref:aharonov2011detectability}.
This approach was also considered in \cRef{ref:Kastoryano2016commutinggibbssampler}
for deriving a strong clustering condition for such Gibbs states.

For now, assume that $\mcL=\sum_i \mcL_i$,
such that each $\mcL_i$ is supported on $O(1)$ neighboring sites and
satisfies GNS condition locally:
\begin{align} \label{eq:localGNS}
    \forall i\qquad\Rightarrow \qquad\mcL_i \Gamma = \Gamma \mcL_i^\dagger, \tab \text{where} \; \; \Gamma(A)=\rho\cdot A.
\end{align}
These maps are only known to exist for commuting Hamiltonians, e.g., with a \emph{Davies} generators (see \cRef{ref:Kastoryano2016commutinggibbssampler}), and therefore will be our focus here.
Let $\mcP_i$ be its infinite time expectations,
$\mcP_i\EqDef \lim_{t\rightarrow \infty} e^{\mcL_i t}$, and divide the set $\{\mcP_i\}$
into layers\footnote{For example, a Lindbladian defined on a 1D lattice
with $\mcL_i$ which are 2-local requires two layers -- ``even" and ``odd".}.
The detectability lemma superoperator is given by the product of these 
layers~\cite{ref:aharonov2011detectability}.


\begin{figure}
    \centering
    \includegraphics[width=0.6\linewidth]{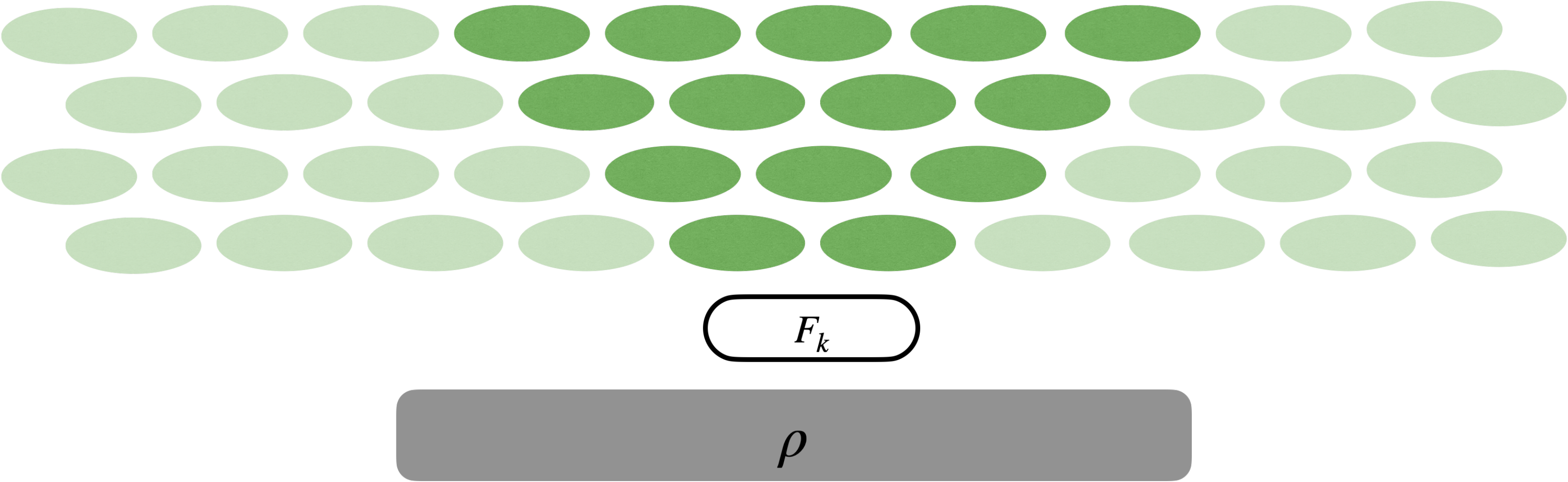}
    \caption{The light cone structure in \eqref{eq:Recovery-DL}.
    The green gates describe $m$ layers of the $\mcP_i$ projectors.
    The faded projectors commute with $F_k$ and absorbed in $\rho$, leading to an effective
    recovery channel with small support.}
    \label{fig:light-cone-DL}
\end{figure}

Let $A$ be a small region, and define the local recovery map
to be the tower of $\mcP_i$'s that are in the light-cone of $m$ layers of the detectability lemma operator
around $A$,
\begin{align}
    \label{eq:Recovery-DL}
    \mcR_m = \prod_{i\in \LC A} \mcP_i .
\end{align}
The light-cone is defined by taking $m$ layers of the detectability lemma operator,
and eliminate all the elements that are not in the light-cone around $A$ (see \Fig{fig:light-cone-DL}). 
We prove that $\mcR_m$ yields local stability.
\begin{prop} \label{clm:LR-Gibbs}
    Let $\rho\propto e^{-\beta H}$ where $\beta>0$ and $H=\sum_i h_i$ is a commuting $k$-local Hamiltonian where $k=O(1)$ and $0\preceq h_i \preceq \Id$.
    Let $\mcL=\sum_i \mcL_i$ be a local Gibbs sampler with local GNS detailed balance \eqref{eq:localGNS} and a spectral gap $\gamma>0$.
    Given a contiguous region $A$, let $\mcR_m$ be recovery channel from \eqref{eq:Recovery-DL}.
    Then for any local channel $\Phi_A(\cdot)=\sum_k F_k (\cdot) F_k^\dagger$,
    we have  
    \begin{align*}
        \sum_k p_k \norm{\mcR_m \big(\frac{1}{p_k}F_k\rho F_k^\dagger\big)-\rho}_1 \le 
        d^{\frac{|A|}{2}}e^{\bigO{\beta|A|}}e^{-m/\xi}
    \end{align*}
    where $\xi = -\ln^{-1}[1-\bigO \gamma ]$.
\end{prop}
Note that the existence of a local recovery map is not unexpected. Indeed, since $\rho$
has vanishing CMI~\cite{ref:Poulin2012markov} and a finite correlation length
determined by $\gamma$~\cite{ref:kastoryano2013rapid}, such a map can be
constructed with a Petz map using \cref{thm:gap->LR}. The novelty here is that the recovery
is implemented directly by the Lindbladian terms.

\begin{proof}
    Let $\mcL = \sum_i \mcL_i$ be a local Lindbladian where each $\mcL_i$ satisfies GNS detailed
    balance with respect to $\rho$.
    As a result, $\rho$ is also a fixed point of every $\mcL_i$,
    and moreover, the infinite time evolution $\mcP_i $ projects onto $\rho$.
    Divide the projectors into layers, and stack $m$ 
    iterations of these layers one by one to get a circuit $\mcP_{(m)}$ consisting 
    of $\bigO m$ layers.
    Let $\Phi(\cdot) =F_k(\cdot)F_k^\dagger $ be a quantum channel on $A$
    given in a Kraus form.
    We claim that 
    \begin{align} \label{eq:1-1-norm} \begin{split}
        \norm{\left(\mcP_{(m)}-\ketbra{\rho}{\Id}\right)  (F_k\rho F_k^\dagger)}_1 & \le 
        C \sqrt{\Tr{F_k\rho F_k^\dagger}}\norm{F_k},\\
        \text{for} \quad C= & \; e^{\bigO{\beta |A|}}(1-\bigO{\gamma})^m 
    \end{split} \end{align}
    We prove \Eq{eq:1-1-norm} below, but for now, assume it holds,
    denote $p_k\EqDef\Tr{F_k\rho F_k^\dagger}$ and sum over \eqref{eq:1-1-norm}:
    \begin{align}
        \sum_k p_k\norm{\frac{1}{p_k}\mcP_{(m)} (F_k \rho F_k^\dagger)-\rho
    }_1 \le C \sum_k \sqrt{p_k} \norm{F_k} \le  C \sqrt{\sum_k p_k}\sqrt{\sum_k \Tr{F_k^\dagger F_k}}= C d^{\frac{|A|}{2}} ,
    \end{align}
    i.e., $\mcP_{(m)}$ is a recovery channel satisfying \Def{def:local-rec}.
    Note that the $\mcP_i$ terms in $\mcP_{(m)}$ supported outside
    the light cone of $A$ may be commuted with $F_k$ and absorbed in $\rho$,
    reducing $\mcP_{(m)}$ to the product of $\mcP_i$ terms that are inside the 
    light cone (see \Fig{fig:light-cone-DL}).
    This leads to a compact recovery channel $\mcR_m$, supported in a ball of radius $\bigO{m}$ around $A$.
    
    The rest of the proof will be devoted to showing \Ineq{eq:1-1-norm}.
    Note that the detailed balance condition, $\mcL_i \Gamma = \Gamma \mcL_i^\dagger $,
    is equivalent to  $\mcH_i\EqDef\Gamma^{-\frac12}\mcL_i \Gamma^{\frac12}$ being hermitian.
    In particular, $\mcH\EqDef - \Gamma^{-\frac12}\mcL \Gamma^{\frac12}$ is an hermitian super-operator,
    frustration free,
    with a unique ground state being $\sqrt \rho$ and $\gap(\mcH)=\gap(\mcL)=\gamma$.
    Each $\mcH_i$ is in fact local due to detailed balance (see Theorem 3.2 in \cRef{ref:firanko2024superhamiltonian}).
    This allows us to apply the detectability lemma for frustration free Hamiltonian
    (Corollary 3 in Ref.~\cite{Anshu2016detectability}):
    \begin{align} \label{eq:detectability-lemma}
        \Delta \EqDef \big(\prod_i P_i \big)^m - \ketbra{\sqrt \rho}{\sqrt \rho}\tab\Rightarrow  \tab \norm{\Delta}_{2\rightarrow 2}
        \le (1 - \bigO{\gamma})^m .
    \end{align}
    where $P_i$ projects onto the groundspace of $\mcH_i$, i.e.,  $P_i=\Gamma^{-\frac12}\mcP_i \Gamma^{\frac12}$.
    Recalling that $\Gamma^{\frac12}(\sqrt \rho)=\rho$ and $\Gamma^{-\frac12}(\sqrt \rho)=\Id$,
    Ineq.~\eqref{eq:1-1-norm} becomes 
    \begin{align} \label{eq:rotated-DL}
        \norm{\left(\Gamma^{\frac12}\Delta\Gamma^{-\frac12}\right)  (F_k\rho F_k^\dagger)}_1 & \le 
        C \Tr{F_k\rho F_k^\dagger}
    \end{align}
        which we now estimate
        \begin{align}\label{eq:long-split}\begin{split}
            \norm{\Gamma^{\frac 1 2}\Delta\Gamma^{-\frac 1 2}(F_k\rho F_k^\dagger)}_1
            & =
            \norm{\rho^{\frac12}\Delta(\rho^{-\frac12}
             F_k\rho F_k^\dagger)}_ 1 
             \\
            \text{(Cauchy-Schwartz)\quad}
              & \le \norm{\rho^{\frac12}}_2 \norm{\Delta(\rho^{-\frac12}
            F_k\rho F_k^\dagger)}_2 \\
           \text{($\Tr{\rho} =1$)} \quad & \le \norm{\Delta}_{2\rightarrow 2}\norm{\rho^{-\frac12}
            F_k\rho F_k^\dagger}_2 \\
            & \le \norm{\Delta}_{2\rightarrow 2}\sqrt{\Tr{F_k\rho F_k^\dagger\rho^{-1}
            F_k\rho F_k^\dagger}}\\
            \text{(H\"older Inequality)} \quad
            & \le \norm{\Delta}_{2\rightarrow 2}\sqrt{\Tr{F_k\rho F_k^\dagger}}
            \sqrt{\norm{\rho^{-1}F_k\rho }\norm{F_k^\dagger}}\\
            \text{First case in \cref{eq:Gibbs-evolution-cases}} \quad & \le \norm{\Delta}_{2\rightarrow 2}
            e^{\bigO{\beta |A|}}\sqrt{p_k}\norm{F_k}.
        \end{split}\end{align}

\end{proof}

One should note the differences between our proof and Thm. 2.6 from \cRef{ref:Kastoryano2016commutinggibbssampler}:
The authors obtained an estimate on the norm of a strong notion of clustering which is given in the $\rho$-weighted norm, at which $\mcL^\dagger$ and $\mcP^\dagger$ are hermitian.
They also use a local Davies generator, and assume it is locally primitive.
We, on the other hand, do not assume a locally primitive Lindbladian that is locally a Davies generator.
More importantly, our proof navigates between the 2-norm (in the rotated basis) and the 1-norm in the usual basis, rather than performing the analysis in the convenient $\rho$-weighted norm.
For this purpose, e.g., it is crucial for us to start from $F_k\rho F_k^\dagger$.

\section{Proofs in \Sec{sec:non-lin}}
\label{app:non-lin}

\subsection{Basic properties}\label{app:non-lin-basic}
\begin{proof}[Proof of Prop.~\ref{prop:Cpqbasic}]

(\cref{prop5:bul-mon} - Monotonicity)
Let us compare $\mcC_{p,q}$ and $\mcC_{p,q'}$ where $q\leq q'$.
We have:
\begin{equation}
    \norm{\rho^{\frac{p}{2}}O\rho^{\frac{q'}{2}}}_{2s'}
    =
    \norm{\rho^{\frac{p}{2}}O\rho^{\frac{q}{2}}\rho^{\frac{q'-q}{2}}}_{2s'}
    \leq 
    \norm{\rho^{\frac{p}{2}}O\rho^{\frac{q}{2}}}_{2s}\norm{\rho^{\frac{q'-q}{2}}}_{\frac{2}{q'-q}}
    =
    \norm{\rho^{\frac{p}{2}}O\rho^{\frac{q}{2}}}_{2s}
\end{equation}
Here, we use the Holder's inequality and the following:
\begin{equation}
    \frac{1}{2s'}=\frac{p+q'}{2}=\frac{p+q}{2}+\frac{q'-q}{2}=\frac{1}{2s}+\frac{1}{2/(q'-q)}.
\end{equation}

    (\cref{prop5:bul-ti} - Continuity) Using triangle inequality and H\"older's inequality, we have 
        \begin{equation} \label{eq:TI-derivation}
        \begin{aligned}
            |\mcC_{p,q}(O,\rho)-\mcC_{p,q}(O,\sigma)|
            \leq \; & \norm{\rho^{\frac{p}{2}}O \rho^{\frac{q}{2}} - \sigma^{\frac{p}{2}}O \sigma^{\frac{q}{2}}}_{2s}\\
            \leq \; & \norm{(\rho^{\frac{p}{2}}-\sigma^{\frac{p}{2}})O\rho^{\frac{q}{2}}}_{2s} + \norm{\sigma^{\frac{p}{2}}O(\rho^{\frac{q}{2}}-\sigma^{\frac{q}{2}})}_{2s}\\
            \leq \; & \norm{\rho^{\frac{p}{2}}-\sigma^{\frac{p}{2}}}_{\frac{2}{p}}\norm{O}_\infty \norm{\rho^{\frac{q}{2}}}_{\frac{2}{q}} + (p\leftrightarrow q)\\
            \leq \; & \norm{\rho^{\frac{1}{2}}-\sigma^{\frac{1}{2}}}_{2}^p \norm{O}_\infty  + (p\leftrightarrow q)
        \end{aligned}
    \end{equation}
     Here in the last two steps we use Eq.(X.7) in \cRef{ref:Bhatia-book}, i.e., for positive semidefinite $A$ and $B$, we have: 
     \begin{equation}\label{eq:bhatia1}
         \norm{A^r-B^r}_k\leq \norm{\abs{A-B}^r}_k,~~(0\leq r\leq 1, k\geq 1),
     \end{equation}
    (here, we pick $A=\rho^{\frac{1}{2}}$, $B=\sigma^{\frac{1}{2}}$, $r=p$, $k=\frac{2}{p}$).
    Finally, using \cref{eq:DBvs2}, we get:
    \begin{align}\label{eq:appcontinuitybound}
        |\mcC_{p,q}(O,\rho)-\mcC_{p,q}(O,\sigma)|
        \le 
        \norm{O}_\infty (2^{\frac p2}\bur{\rho}{\sigma}^p+2^{\frac q2}\bur{\rho}{\sigma}^q)
        \le 
        3 \norm{O}_\infty \bur{\rho}{\sigma}^{\min\{p,q\}}.
    \end{align}

    Remark: the numerical constant 3 is not always optimal for all $p$ and $q$.
    When $p=q=1$, similar proof yields:
    \begin{equation}
        |\mcC_{1,1}(O,\rho)-\mcC_{1,1}(O,\sigma)|
        \le 
        2 \norm{O}_\infty \bur{\rho}{\sigma}.
    \end{equation}

    (\cref{prop5:bul-log} - Log-Convexity:) Fixing $(p_0,q_0)$, $(p_1,q_1)$, we consider the following matrix-valued function:
    \begin{equation}
        G(z)= \rho^{\frac{(1-z)p_0+zp_1}{2}}O \rho^{\frac{(1-z)q_0+zq_1}{2}}.
    \end{equation}
    This is an (element-wise) bounded holomorphic function of $z\in\mathbb{C}$.
    Fixing $O$ and $\rho$, $G(z)$ is (element-wise) bounded when $\Re(z)\in[0,1]$. 
    Moreover, for $\forall t\in\mathbb{R}$,
    \begin{align}
        \norm{G(it)}_{2s_0}=\mcC_{p_0,q_0},~~
        \norm{G(1+it)}_{2s_1}=\mcC_{p_1,q_1}.~~
    \end{align}
    Therefore, using the Hadamard three-line theorem for operator-valued functions 
    (see, e.g., Theorem 2 in \cRef{beigi2013sandwiched}), 
    we get, for $\forall \theta\in[0,1]$, that:
    \begin{equation}\label{eq:Had3line}
        \log\norm{G(\theta)}_{2s_\theta} \leq (1-\theta) \log \mcC_{p_0,q_0} + \theta\log \mcC_{p_1,q_1},
    \end{equation}
    where $s_\theta$ is defined by
    \begin{equation}
        \frac{1}{2s_\theta} = \frac{1-\theta}{2s_0} + \frac{\theta}{2s_1},\qquad s_{i}=\frac{1}{p_i +q_i}~~ (i=0,1).
    \end{equation}
    Straightforward calculation shows that 
    \begin{equation}
        \norm{G(\theta)}_{2s_\theta}=\norm{\rho^{\frac{p_\theta}{2}}O \rho^{\frac{q_\theta}{2}}}_{\frac{2}{p_\theta+q_\theta}}=\mcC_{p_\theta,q_\theta},
    \end{equation}
    where $p_\theta=(1-\theta)p_0+\theta p_1$ and $q_\theta=(1-\theta)q_0+\theta q_1$. 
    Therefore, \cref{eq:Had3line} is exactly the desired log-convexity.   
\end{proof}

\subsection{Local computability}\label{app:non-lin-local-comp}
To prove \cref{thm:local-comp}(\ref{bul2:QMC-operational}), we pass to a purification, where the hypothesis becomes a statement about the single channel $\Tr_C$ acting on a family of operators on the purifying side.
Let $\ket{\Psi}_{ABCE}$ be a purification of $\rho_{ABC}$ (we also write $\Psi$ for its density matrix), and set
\begin{equation}\label{eq:omega-eta}
    \Phi\EqDef\Tr_C,\qquad
    \omega\EqDef\Psi_{CE},\qquad
    \eta\EqDef\Psi_{E}=\Phi(\omega).
\end{equation}
For any operator $O_A$, define
\begin{equation}\label{eq:XO}
    X_O\EqDef\Tr_{A}\big[(O_A\otimes\Id_{CE})\,\Psi_{ACE}\big],
\end{equation}
so that $\Phi(X_O)=\Tr_{A}[(O_A\otimes\Id_{E})\Psi_{AE}]$.

\begin{lemma}\label{lem:purification}
\leavevmode
\begin{enumerate}[(a)]
    \item The map $O_A\mapsto X_O$ is linear, and $X_{\Id}=\omega$. If $O_A\ge0$, then $0\preceq X_O\preceq \norm{O_A}_\infty\,\omega$. Consequently, all $X_O$ are supported on $\supp(\omega)$, and all $\Phi(X_O)$ on $\supp(\eta)$.
    \item     If there is a quantum channel $\mcR:E\to CE$ such that $\mcR(\Phi(X_O))=X_O$ for all $O_A$, then $\condinf{A}{C}{B}_\rho=0$.
    \item For any operator $O_A$,
    \begin{equation}\label{eq:purification-identity}
    \begin{aligned}
        \mcC_{p,q}(O_A,\rho_{AB})&=\norm{\omega^{\frac{q-1}{2}}X_O\,\omega^{\frac{p-1}{2}}}_{2s},\\
        \mcC_{p,q}(O_A,\rho_{ABC})&=\norm{\eta^{\frac{q-1}{2}}\Phi(X_O)\,\eta^{\frac{p-1}{2}}}_{2s}.
    \end{aligned}
    \end{equation}
\end{enumerate}
\end{lemma}

\begin{proof}
(a) Linearity and $X_{\Id}=\omega$ are clear.
For $O_A\ge0$, $X_O=\Tr_A[(O_A^{1/2}\otimes\Id)\Psi_{ACE}(O_A^{1/2}\otimes\Id)]\ge0$; applying this to $\norm{O_A}_\infty\Id-O_A$ gives $X_O\le\norm{O_A}_\infty\,\omega$.
Hence $X_O$ is supported on $\supp(\omega)$ if $O_A\ge 0$.
Same conclusion holds for every $X_O$, since any operator is a linear combination of four positive ones.
Applying $\Phi$ gives the statement for $\Phi(X_O)$.

(b) Let $\{\ket a\}$ be a basis of $\BBH_A$. By \cref{eq:XO}, 
\begin{equation}\label{eq:Psi-expansion}
\begin{aligned}
    \Psi_{ACE}&=\sum_{a,b}\ket{a}\!\bra{b}\otimes X_{\ket{b}\!\bra{a}},\\
    \Psi_{AE}&=\sum_{a,b}\ket{a}\!\bra{b}\otimes \Phi\big(X_{\ket{b}\!\bra{a}}\big),
\end{aligned}
\end{equation}
so $\mcR(\Psi_{AE})=\Psi_{ACE}$. 
Together with $\Tr_C\Psi_{ACE}=\Psi_{AE}$, the data-processing inequality of mutual information enforces $\mutinf{A}{CE}_\Psi=\mutinf{A}{E}_\Psi$, i.e., $\condinf{A}{C}{E}_\Psi=0$.
Finally, since $\Psi$ is pure, we have
\begin{equation}\label{eq:CMI-duality}
    \condinf{A}{C}{B}_\rho = \condinf{A}{C}{B}_\Psi =\condinf{A}{C}{E}_\Psi = 0.
\end{equation}

(c) Take a Schmidt decomposition $\ket{\Psi}=\sum_i\sqrt{\lambda_i}\,\ket{e_i}_{AB}\ket{f_i}_{CE}$ with $\lambda_i>0$.
Then $\rho_{AB}=\sum_i\lambda_i\ketbra{e_i}$, $\omega=\sum_i\lambda_i\ketbra{f_i}$, and
$X_O=\sum_{i,j}\sqrt{\lambda_i\lambda_j}\,\bra{e_j}O_A\ket{e_i}\,\ket{f_i}\!\bra{f_j}$.
Hence the matrix elements are
\begin{equation}
\begin{aligned}
    \bra{f_i}\omega^{\frac{q-1}{2}}X_O\,\omega^{\frac{p-1}{2}}\ket{f_j}
    &=\lambda_i^{\frac{q-1}{2}}\sqrt{\lambda_i\lambda_j}\,\lambda_j^{\frac{p-1}{2}}\bra{e_j}O_A\ket{e_i}
    =\lambda_i^{\frac q2}\lambda_j^{\frac p2}\bra{e_j}O_A\ket{e_i}.
\end{aligned}
\end{equation}
Note that:
\begin{equation}
\begin{aligned}
    \bra{e_j}\rho_{AB}^{\frac p2}O_A\,\rho_{AB}^{\frac q2}\ket{e_i}
    &=\lambda_j^{\frac p2}\lambda_i^{\frac q2}\bra{e_j}O_A\ket{e_i},
\end{aligned}
\end{equation}
thus the matrix of $\omega^{\frac{q-1}{2}}X_O\,\omega^{\frac{p-1}{2}}$ in the basis $\{\ket{f_i}\}$ is the transpose of the matrix of $\rho_{AB}^{\frac p2}O_A\,\rho_{AB}^{\frac q2}$ in the basis $\{\ket{e_i}\}$.
Transposition preserves singular values, which gives the first identity.

The second one follows in the same way from a Schmidt decomposition across $ABC|E$, using $\Phi(X_O)=\Tr_{ABC}[(O_A\otimes\Id)\Psi]$.

\end{proof}

By \cref{lem:purification}(a), we can regard $\Phi$ as a channel from operators on $\supp(\omega)$ to operators on $\supp(\eta)$, on which $\omega$ and $\eta$ are invertible.
By \cref{lem:purification}(b), it suffices to show that the Petz map $\mcP_{\omega,\Phi}$, which does not depend on $O_A$, recovers every $X_O$.
By \cref{lem:purification}(c), the assumption \cref{eq:localcompute1}) becomes
\begin{equation}\label{eq:hyp-complementary}
    \norm{\omega^{\frac{q-1}{2}}X_O\,\omega^{\frac{p-1}{2}}}_{2s}=\norm{\eta^{\frac{q-1}{2}}\Phi(X_O)\,\eta^{\frac{p-1}{2}}}_{2s}
\end{equation}
for all $O_A$.

\begin{proof}[Proof of \cref{thm:local-comp} (\ref{bul2:QMC-operational}) for $p=q=1$]
In this case \cref{eq:hyp-complementary} reads $\norm{X_O}_1=\norm{\Phi(X_O)}_1$.
We use the following characterization of sufficiency via hypothesis testing.

\begin{fact}[\cRef{jencova2024hypothesis}, Theorem 4]\label{fact:jencova}
Let $\tau,\sigma$ be states and $\Phi$ a quantum channel. If $\norm{\tau-t\sigma}_1=\norm{\Phi(\tau)-t\Phi(\sigma)}_1$ for all $t\ge0$, then $\Phi$ is sufficient for $(\tau,\sigma)$.
\end{fact}
Let $O_A\ge0$ with $X_O\ne0$, and let $\tau_O\EqDef X_O/\Tr X_O$, which is a state by \cref{lem:purification}(a).
For any $t\geq 0$, set $O'=O_A/\Tr X_O-t\Id$. By linearity of $O_A\mapsto X_O$,
\begin{equation*}
    \tau_O-t\omega
    =\frac{X_O}{\Tr X_O}-tX_{\Id}
    =X_{O'}.
\end{equation*}
Applying $\Phi$ and using $\Phi(\omega)=\eta$ gives
\begin{equation*}
    \norm{\tau_O-t\omega}_1
    =\norm{X_{O'}}_1
    =\norm{\Phi(X_{O'})}_1
    =\norm{\Phi(\tau_O)-t\eta}_1,
\end{equation*}
where the middle equality follows from \cref{eq:hyp-complementary} with $p=q=1$.
By Fact~\ref{fact:jencova}, $\Phi$ is sufficient for $(\tau_O,\omega)$.
Since $\tau_O$ is supported on $\supp(\omega)$, by Petz's theorem~\cite{petz1986sufficient,petz1988sufficiency}, we further have
\begin{equation}
    \mcP_{\omega,\Phi}(\Phi(\tau_O))=\tau_O.
\end{equation}
Thus $\mcP_{\omega,\Phi}(\Phi(X_O))=X_O$ for all $O_A\ge0$, and by linearity for all $O_A$.
$\rho$ is then a QMC according to \cref{lem:purification} (b).
\end{proof}

For $(p,q)\ne(1,1)$, we use a family of quadratic forms on the tangent directions at a state.
Let $\sigma=\sum_i\lambda_i\ketbra{i}$ be a state, where the sum runs over $\supp(\sigma)$ so that $\lambda_i>0$.
For $t>0$ and an operator $X$ supported on $\supp(\sigma)$, define
\begin{equation}\label{eq:Qt}
    Q_t(\sigma;X)\EqDef\Tr\big[X^\dagger(L_\sigma+tR_\sigma)^{-1}(X)\big]=\sum_{i,j}\frac{\abs{X_{ij}}^2}{\lambda_i+t\lambda_j},
\end{equation}
where $L_\sigma(Z)=\sigma Z$, $R_\sigma(Z)=Z\sigma$, and $X_{ij}=\bra{i}X\ket{j}$.
For self-adjoint $X$, the normalized monotone metrics (also known as quantum Fisher information metrics) form a convex set, whose extreme points are exactly $(1+t)Q_t$ with $t\in[0,1]$~\cite{hansen2008metric}; the cases $t=1$ and $t=0$ correspond to the symmetric-logarithmic-derivative (SLD) and right-logarithmic-derivative (RLD) metrics, respectively.
Below we use them for general, non-self-adjoint $X$.
 
For a state $\sigma$ and a quantum channel $\Phi$, we regard $\Phi$ as a channel from operators on $\supp(\sigma)$ to operators on $\supp(\Phi(\sigma))$, on which $\sigma$ and $\Phi(\sigma)$ are invertible (recall \cref{lem:purification}(a)).
 
\begin{fact}\label{lem:schwarz}
Let $\sigma$ be a state, $\Phi$ a quantum channel, and $X$ an operator supported on $\supp(\sigma)$.
\begin{enumerate}[(a)]
    \item For any $t>0$,
    \begin{equation}\label{eq:Qt-monotone}
        \Delta_t(X)\EqDef Q_t(\sigma;X)-Q_t\big(\Phi(\sigma);\Phi(X)\big)\ge0.
    \end{equation}
    \item The Petz recovery map satisfies
    \begin{equation}\label{eq:chi2-recovery}
        \norm{\sigma^{-\frac14}\big(X-\mcP_{\sigma,\Phi}(\Phi(X))\big)\sigma^{-\frac14}}_2^2
        \le
        \norm{\sigma^{-\frac14}X\sigma^{-\frac14}}_2^2-\norm{\Phi(\sigma)^{-\frac14}\Phi(X)\Phi(\sigma)^{-\frac14}}_2^2 .
    \end{equation}
\end{enumerate}
\end{fact}

Item (a) is Theorem 2.15 of \cRef{lesniewski1999monotone}, which holds for arbitrary (not necessarily self-adjoint) $X$; the case $t=0$ is the Schwarz-type inequality of \cRef{lieb1974schwarz}, and see Theorem 2.2 of \cRef{hiai2023equality} for a general treatment.
Item (b) is Theorem 4.10 of \cRef{gao2024sufficient}.

The connection to $\mcC_{p,q}$ is through the second-order variation around $\sigma$.
For a state $\sigma$ and an operator $Y$ supported on $\supp(\sigma)$, let $G_\sigma(Y)\EqDef\norm{\sigma^{\frac{q-1}{2}}Y\sigma^{\frac{p-1}{2}}}_{2s}^{2s}$.
 
\begin{lemma}\label{lem:hessian}
Let $(p,q)\in(0,1]^2$ and $(p,q)\neq(1,1)$. For any state $\sigma$ and any operator $X$ supported on $\supp(\sigma)$,
\begin{equation}\label{eq:hessian-integral}
    \partial_z\partial_{\bar z}\,G_\sigma(\sigma+zX)\Big|_{z=0}
    =s^2\int_0^\infty w_{p,q}(t)\,Q_t(\sigma;X)\,dt,
\end{equation}
where $w_{p,q}(t)>0$ is given by
\begin{equation}\label{eq:wpq}
    w_{p,q}(t)=\frac{p+q}{\pi}\,\frac{(1+t)\,t^{q-1}\big[\sin(\pi q)+t^{p+q}\sin(\pi p)\big]}{1-2t^{p+q}\cos(\pi(p+q))+t^{2(p+q)}}.
\end{equation}
\end{lemma}
 
\begin{proof}
Let $D=\sigma^{\frac{1}{2s}}$, $Y=\sigma^{\frac{q-1}{2}}X\sigma^{\frac{p-1}{2}}$, and $f(x)=x^{s}$, so that $G_\sigma(\sigma+zX)=\Tr f\big(D^2+zP+\bar zP^\dagger+\abs{z}^2Y^\dagger Y\big)$ with $P=DY$.
Since $D^2$ is invertible, $G_\sigma(\sigma+zX)$ is smooth near $z=0$.
To compute the derivative, we work in the eigenbasis of $\sigma$, so that $D^2$ is diagonal with eigenvalues $\mu_i=\lambda_i^{p+q}$.
By the second-order perturbation formula for $\Tr f$ (see, e.g., \cRef{ref:Bhatia-book}, Chapter V), collecting the coefficient of $\abs{z}^2$ gives
\begin{equation}
\begin{aligned}
    \partial_z\partial_{\bar z}G_\sigma(\sigma+zX)\Big|_{z=0}
    &=\sum_{i,j}\Big[f'(\mu_j)+\mu_i\frac{f'(\mu_i)-f'(\mu_j)}{\mu_i-\mu_j}\Big]\abs{Y_{ij}}^2\\
    &=\sum_{i,j}\frac{\mu_if'(\mu_i)-\mu_jf'(\mu_j)}{\mu_i-\mu_j}\abs{Y_{ij}}^2,
\end{aligned}
\end{equation}
with the divided difference understood as a derivative when $\mu_i=\mu_j$.
Using $\mu_i f'(\mu_i)=s\mu_i^{s}=s\lambda_i$ (since $\mu_i^s=\lambda_i$) and $\abs{Y_{ij}}^2=\lambda_i^{q-1}\lambda_j^{p-1}\abs{X_{ij}}^2$, we obtain
\begin{equation}\label{eq:hessian-kernel}
    \partial_z\partial_{\bar z}G_\sigma(\sigma+zX)\Big|_{z=0}=s^2\sum_{i,j}\frac{1}{\lambda_j}\,g_{p,q}\Big(\frac{\lambda_i}{\lambda_j}\Big)\abs{X_{ij}}^2,
\end{equation}
where $g_{p,q}(x)\EqDef(p+q)\,x^{q-1}\,\frac{x-1}{x^{p+q}-1}$.

To prove \cref{eq:hessian-integral}, it suffices to show 
\begin{equation}\label{eq:stieltjes}
    g_{p,q}(x)=\int_0^\infty\frac{w_{p,q}(t)}{x+t}dt.
\end{equation}
Regard $g_{p,q}(\zeta)$ as analytic on $\BBC\setminus(-\infty,0]$ (take principal branches; the only zero of $\zeta^{p+q}-1$ there is $\zeta=1$, cancelled by $(\zeta-1)$ in the numerator).
Fix $x>0$ and write
\begin{equation}
    g_{p,q}(x)=\frac{1}{2\pi i}\oint_{\Gamma}\frac{g_{p,q}(\zeta)}{\zeta-x}\,d\zeta,
\end{equation}
where $\Gamma$ is a keyhole contour consisting of the circles $\abs{\zeta}=R$ (counterclockwise) and $\abs{\zeta}=r$ (clockwise), joined by the two edges of the cut.
Since $g_{p,q}(\zeta)=O(\abs{\zeta}^{q-1})$ as $\zeta\to0$ and $O(\abs{\zeta}^{-p})$ as $\zeta\to\infty$, the circles contribute $O(r^q)+O(R^{-p})\to0$, and only the edges survive.
As $g_{p,q}$ is real on $(0,\infty)$, $g_{p,q}(-t-i0)=\overline{g_{p,q}(-t+i0)}$, so the jump across the cut is $2i\Im g_{p,q}(-t+i0)$; this gives \cref{eq:stieltjes} with $w_{p,q}(t)=-\frac1\pi\Im g_{p,q}(-t+i0)$, i.e., $w_{p,q}$ is the discontinuity of $g_{p,q}$ across the negative axis.
Putting $\zeta=te^{i\pi}$,
\begin{equation}
    g_{p,q}(-t+i0)=\frac{(p+q)(1+t)\,t^{q-1}e^{i\pi q}}{t^{p+q}e^{i\pi(p+q)}-1}
    =\frac{(p+q)(1+t)\,t^{q-1}\big(t^{p+q}e^{-i\pi p}-e^{i\pi q}\big)}{\abs{1-t^{p+q}e^{i\pi(p+q)}}^2},
\end{equation}
whose imaginary part gives \cref{eq:wpq}.

Finally, $w_{p,q}(t)>0$: the denominator is nonzero since $p+q<2$, and $\sin(\pi p),\sin(\pi q)\ge0$ with at least one of them positive because $(p,q)\neq(1,1)$.
\end{proof}

\begin{proof}[Proof of \cref{thm:local-comp}(\ref{bul2:QMC-operational}) for $(p,q)\neq(1,1)$]
Fix an operator $O_A$, and apply \cref{eq:hyp-complementary} to $\Id+zO_A$, $z\in\BBC$.
Since $X_{\Id+zO}=\omega+zX_O$, we get $G_\omega(\omega+zX_O)=G_\eta(\eta+z\Phi(X_O))$ for all $z$.
Taking $\partial_z\partial_{\bar z}$ at $z=0$ and using \cref{lem:hessian} on both sides gives
\begin{equation}
    \int_0^\infty w_{p,q}(t)\,\Delta_t(X_O)\,dt=0,
\end{equation}
where $\Delta_t$ is defined in \cref{lem:schwarz} with $\sigma=\omega$.
Since $\Delta_t(X_O)\ge0$ is continuous in $t$ and $w_{p,q}>0$, we have $\Delta_t(X_O)=0$ for all $t>0$.

Next, by the integral representation of monotone metrics applied to $f(x)=\sqrt{x}$~\cite{lesniewski1999monotone}, we have $\norm{\sigma^{-\frac14}X\sigma^{-\frac14}}_2^2=\frac1\pi\int_0^\infty t^{-\frac12}\,Q_t(\sigma;X)\,dt$ for any state $\sigma$ and any $X$ supported on $\supp(\sigma)$.
Applying this to $(\omega,X_O)$ and $(\eta,\Phi(X_O))$, we find
\begin{equation}
    \norm{\omega^{-\frac14}X_O\,\omega^{-\frac14}}_2^2-\norm{\eta^{-\frac14}\Phi(X_O)\,\eta^{-\frac14}}_2^2
    =\frac1\pi\int_0^\infty t^{-\frac12}\,\Delta_t(X_O)\,dt=0.
\end{equation}
Thus the right-hand side of \cref{eq:chi2-recovery} vanishes, and hence
\begin{equation}
    \mcP_{\omega,\Phi}\big(\Phi(X_O)\big)=X_O .
\end{equation}
\end{proof}

    \begin{proof}[Proof of \cref{thm:local-comp-1}(1)]
    We apply the data processing inequality \cref{eq:DPICpq} twice,
    once with $\mathrm{Tr}_C$ and once with the recovery map $\mcR$:
    \begin{align}
    \mcC_{p,q}(O_A,\rho_{ABC})
    \leq
    \mcC_{p,q}(O_A,\rho_{AB})
     \leq
     \mcC_{p,q}(O_A,\mcR(\rho_{AB})).
    \end{align}
    Then, applying continuity bound (\cref{prop5:bul-ti} of Prop.~\ref{prop:Cpqbasic}, \cref{eq:appcontinuitybound}) with $\rho_{ABC}$ and $\mcR(\rho_{AB})$, we find:
    \begin{equation}
        \mcC_{p,q}(O_A,\mcR(\rho_{AB}))
        \leq
        \mcC_{p,q}(O_A,\rho_{ABC})
        + 3 \norm{O_A}_\infty \bur{\mcR(\rho_{AB})}{\rho_{ABC}}^p.
    \end{equation}
    The desired result is then given by fact \ref{fact:CMI}.
    \end{proof}

\subsection{Local computability--an alternative proof}\label{app:non-linear-alternative}

In this subsection, we give an alternative proof of \cref{thm:local-comp} (\ref{bul2:QMC-operational}) under the assumption that $(p,q)\neq(1,1)$ and $\rho_{ABC}$ has full rank.
In this proof, it suffices to assume $\mcC_{p,q}(O_A,\rho_{AB}) = \mcC_{p,q}(O_A,\rho_{ABC})$ for all unitary operators $O_A=U_A$.

\begin{proof}
    We can express $\mcC_{p,q}$ in terms of $\mathrm D_{\alpha,z}$ (see \cref{eq:pq->az}).
    Therefore, for any unitary $U_A$,
    \begin{equation}
        \mathrm{D}_{\alpha,z}(U_A\rho_{ABC} U_A^\dagger || \rho_{ABC}) = \mathrm{D}_{\alpha,z}(U_A\rho_{AB} U_A^\dagger || \rho_{AB}),
    \end{equation}
    where $p=\frac{1-\alpha}{z}$, $q=\frac{\alpha}{z}$. 
    Under the assumption that $\rho$ is full-rank, the region of $(p,q)$ we consider corresponds to the region of $(\alpha,z)$ such that the equality of the data processing of $D_{\alpha,z}$ implies sufficiency (\cRef{hiai2024alpha}, Theorem 4.5).
    Namely, there exists a quantum channel $\mcR:AB\rightarrow ABC$ such that:
    \begin{align}
            \mcR(\rho_{AB})&=\rho_{ABC},\label{eq:sufficency1}\\
            \mcR(U_A\rho_{AB}U_A^\dagger)&=U_A\rho_{ABC}
            U_A^\dagger.\label{eq:sufficency2}
    \end{align}
    While the channel $\mcR$ may depend on $U_A$ a priori, Ref.~\cite{petz1986sufficient} shows that the following $U_A$-independent Petz map always suffices:
    \begin{equation}
        \mcR_{\rho_{ABC},\Tr_C}(\cdot)=\rho_{ABC}^{\frac{1}{2}}(\rho_{AB}^{-\frac{1}{2}}(\cdot)\rho_{AB}^{-\frac{1}{2}}\otimes \Id_C)\rho_{ABC}^{\frac{1}{2}}.
    \end{equation}
    Now integrating over $U_A$ in \cref{eq:sufficency2}, we find:
    \begin{equation}
        \mcR(\tau_A\otimes \rho_B)= \tau_A \otimes \rho_{BC},
    \end{equation}
    where $\tau_A=\frac{\Id_A}{d_A}$.
    Combined with \cref{eq:sufficency1} and the 
    data processing inequality, we get:
    \begin{equation} \label{eq:CMI-almostequality}
        \drel{\rho_{ABC}}{\tau_A\otimes \rho_{BC}}  =\drel{\rho_{AB}}{\tau_A\otimes \rho_{B}}.
    \end{equation} 
    Straightforward calculation shows that this is equivalent to $\condinf{A}{C}{B}_\rho=0$.
\end{proof}

\subsection{Decay of nonlinear correlation}

\begin{proof}[Proof of \cref{thm:clusterCpq}]\label{proofclusterCpq}
    We follow the proof strategy outlined in the main text, relaxing the assumption on $\rho$ to allow for decaying correlations and conditional mutual information.
    Define $B_1$ (and $B_2$) from \cref{fig:corr-proof} to be an annulus of radius $r$ around $A_1$ (and $A_2$).
    We separate the desired bound into three components:
    \begin{equation}\label{eq:clustering-nts}
        \begin{split}
            |\mcC_{p,q}(O_1O_2,\rho)-\mcC_{p,q}(O_1,\rho)\mcC_{p,q}(O_2,\rho)|
           \le &  \underbrace{|\mcC_{p,q}(O_1O_2,\rho)-\mcC_{p,q}(O_1O_2,\rho_{AB})|}_ {\text{(a)}}\\
            + &\underbrace{|\mcC_{p,q}(O_1O_2,\rho_{AB})-\mcC_{p,q}(O_1,\rho_{A_1B_1})\cdot \mcC_{p,q}(O_2,\rho_{A_2B_2})|} _{\text{(b)}}\\
            + & \underbrace{|\mcC_{p,q}(O_1,\rho_{A_1B_1})\cdot \mcC_{p,q}(O_2,\rho_{A_2B_2})-\mcC_{p,q}(O_1,\rho)\cdot\mcC_{p,q}(O_2,\rho)|}
             _{\text{(c)}}.
        \end{split}
    \end{equation}

    In the following, we upper bound the above three terms.
    For convenience, we denote 
    \begin{gather}
        \condinf{A}{C}{B}=\delta_{12}^2,~~~
        \condinf{A_1}{CB_2A_2}{B_1}=\delta_{1}^2,~~~
        \condinf{A_2}{CB_1A_1}{B_2}=\delta_{2}^2,\label{eq:defdeltaepsilon1}\\
        \bur{\rho_{AB}}{\rho_{A_1 B_1}\otimes \rho_{A_2 B_2}}= \epsilon.\label{eq:defdeltaepsilon2}
    \end{gather}
     We assume $\norm{O_1}_\infty = \norm{O_2}_\infty =1$ for simplicity and assume WLOG that $p\le q$.
\begin{itemize}
    \item     Term (a).
     We apply \cref{eq:localcompute3} with $A=A_1A_2$,
     $B=B_1B_2$ and $C$ (see \cref{fig:corr-proof})
    with $O_A=O_1O_2$ and find:
    \begin{align} \label{eq:clustering-proof-1}
        |\mcC_{p,q}(O_A,\rho_{ABC})-\mcC_{p,q}(O_A,\rho_{AB})|
        \lesssim \delta_{12}^{p} .
    \end{align}
    \item Term (b). 
    We apply  the continuity bound (\cref{prop5:bul-ti}) for $\rho_{AB}$ and $\rho_{A_1B_1}\otimes \rho_{A_2B_2}$, then apply the multiplicativity of $\mcC$ (\cref{prop5:bul-mult}):
    \begin{align}
        |\mcC_{p,q}(O_A,\rho_{AB})-\mcC_{p,q}(O_1,\rho_{A_1B_1})\cdot \mcC_{p,q}(O_2,\rho_{A_2B_2})|\lesssim
        \epsilon^p .
    \end{align}
    \item Term (c). 
    Applying \cref{eq:localcompute3} on each individual terms gives
    \begin{equation}
        \abs{\mcC_{p,q}(O_i,\rho_{A_iB_i})-\mcC_{p,q}(O_i,\rho)}\lesssim  \delta_i^{p},~~i=1,2.
    \end{equation}
    Using the bound $\mcC_{p,q}(O_i,\rho_{A_iB_i})\le 1$ (\cref{prop5:bul-range} of Prop.~\ref{prop:Cpqbasic}) and triangle inequality
    we conclude that the term (c) is upper bounded by 
    $\lesssim \delta_1^p+\delta_2^p$.
\end{itemize}
    Collecting all the terms, we get:
    \begin{align}  \label{eq:Cpq-deltas} 
         |\mcC_{p,q}(O_1O_2,\rho)-\mcC_{p,q}(O_1,\rho)\mcC_{p,q}(O_2,\rho)|         
         \lesssim
          (\delta_1^p+\delta_2^p+\delta_{12}^p+\epsilon^p) .
    \end{align}

    Next, let us estimate the $\delta$'s and the $\epsilon$ associated with a SRC state (recall \cref{def:gapped}).
    \begin{itemize}
        \item $\delta_1, \delta_2$. 
   The SRC assumption \cref{def:CMI-decay} on $\rho$ gives:
     \begin{align} \label{eq:bound-delta3}
        \delta_i  \lesssim (\condinf{A_i}{(A_iB_i)^c}{B_i})^{\frac 12}
        \le g({A_i})^{\frac12}\, e^{-\frac r{2\zeta}},\qquad i=1,2.
    \end{align}
    \item $\delta_{12}$.
    The chain rule and strong sub-additivity imply:
\begin{equation}\label{eq:entropy-calculus}
\begin{aligned}
    \condinf{A_1A_2}{C}{B_1B_2}
    = \condinf{A_1}{C}{B_1B_2} + \condinf{A_2}{C}{ A_1B_1B_2} \le \condinf{A_1}{A_2B_2C}{ B_1} + \condinf{A_2}{A_1B_1C}{ B_2},
\end{aligned}
\end{equation}
which further implies:
    \begin{align} \label{eq:bound-delta1}
        \delta_{12} \leq \sqrt{\delta_1^2+\delta_2^2}\lesssim \max_i\{g(A_i)^{\frac12} \}\, e^{-\frac r{2\zeta}} .
    \end{align}
    \item $\epsilon$.
    The SRC assumption (\cref{eq:DOC}) on $\rho$ gives the correlation decay, which further implies the bound on the trace distance (\cref{eq:third-equivlence}):
    \begin{align*}
    \norm{\rho_{AB}-\rho_{A_1B_1} \rho_{A_2B_2}}_1 
    \lesssim
    d_{A_1B_1} f(A_1B_1)e^{-\dist(B_1,B_2)/\eta} 
     = e^{\bigO{r^D}-\dist(A_1,A_2)/\eta},
    \end{align*} 
    where we used
    $\{f(x),d_x\}=e^{\bigO x}$ and $\dist(B_1,B_2)=\dist(A_1,A_2)-2r$ 
    (here the constant factor in the big $O$ notation is $A$-independent).
    The Fuchs–van de Graaf inequality then implies:
    \begin{equation} \label{eq:bound-epsilon}
    \epsilon=\bur{\rho_{AB}}{\rho_{A_1B_1}\rho_{A_2B_2}} \le \norm{\rho_{AB}-\rho_{A_1B_1}\otimes \rho_{A_2B_2}}_1^{\frac 12}
    \leq
    e^{\bigO{r^D}-\dist(A_1,A_2)/2\eta},
    \end{equation} 
\end{itemize}

A bound on l.h.s. of \cref{eq:clustering-nts} can be obtained by plugging bounds \eqref{eq:bound-delta3},\eqref{eq:bound-delta1},\eqref{eq:bound-epsilon} into \cref{eq:Cpq-deltas} and tuning $B_1$ and $B_2$ (equivalently, tuning $r$) to optimize the overall result.
To make \eqref{eq:bound-epsilon} small, we choose:
    \begin{align} \label{eq:choise-r}
        r=\bigO{\dist(A_1,A_2)}^{ 1/ D}
    \end{align}
Under this choice, we get:
\begin{equation}
    \delta_1, \delta_2, \delta_{12} \lesssim g(A_1)^{\frac{1}{2}}e^{-\bigO{\dist(A_1,A_2)}^{ 1/ D}},~~~
    \epsilon\lesssim e^{-\bigO{\dist(A_1,A_2)}}.
\end{equation}
Plugging it into \cref{eq:Cpq-deltas}, we get
\begin{equation}
    |\mcC_{p,q}(O_1O_2,\rho)-\mcC_{p,q}(O_1,\rho)\mcC_{p,q}(O_2,\rho)| \leq C(A)
    e^{-p\bigO{\dist(A_1,A_2)}^{ 1/ D}}.
\end{equation}
\end{proof}

\textbf{Remark:}
    In case where
    $\norm{\rho_{AB}-\rho_{A_1B_1} \rho_{A_2B_2}}_1 \le
    \poly(A_1B_1)\,e^{-\dist(B_1,B_2)/\eta}$, 
    we can choose $r=\bigO{\dist(A_1,A_2)}$.
    Then the final bound admits the form
    \begin{align} \label{eq:proving-remark}
|\mcC_{p,q}(O_1O_2,\rho)-\mcC_{p,q}(O_1,\rho)\mcC_{p,q}(O_2,\rho)|
         \leq C(A)
    e^{-p\cdot\bigO{\dist(A_1,A_2)}}.
    \end{align}

\section{Proofs in \Sec{sec:Purification}}\label{app:purifications}



\begin{proof}[Proof of \cref{thm:TFD-LC}(local computability of canonical purification)]  

Consider the quantum channel
\begin{equation}
    \Phi=\Tr_C:ABC\longrightarrow AB
\end{equation}
and the states
\begin{equation}
    \eta_{ABC}=\frac{\mathbbm{1}_A}{d_A}\otimes\rho_{BC},
    \qquad
    \Phi(\eta_{ABC})
    =\frac{\mathbbm{1}_A}{d_A}\otimes\rho_B.
\end{equation}
We have:
\begin{equation}
    D(\rho_{ABC}\Vert\eta_{ABC})
    -
    D(\Phi(\rho_{ABC})\Vert \Phi(\eta_{ABC}))
    = \Delta :=
    I(A:C|B)_\rho.
    \label{eq:cmi-dpi-deficit}
\end{equation}

    We use an $L_2$ version of the Petz map developed in  \cRef{JungeLaRacuente2022}.
    In particular, Theorem~1.5 therein with $p=2$ asserts that\footnote{Note that in our papar the logarithms and entropies have base 2 where \cRef{JungeLaRacuente2022} uses base $e$.}
\begin{equation}
    4\int_{\mathbb{R}}
    -\log f_2\!\left(
        \sqrt{\rho_{ABC}},
        \mcR_t(\sqrt{\rho_{AB}})
    \right)
    \beta_0(t)\,\dd t
    \leq \Delta.
\end{equation}
Here, 
\begin{equation}
    f_2(a,b) = \norm{\sqrt{a}\sqrt{b}}_2 = \sqrt{\Tr(ab)},
    \qquad
    \beta_0(t)
    =\frac{\pi}{2\big(\cosh(\pi t)+1\big)}.
\end{equation}
The rotated $L_2$-Petz map is defined as:
\begin{equation}
\begin{aligned}
    \mcR_t(X)
    :=&
        \eta^{\frac{\bar z}{2}}
    \Phi^\dagger\!\left(
        \Phi(\eta)^{-\frac{\bar z}{2}}
        X
        \Phi(\eta)^{-\frac{z}{2}}
    \right)
    \eta^{\frac{ z}{2}} \qquad (\text{here~}z=(1+it)/2)\\
    =&
    \rho_{BC}^{\frac{1-it}{4}}
    \left(
        \rho_B^{-\frac{1-it}{4}}
        X
        \rho_B^{-\frac{1+it}{4}} \otimes \mathbbm{1}_C
    \right)
    \rho_{BC}^{\frac{1+it}{4}}.
    \label{eq:rotated-l2-petz}
\end{aligned}
\end{equation}
Since $\beta_0$ is a probability density, there exists a $t\in\mathbb{R}$ such that
\begin{equation}
    f_2\!\left(
        \sqrt{\rho_{ABC}},
        \mcR_t(\sqrt{\rho_{AB}})
    \right)
    \geq 2^{-\Delta/4},
    \label{eq:jl-fidelity}
\end{equation}
which simplifies to
\begin{equation}
    \Tr\!\left[
        \sqrt{\rho_{ABC}}\,
        \mcR_t(\sqrt{\rho_{AB}})
    \right]
    \geq 2^{-\Delta/2}.
    \label{eq:amplitude-overlap}
\end{equation}

Now we vectorize $\sqrt{\rho_{AB}}$, $\sqrt{\rho_{ABC}}$, and $\mcR_t(\sqrt{\rho_{AB}})$.
Importantly, recalling that in the definition \cref{eq:rotated-l2-petz}, there is no left nor right action on $A$.
Therefore, 
\begin{equation}
\ket{\mcR_t(\sqrt{\rho_{AB}})} = (\mathbbm{1}_{A\bar A} \otimes W_t) \ket{\sqrt{\rho_{AB}}},
\end{equation}
where $W_t: B\bar B\to B\bar B C\bar C$.
\Cref{eq:amplitude-overlap} becomes:
\begin{equation}\label{eq:Wtbound}
    \bra{{\sqrt{\rho_{ABC}}}}\mathbbm{1}_{A\bar A} \otimes W_t\ket{{\sqrt{\rho_{AB}}}} \geq 2^{-\Delta/2}.
\end{equation}
Moreover, Proposition~4.2 of Ref.~\cite{JungeLaRacuente2022} asserts that the unrotated $L_2$-Petz map is a $2\to 2$ contraction: for any $X: AB\to AB$, we have:
\begin{equation}
    \norm{\mcR_0(X)}_2\leq \norm{X}_2.
\end{equation}
Since $\mcR_t$ is just $\mcR_0$ composed with two unitary conjugations, we also have:
\begin{equation}
    \norm{\mcR_t(X)}_2\leq \norm{X}_2.
        \label{eq:t-contraction}
\end{equation}
On the vectorized side, this implies that
\begin{equation}
    \norm{W_t}_\infty\leq 1   ~~~~\text{(as a $B\bar B\to B\bar B C\bar C$ operator)}.
\end{equation}

We claim that we can replace the contraction $W_t$ by an isometry $V: B\bar B\to B\bar B C\bar C$ without decreasing the overlap. 
There exists an operator $M:B\bar B\to B\bar B C\bar C$ such that, for every $X:B\bar B\to B\bar B C\bar C$,
\begin{equation}
    \bra{{\sqrt{\rho_{ABC}}}}     \mathbbm{1}_{A\bar A}\otimes X     \ket{{\sqrt{\rho_{AB}}}}
    =
    \Tr(M^\dagger X).
\end{equation}
Since $\dim(B\bar B C\bar C)\geq\dim(B\bar B)$, we can always polar decompose $M$ as $M=V|M|$, where $|M|:B\bar B\to B\bar B$ is non-negative and $V:B\bar B\to B\bar B C\bar C$ is an isometry.
Then,
\begin{equation}
        \bra{{\sqrt{\rho_{ABC}}}}     \mathbbm{1}_{A\bar A}\otimes V     \ket{{\sqrt{\rho_{AB}}}}
    =
    \Tr(M^\dagger V) = \norm{M}_1 
    \geq \abs{\Tr(M^\dagger W_t)} \geq 2^{-\Delta/2}.
\end{equation}

Since $V$ is an isometry acting only on $B\bar B$, the states $\ket{{\sqrt{\rho_{ABC}}}}$ and $(\mathbbm{1}_{A\bar A}\otimes V)\ket{{\sqrt{\rho_{AB}}}}$ purify the two reduced states in
the statement of the theorem. 
Uhlmann's theorem therefore lower bounds their fidelity by $2^{-\Delta/2}$. 
\end{proof}

For convenience, before proving \cref{thm:canonical-clustering}, let us first state a continuity bound for the canonical purification:
\begin{equation}\label{eq:sqrt-dist}
    \norm{\ketbra{\sqrt{\rho}}-\ketbra{\sqrt{\sigma}}}_1
    \leq 2 \norm{\ket{\sqrt{\rho}}-\ket{\sqrt{\sigma}}}
    = 2\norm{\sqrt{\rho}-\sqrt{\sigma}}_2
    \leq 2\sqrt{2} \bur{\rho}{\sigma}.
\end{equation}

\begin{proof}[Proof of \cref{thm:canonical-clustering} (decay of correlation in canonical purification)] 

We use the same strategy as that in \cref{thm:clusterCpq}.
In particular, we use the definition \cref{eq:defdeltaepsilon1,eq:defdeltaepsilon2}.
Combining local computability (\cref{thm:TFD-LC}) and trace-norm bound (\cref{eq:sqrt-dist}) we find:
\begin{align}  |\expval{O_{A_1\bar{A}_1}}_{\sqrt{\rho_{A_1B_1}}} - \expval{O_{A_1\bar{A_1}}}_{\sqrt{\rho_{ABC}}}| 
    \lesssim 
    \delta_1 \norm{O_{A_1\bar A_1}}_\infty,\label{locD1}\\
|\expval{O_{A_2\bar{A_2}}}_{\sqrt{\rho_{A_2B_2}}} - \expval{O_{A_2\bar{A_2}}}_{\sqrt{\rho_{ABC}}}| 
    \lesssim 
    \delta_2 \norm{O_{A_2\bar A_2}}_\infty,\label{locD2}\\
|\expval{O_{A\bar{A}}}_{\sqrt{\rho_{AB}}} - \expval{O_{A\bar{A}}}_{\sqrt{\rho_{ABC}}}| 
    \lesssim 
    \delta_{12} \norm{O_{A\bar A}}_\infty.\label{locD3}
\end{align}

Moreover, using the continuity bound \cref{eq:sqrt-dist}, we get:
 \begin{equation} \label{eq:OAA-tensorization}
    |\expval{O_{A\bar{A}}}_{\sqrt{\rho_{AB}}} - \expval{O_{A\bar{A}}}_{\sqrt{\rho_{A_1B_1}\otimes\rho_{A_2B_2}}}| 
    \lesssim 
    \epsilon\norm{O_{A\bar A}}_\infty.
\end{equation}
In our problem, $O_{A\bar A}=O_{A_1\bar A_1}O_{A_2\bar A_2}$, hence
\begin{align}
    \expval{O_{A\bar{A}}}_{\sqrt{\rho_{A_1B_1}\otimes{\rho_{A_2B_2}}}} 
    = 
    \expval{O_{A_1\bar{A_1}}}_{\sqrt{\rho_{A_1B_1}}}
    \expval{O_{A_2\bar{A_2}}}_{\sqrt{\rho_{A_2B_2}}} .
\end{align}

Collecting the above equations, we arrive at
\begin{align}
    |\expval{O_{A_1\bar{A_1}}O_{A_2\bar{A_2}}}_{\sqrt{\rho}} - \expval{O_{A_1\bar{A_1}}}_{\sqrt{\rho}}\expval{O_{A_2\bar{A_2}}}_{\sqrt{\rho}}
    \lesssim 
    \norm{O_{A_1\bar{A_1}}}_\infty \norm{O_{A_2\bar{A_2}}}_\infty\big(\delta_1+\delta_2+\delta_{12}+\epsilon \big).
\end{align}
Let $L=\dist(A_1,A_2)$. Using
   \cref{eq:bound-delta3,eq:bound-delta1,eq:bound-epsilon}
   with $r=cL^{1/D}$ for a sufficiently small constant
   $c>0$, we obtain
   \begin{align}
       \delta_1+\delta_2+\delta_{12}+\epsilon
       \lesssim e^{-r/(2\zeta)}
       +e^{\bigO{r^D}-L/(2\eta)}
       \le C e^{-(L/\xi)^{1/D}}.
   \end{align}
   Substituting this estimate into the preceding bound
   proves \cref{thm:canonical-clustering}.

\end{proof}

\section{Some Counterexamples}\label{app:countereg}
\subsection{Zero MI but not locally implementable}
Consider the following $n$-qubit state:
\begin{equation}
    \rho=\frac{1}{2}\ketbra{0}\otimes\frac{\Id}{2^{n-1}} + \frac{1}{2}\ketbra{1}\otimes\rho_{\text{parity}},
\end{equation}
where $\rho_{\text{parity}}$ is a $(n-1)$-qubit state:
\begin{equation}
    \rho_{\text{parity}} \propto \sum_{a_2+\cdots+a_n=0} \ketbra{a_2a_3\cdots a_n}.
\end{equation}
Tracing out a qubit $k\in\{2,\cdots, n\}$ on $\rho$ gives:
\begin{equation}
    \Tr_k\rho = \frac{\Id}{2^{n-1}},
\end{equation}
hence the state $\rho$ has exactly zero MI:
$\mutinf{A_1\cdots A_{k-1}}{ A_{k+1}\cdots A_n}_\rho=0$.

Now let us postselect $\ket{1}$ on the first qubit. We get:
\begin{equation}
    \sigma = \ketbra{1}\otimes\rho_{\text{parity}}.
\end{equation}
For this state,
\begin{equation}
    \mutinf{A_1\cdots A_{n-1}}{A_n}_\sigma=1.
\end{equation}
However, the original state has
\begin{equation}
    \mutinf{A_1\cdots A_{n-1}}{A_n}_\rho=\frac{1}{2}.
\end{equation}
Data processing inequality implies that there is no quantum channel on $A_1\cdots A_{n-1}$ that converts $\rho$ to $\sigma$.

\subsection{Nonzero CMI but locally implementable}
Consider $\rho=\rho_{\text{parity}}$ to be the $n$-qubit parity state as above. 
We claim that any local operation on qubit 1 can be implemented by a quantum channel on $1$ and $2$.

To show it, we rewrite $\rho$ as:
\begin{equation}
    \rho \propto \ketbra{00}\otimes \rho_0 +  \ketbra{11}\otimes \rho_0 +\ketbra{01}\otimes \rho_1
    + \ketbra{10}\otimes \rho_1.
\end{equation}
Here $\rho_0$ and $\rho_1$ are $(n-2)$-qubit parity states with even and odd parity, respectively.
Applying a Kraus operator $F$ on qubit 1, we get:
\begin{equation}
    \sigma \propto K\rho K^\dagger =
    \frac{1}{2(\alpha^2+\beta^2)}
    \ketbra{\alpha0}\otimes \rho_0 +  \ketbra{\beta1}\otimes \rho_0 
    +\ketbra{\alpha1}\otimes \rho_1     + \ketbra{\beta 0}\otimes \rho_1.
\end{equation}
Here $\ket{\alpha}=K\ket{0}$ and $\ket{\beta}=K\ket{1}$ are unnormalized states, $\alpha^2=\braket{\alpha}$, $\beta^2=\braket{\beta}$.

To prepare $\sigma$ from $\rho$, it suffices to apply a quantum channel $\Phi$ on 1 and 2 such that:
\begin{equation}
\begin{aligned}
        \Phi(\ketbra{00})=\Phi(\ketbra{11}) &= \frac{1}{\alpha^2+\beta^2}(\ketbra{\alpha 0} + \ketbra{\beta1}),\\
        \Phi(\ketbra{01} )=\Phi(\ketbra{10}) &= \frac{1}{\alpha^2+\beta^2}(\ketbra{\alpha 1} + \ketbra{\beta0}).
\end{aligned}
\end{equation}
Such $\Phi$ always exists. For example, we can measure the parity, replace qubits 1 and 2 by the desired states conditioning on the measurement result, then erase the register.

The same arguments also establish that any local operation on qubit 1 to $k$ can be implemented by a quantum channel on 1 to $k+1$.

\subsection{Zero MI and CMI but not locally recoverable}

We present a tripartite state over $ABC$ such that $I(A:C)=I(A:C|B)=0$ but a postselction on $A$ is not recoverable by a quantum channel on $AB$.
This shows that a partition like the one in \cref{fig:Gap-recovery}, with four subregions, is necessary.

Let $d_A=d_C=2, d_B=3$.
We denote the basis on the qubits ($A$ and $C$) as $\{\ket 0, \ket 1\}$; 
for the qutrit $B$ we denote the basis as $\{\ket 0, \ket 1,\ket 2\}$. 
Define $\rho_{ABC}$ as:
\begin{equation}
\begin{aligned}
    &\rho_{ABC}=\frac{1}{2}\rho_{ABC}^{(0)}+\frac{1}{2}\rho_{ABC}^{(1)},\\
    &\rho_{ABC}^{(0)}=\ketbra{0_A}\otimes (\frac{\ket{0_B0_C}+\ket{1_B1_C}}{\sqrt{2}})(\frac{\bra{0_B0_C}+\bra{1_B1_C}}{\sqrt{2}}),\\    &\rho_{ABC}^{(1)}=\ketbra{1_A}\otimes\ketbra{2_B}\otimes\frac{\ketbra{0_C}+\ketbra{1_C}}{2}.
\end{aligned}
\end{equation}
Straightforward calculation shows that $\mutinf{A}{C}_\rho=\condinf{A}{C}{B}_\rho=0$.

However, suppose we perform a projection $\ketbra{1_A}$, which has success probability 1/2, then we get $\rho_{ABC}^{(1)}$.
This state has no correlation between $AB$ and $C$, so we can not recover the original state by a quantum channel on $AB$, which has nonzero $AB|C$ correlation.

\subsection{Non-Markov but zero MI in canonical purification}
\label{sec:counteregTFD}
Consider the 4--qubit state
\begin{equation}
    \rho
    = \frac{1}{16}\!\left(
        \mathbf{1}
        + \frac12 (Z_1 Z_2 + Z_3 Z_4)
    \right).
\end{equation}
We regard it as a tripartite system, where $A=\{1\}, B=\{2,3\}, C=\{4\}$.

The state is classical so all von-Neumann entropies are Shannon entropies.
Explicit calculation shows that:
\begin{equation}
\begin{aligned}
S(1,2,3,4) &= 4\times (-\frac 18\log_2(\frac 18))
+8\times (-\frac{1}{16}\log_2(\frac{1}{16}))=  \frac72,\\
S(2,3) &= 2,\\
S(1,2,3) &= S(2,3,4) = 1+ 2\times (-\frac38\log_2(\frac38)-\frac18\log_2(\frac18))
         = 4 - \frac{3}{4}\log_2 3.
\end{aligned}
\end{equation}
Thus the conditional mutual information is
\begin{equation}
\begin{aligned}
\condinf{1}{4}{2,3}
= H(1,2,3) + H(2,3,4)
   - H(2,3) - H(1,2,3,4)
= \frac52 - \frac{3}{2}\log_2 3\neq 0.
\end{aligned}
\end{equation}

To compute its canonical purification, we note that
\begin{align}
    \sqrt \rho = \frac{1}{16}\left[(2+\sqrt 2)\mathbf 1
    + \sqrt 2 (Z_1Z_2+Z_3Z_4)+ (\sqrt 2 -2)Z_1Z_2Z_3Z_4\right].
\end{align}
Therefore, in the Bell basis, 
\begin{align}
    \ket {\sqrt\rho}= \frac{1}{4}\left[(2+\sqrt 2)\ket{\phi_+\phi_+\phi_+\phi_+}
    + \sqrt 2 (\ket{\phi_-\phi_-\phi_+\phi_+}+\ket{
    \phi_+\phi_+\phi_-\phi_-
    })+ (\sqrt 2 -2)\ket{\phi_-\phi_-\phi_-\phi_-}\right].
\end{align}
Explicit calculation shows that
\begin{align*}
    \rho_{A\bar A} = \rho_{1 \bar1}
    & =  \frac{1}{4}\left[(2+\sqrt2)\ketbra{\phi_+}{\phi_+} 
    +(2-\sqrt2)\ketbra{\phi_-}{\phi_-} 
    \right].
\end{align*}
By symmetry, $\rho_{C\bar C}$ has a similar form.
The joint state $\rho_{A\bar AC\bar C}$ is
\begin{align*}
    \rho_{A\bar AC\bar C} =  \frac{1}{16}\left[(2+\sqrt2)^2\ketbra{\phi_+\phi_+}{\phi_+\phi_+} 
    +(2-\sqrt2)^2\ketbra{\phi_-\phi_-}{\phi_-\phi_-}
    +2(\ketbra{\phi_-\phi_+}{\phi_-\phi_+}+\ketbra{\phi_+\phi_-}{\phi_+\phi_-})
     \right]
\end{align*}
It is then obvious that $\rho_{A\bar AC\bar C} = \rho_{A\bar A} \otimes \rho_{C\bar C}$ hence $I(A\bar A:C\bar C)=0$.

\end{document}